\documentclass[a4paper,11pt]{article}
\usepackage[utf8]{inputenc}
\usepackage{amsmath,amsfonts}
\usepackage{amsthm} 
\usepackage{amssymb}
\usepackage{multirow}
\usepackage{latexsym}
\usepackage{color}
\usepackage{graphicx} 
\usepackage{subcaption}
\usepackage[round]{natbib}
\usepackage[colorlinks=true,linkcolor=blue,urlcolor=blue,citecolor=blue]{hyperref}
\usepackage[parfill]{parskip}
\usepackage{bbm}

\newtheorem{lemma}{Lemma}[section]

\newtheorem{theorem}{Theorem}[section]
 
\newtheorem{proposition}[theorem]{Proposition}

\newtheorem{remark}[theorem]{Remark}

\newcommand{\argmax}[1]{\underset{#1}{\operatorname{argmax}}}
\allowdisplaybreaks

\newcommand{\IR}{\mathbb{R}}
\newcommand{\IE}{\mathbb{E}}
\newcommand{\IP}{\mathbb{P}}
\newcommand{\Ind}{\mathbbm{1}}
\newcommand{\IN}{\mathbb{N}}
\newcommand{\IZ}{\mathbb{Z}}
\newcommand{\IT}{\mathbbm{T}}

\renewcommand{\tilde}{\widetilde}
\renewcommand{\bar}{\overline}
\renewcommand{\epsilon}{\varepsilon}
\renewcommand{\phi}{\varphi}

 \numberwithin{equation}{section}
\begin{document}
\title{\bf Testing For Global Covariate Effects in Dynamic Interaction Event Networks}
\author{Alexander Kreiss \footnote{Computations for this work were done (in part) using resources of the Leipzig University Computing Centre.} \\
    Institute of Mathematics, Leipzig University \\
    and \\
    Enno Mammen \\
    Institute for Applied Mathematics, Heidelberg University \\
    and \\
    Wolfgang Polonik \footnote{W. Polonik acknowledges partial support by the National Science Foundation under Grant No. DMS-1713108.} \\
    Department of Statistics, University of California at Davis}
\date{\today}
\maketitle
\begin{abstract}
In statistical network analysis it is common to observe so called interaction data. Such data is characterized by actors forming the vertices and interacting along edges of the network, where edges are randomly formed and dissolved over the observation horizon. In addition covariates are observed and the goal is to model the impact of the covariates on the interactions. We distinguish two types of covariates:  global, system-wide covariates (i.e. covariates taking the same value for all individuals, such as seasonality) and local, dyadic covariates modeling interactions between two individuals in the network. Existing continuous time network models are extended to allow for comparing a completely parametric model and a model that is parametric only in the local covariates but has a global non-parametric time component. This allows, for instance, to test whether global time dynamics can be explained by simple global covariates like weather, seasonality etc. The procedure is applied to a bike-sharing network by using weather and weekdays as global covariates and distances between the bike stations as local covariates.
\end{abstract}

\section{Introduction}
\label{sec:Int}

One branch in statistical network analysis is concerned with the analysis of so called interaction data. Some references for this type of data are for example \citet{B08,PW13,MRV18,KMP19}. Such data typically consists of a dynamic, random network where the vertices are some sort of actors who can interact with each other along the edges of the network. By a \emph{dynamic, random network} we mean a network which has a fixed set of vertices, i.e. actors, but the edge set, i.e. their relations, may change randomly over time. A classical example would be social contact networks: Here the vertices of the network under consideration are people and two persons are connected by an edge if they have the potential to interact with each other (e.g. being in a 50m radius of one another). An interaction, or an event, between two people could be an instance of close contact like the start of a face-to-face conversation. Over the course of the day the potential interaction partners change because people tend to be in different locations during the day (e.g. at work or at home). In addition to the network and the interactions, one typically also observes a set of covariates. Such covariates describe the relation of each pair in the network and can be composed of covariates on actor-level, pair specific covariates or system-wide covariates. The interest lies in modeling the relation between the interactions and the covariates.

We emphasize already here that the edge set does not necessarily need to describe physical (or otherwise established) relations like in the previous example. It is also possible that one imposes a network on the data, e.g. by specifying which pairs are relevant for the question of interest. In this case, the edges of the network can be understood as model inclusion dummies, indicating at which point in time which pair is relevant for the study. Such scenarios might be relevant when the interest lies on interactions that are only of interest if other conditions are met, such as the use of social media while using the smart-phone.

We model interaction event data on networks by formulating a counting process model in which the intensity function depends on covariates. In our framework, the covariate process is not required to have \emph{short-memory} properties. This flexibility allows for models with complex dependence of the covariates on past events. Such models have been studied in both parametric and non-parametric settings, see, for instance, \cite{PW13} and \cite{KMP19}, respectively. In order to illustrate the contribution of this paper, we let $X_{n,ij}(t)\in\IR^p$ be local, pair-specific covariate functions. Consider the following model for the intensity $\lambda_{n,ij}$ of the counting process which counts events from $i$ to $j$
$$\lambda_{n,ij}(t):=\alpha_0(t)\exp(\beta_0(t)'X_{n,ij}(t)).$$
The most flexible model (\emph{fully non-parametric}) allows both $\alpha_0(t)$ and $\beta_0(t)$ to vary in time. This is studied in \citet{KMP19}. \citet{PW13} assume $\beta_0\equiv\textrm{const.}$ and treat $\alpha_0$ as nuisance parameter. The case of both $\beta_0\equiv\textrm{const.}$ and $\alpha_0\equiv\textrm{const.}$, which we call a \emph{completely parametric} model, was used in \citet{K20} to study how to test the completely parametric versus the (fully) non-parametric model by using the $L_2$-distance between a parametric and a non-parametric estimator of the intensity function as a test statistic, similar to what is done in \citet{HM93} in nonparametric regression. None of the papers \citet{PW13}, \citet{KMP19}, \citet{K20} allows that any entry of $X_{n,ij}(t)$ is the same for all pairs $(i,j)$, thereby rendering the inclusion of a global, system-wide covariate $Z(t)$ impossible. Our work is addressing this issue by including $Z(t)$ in the baseline by assuming that $\alpha_0(t)=\alpha(Z(t);\theta_0)$ for some unknown parameter $\theta_0$ and a known link function $\alpha$ (note that one entry of $Z(t)$ is allowed to be $t$ and hence $\alpha(Z(t);\theta_0)$ might explicitly depend on time). Our goal then is to test under the additional assumption that $\beta_0 \equiv \textrm{const.}$ the hypothesis $\alpha_0(t)=\alpha(Z(t);\theta_0)$ for some $\theta_0$ versus the alternative of a non-parametric specification of $\alpha_0.$ 

Being unable to properly include global covariates in the testing is restrictive because many data-sets show a clear seasonality and therefore the superiority of the fully non-parametric model might simply stem from the situation that the completely parametric model under consideration does not accurately account for this seasonality. Our new test allows to explore whether extending a completely parametric model by allowing for a non-parametric seasonality is producing a meaningful extension of the model. Hence, when discussing the question whether there is a need to use a non-parametric model or whether a parametric model is sufficient, it is natural, to consider intermediate steps between completely parametric and fully non-parametric. We provide the first step for such comparisons. In practice such tests are relevant because a parametric model allows predictions in situations when the global covariates change (like the weather in the next month) while the non-parametric estimate for $\alpha_0$ can only be transferred to other time periods if one assumes that the global covariates remain the same. Lastly, if one has a specific hypothesis about what causes the seasonality, our testing framework provides a methodology to test for this hypothesis.

The complex dependence structure in network settings makes the mathematics behind such types of analysis significantly more challenging. While in standard situations individuals are typically considered to behave independently, such an assumption is rarely plausible in a network set-up. Most of the time it is rather the case that neighboring individuals influence each other and should therefore not be treated as independent. On the other hand it is intuitively clear that this dependence diminishes exponentially as the distance of actors grows in the network: If one is equally influenced by $k$ friends and they, in turn, are equally influenced by their $k$ friends, the impact of the friends of the friends (i.e. of actors of distance $2$), is of the order $k^{-2}$. This assumes that no single actor has a  disproportionate influence on others (no hubs). Intuitively, if only a single actor influences things, then observing new actors will not have a major impact. This intuition has been made precise in \citet{K20} and we make use of these ideas.

The organization of this paper is as follows: In Section \ref{sec:model} we introduce the exact model and present the hypothesis and our suggested test statistic. Afterwards in Section \ref{sec:theory} we present the theory for our test statistic. We will discuss the practical implementation in Section \ref{sec:numerics} and provide a real-world data example as well as a synthetic simulation study. Section \ref{sec:conclusion} concludes. Proofs and other technical details are collected in the supplementary Sections \ref{app:mixing}-\ref{supp:main_proof}.

\section{Model Specification and Test Strategy}
\label{sec:model}
In the following we introduce the exact data generating process used in the following, and formulate the testing problem of interest. Suppose that we observe  a sequence of networks $G_{n,t}=(V_n,E_{n,t})$ for $n\in\IN$ over a time interval $t \in [0,T]$, where the deterministic population $V_n$ is of growing size $|V_n|=n$. The edge set $E_{n,t}\subseteq V_n\times V_n$ is random and time varying. To simplify the notation we identify $V_n$ with $[n] =\{1,\ldots,n\}$. For $i,j\in V_n,$ the functions $C_{n,ij}(t)=\Ind((i,j)\in E_{n,t})$ indicate whether $i$ and $j$ are connected at time $t\in[0,T]$ so that $\left(C_{n,ij}\right)_{i,j\in V_n}$ can also be understood as the random, time varying adjacency matrix. Within the population, individuals can interact with each other if they are connected by an edge. For $i,j\in V_n,$ we denote by $N_{n,ij}(t)$ the number of interactions between $i$ and $j$ up to and including time $t\in[0,T]$ . The processes $N_{n,ij}:[0,T]\to\IN$ are thus counting processes. While networks can be directed or undirected, we only consider undirected networks for simplicity. Thus, we assume throughout $C_{n,ij}=C_{n,ji}$ as well as $N_{n,ij}=N_{n,ji}$ for all $i,j\in V_n$. Moreover, also for simplicity, we exclude self-interactions, meaning that we set $C_{n,ii}\equiv0$ and $N_{n,ii}\equiv0$ for all $i\in V_n$. We suppose that the network process (via $C_{n,ij}$) and the interactions (via $N_{n,ij}$) are observed. In addition to these we also observe random covariates $X_{n,ij}:[0,T]\to\IR^p$ which are specific for the pair $i,j\in V_n$. We are interested in modelling the interactions. We do not model the network process $C_{n,ij}$ and also not the covariate processes $X_{n,ij}$. As outlined in \cite{KMP19}, one could use our framework for the following more specific model where one observes two interaction processes $N^-_{n,ij}$ and $N^+_{n,ij}$ that define one network process: $N^+_{n,ij}$ jumps if an edge between $i$ and $j$ is added and $N^-_{n,ij}$ jumps if an edge between $i$ and $j$ is removed. Here one would define the network processes belonging to one of the two interaction processes as $C_{n,ij}$ or $1-C_{n,ij}$, respectively. In this paper we will not pursue this setting.

Note lastly that one can also use $C_{n,ij}$ as filters for the researcher to select pairs of interest. While there could potentially be interactions between all vertices, there might be reasons to assume that only specific interactions are relevant for the model \eqref{eq:mod} at a given time $t$. In this case, the researcher has the flexibility to achieve this by setting all the corresponding $C_{n,ij}(t)$ equal to 1. As was discussed in previous work, the researcher does not have to select the relevant pairs perfectly as long as the selection is not \emph{too liberal} (cf. p. 2769 in \citet{KMP19}).

Throughout we will assume that the array $(N_{n,ij})_{i,j\in V_n}$ forms a multivariate counting process with respect to a filtration $(\mathcal{F}_t^{\,n})_{t\in[0,T]}$. Unless specified otherwise all counting processes and martingales are understood to be defined with respect to this filtration. Note that by definition of a multivariate counting process no two counting processes jump at the same time (with probability one). As discussed above, the covariates and the interactions are connected through the intensity functions $\lambda_{n,ij}:[0,T]\to[0,\infty)$ for which we assume the proportional hazards model (cf. \citet{ABGK93,MS06,C72,AG82}), that is, we suppose that the intensity function with respect to the filtration $(\mathcal{F}_t^{\,n})_{t\in[0,T]}$ is given by
\begin{equation}
\label{eq:mod}
\lambda_{n,ij}(t)=C_{n,ij}(t)\alpha_0(t)\Psi(X_{n,ij}(t);\beta_0),
\end{equation}
where $X_{n,ij}(t) \in {\mathbb R}^p$ are random covariates depending on time, $\beta_0\in\IR^q$ is an unknown parameter and $\alpha_0:[0,T]\to[0,\infty)$ is an unknown, deterministic baseline intensity. The link function $\Psi$ is supposed to be known to the researcher, e.g. in a Cox-type model $p=q$ and $\Psi(X_{n,ij}(t);\beta_0)=\exp(\beta_0'X_{n,ij}(t))$ (cf. \citet{SM04}). In order to assure identifiability, $\Psi$ (or $X_{n,ij}$) may not include an intercept, and we have to impose, e.g. $\Psi(0;\beta_0)=1$. While $\Psi(X_{n,ij}(t);\beta_0)$ describes the pair-specific part of the intensity, $\alpha_0(t)$ can be interpreted as global component of the intensity which applies to all pairs in the system. Our interest then lies in testing whether the baseline $\alpha_0$ can be adequately modeled by deterministic system wide covariates, i.e., covariates that are the same for all individuals, such as weather or economic development. We denote these covariates by $Z:[0,T]\to\IR^d$.  By `deterministic' we here mean `measurable with respect to $\mathcal{F}_0^{\,n}$', where the measurability assumption on $Z$ is made for simplicity. Without it, the asymptotic analysis would become significanly more complex. Intuitively the assumption is justified if reliable short-time predictions of future developments of $Z$ exist, e.g., weather forecasts. Since the covariates $Z$ are supposed to be deterministic and to be the same for the entire network (regardless of its size), we assume also that they do not change with $n$. Our aim is testing the hypothesis
$$H_0:\alpha_0(t)=\alpha(Z(t);\theta_0)\textrm{ for some }\theta_0\in\Theta,$$
where $\Theta\subseteq\IR^d$ is a suitable parameter space and $\alpha:\IR^d\times\Theta \to [0,\infty)$ is a known link function. To simplify notation, we let $\alpha(\theta,t):=\alpha(Z(t);\theta)$. The test statistic we use for testing this hypothesis is along the lines of \cite{HM93}, that is, we compare a parametric and a non-parametric estimator. In order to define those estimators, we need the following definitions
\begin{alignat*}{2}
&N_n(t):=\sum_{i,j\in V_n}N_{n,ij}(t),\qquad &&\overline{\Psi}_n(t;\beta):=\sum_{i,j\in V_n}C_{n,ij}(t)\Psi(X_{n,ij}(t);\beta) \\
&C_n(t):=\Ind\left(\exists\, i,j\in V_n: C_{n,ji}(t)=1\right),\qquad&&\lambda_n(t,\beta):=\alpha_0(t)\overline{\Psi}_n(t,\beta),
\end{alignat*}
where the above sums over $i,j\in V_n$ are understood as \emph{sum over all undirected pairs $(i,j)\in V_n\times V_n$ with $i\neq j$}. We will use this notation throughout the paper. The process $N_n$ is again a counting process with respect to $(\mathcal{F}_t^{\,n})_{t\in[0,T]}$ because in our model with probability one no two individual processes jump at the same time. The intensity function of $N_n$ is given by $\lambda_n(t,\beta_0)$. The function $C_n$ is an indicator that equals $1$ at time $t$ if there is at least one edge present in the network. We denote by $M_{n,ij}(t):=N_{n,ij}(t)-\int_0^t\lambda_{n,ij}(s)ds$ the martingale associated with the counting process $N_{n,ij}$. Then, $M_n:=\sum_{i,j\in V_n}M_{n,ij}$ is the martingale associated with $N_n$ (all with respect to $(\mathcal{F}_t^{\,n})_{t\in[0,T]}$). 

The estimators used for our test statistic are as follows. Our parametric estimator is the maximum likelihood estimator
\begin{align}
\left(\hat{\theta}_n,\hat{\beta}_n\right):=\argmax{(\theta,\beta)}\sum_{i,j\in V_n}\int_0^T\log\left(\alpha(\theta,t)\Psi(X_{n,ij}(t);\beta)\right)dN_{n,ij}(t)-\int_0^T\alpha(\theta,t)\overline{\Psi}_n(t;\beta)dt. \label{eq:ML}
\end{align}
For the non-parametric estimator, we first use an initial estimator for $\beta_0$ not depending on any specific form of $\alpha_0$. For this we take the partial-maximum likelihood estimator (cf. \citet{C75,PW13}).
\begin{equation}
\label{eq:PL}
\tilde{\beta}_n:=\argmax{\beta}\sum_{i,j\in V_n}\int_0^T\left[\log\Psi\left(X_{n,ij}(t);\beta\right)-\log\overline{\Psi}_n(t;\beta)\right] dN_{n,ij}(t).
\end{equation}
Our non-parametric estimator is the smoothed Nelson-Aalen estimator (cf. \citet{ABGK93,A78,N69,RH83})
\begin{align}
\label{eq:nelson-aalen}
\hat{\alpha}_n\left(t,\tilde{\beta}_n\right)&:=\int_0^TK_{h,t}(s)\frac{C_n(s)}{\overline{\Psi}_n\left(s;\tilde{\beta}_n\right)}dN_n(s),
\end{align}
where $K \ge 0$ is a kernel function and $K_{h,t}(s):=\frac{1}{h}K\left(\frac{s-t}{h}\right)$. Above we use the convention $0/0:=0$. Let further 
\begin{equation*}
\alpha_{\textrm{smooth}}(\theta,t):=\int_0^TK_{h,t}(s)\alpha(\theta,s)ds\quad\textrm{ for all }\quad\theta\in\Theta,
\end{equation*}
be a smoothed versions of the parametric estimator. \citet{HM93} argue that, when comparing non-parametric and parametric estimators, the nature of the test might be dominated by the bias of the nonparametric estimator, this can be avoided by smoothing the parametric estimator. Our test statistic for testing is hence
\begin{equation}
\label{eq:Tn}
T_n:=\int_0^T\left(\hat{\alpha}_n\left(t;\tilde{\beta}_n\right)-\alpha_{\textrm{smooth}}(\hat{\theta}_n,t)\right)^2w(t)dt,
\end{equation}
where $w:[0,T]\to[0,\infty)$ is a weight function with $\IT:=\textrm{supp}\, (w)\subset(0,T)$, where $\textrm{supp}\, (w)$ denotes the closed support of $w$. Thus $w$ cuts off the boundary and therefore we may ignore possible boundary issues of the kernel type Nelson-Aalen estimator. We consider the boundary cut-off to be the main role of $w$, thus we will later in the simulation choose $w(t):=\Ind(t\in[\delta,T-\delta])$ for some small $\delta>0$.

\section{Main Results}
\label{sec:theory}
Here we present the main theoretical result of the paper, which states the asymptotic behavior of $T_n$ under the null-hypothesis and under local alternatives. The assumptions needed for this result to hold are presented and discussed in detail in Section \ref{subsec:assumptions}.

\subsection{Main Result}
\label{subsec:result}

In order to state our main result, we introduce the following notation, where by Assumption (VX) below, all these quantities are well defined:
\begin{align}
&\mu_n(t;\beta):= \IE\left(\Psi(X_{n,ij}(t);\beta)\big|C_{n,ij}(t)=1\right), \label{eq:defXn} \\
&p_n(t): = p_n \pi(t) = \IP(C_{n,ij}(t)=1) \label{eq:defpn} \\
&N:= m\,p_n\;\Big(\int_0^T\int_0^TK_{h,t}(s)\frac{w(t)}{ \pi(s)}ds\,dt\Big)^{-1} , \label{eq:defan}
\end{align}
where $m := \binom{n}{2}$ denotes the total number of undirected pairs and $p_n$ and $\pi$ satisfy: 

{\bf Assumption (SP)}: (Sparsity) \emph{We assume $p_n>0$, $\pi(t)\geq1$, and $\pi$ continuous with $0 \le p_n\pi(t) \le 1$ for all $n \ge 1$ and $t \in [0.T]$}

 Note that assuming $\pi(t)\geq1$ is no restriction when $0 < \inf_{t \in [0,T]} \pi(t),$ because in that case we may simply rescale $p_n$. Observe that the expected number of undirected pairs at time $t$ equals $m p_n \pi(t),$ and that we explicitly allow that $p_n\to0$ so that our set-up includes sparse networks in which the expected number of edges $mp_n(t)=O(n)$. Finally, $N$ can be interpreted as a weighted time average of these expected numbers, see also Remark \ref{rem:timeaver}. Also note that $mp_n(t) \mu_n(t;\beta) = \IE \overline{\Psi}_n(t,\beta).$

\begin{theorem}
\label{thm:T1}
Suppose that Assumption (SP) and all the assumptions from Section \ref{subsec:assumptions} hold. Further, assume that model \eqref{eq:mod} holds with $\alpha_0(t)=\alpha(\theta_0,t)+c_n\Delta_n(t),$ where $c_n=\big(N\sqrt{h}\big)^{-1/2},$ and $\Delta_n$ is uniformly bounded and continuously differentiable with uniformly bounded derivative (that is, uniformly in both $t$ and $n$). Then, as $n \to \infty,$ 
$$\frac{N\sqrt{h}\left(\displaystyle{T_n-\frac{A_n}{Nh}-\int_0^T\left(\int_0^TK_{h,t}(s)c_n\Delta_n(s)ds\right)^2w(t)dt}\right)}{\sqrt{B_n}}\;\; \overset{d}{\to}\;\;\mathcal{N}(0,1),$$
where with  $f_n(r,s):=\int_0^ThK_{h,t}(s)K_{h,t}(r)w(t)dt,$  $\gamma=\int_0^T\int_0^TK_{h,t}(s)\frac{w(t)}{\pi(s)}dsdt,$ and $K^{(2)}:=\frac{1}{2} \|K\star K\|_2^2,$ 
\begin{align*}
&A_n:=\frac{1}{N\gamma^2}\sum_{i,j\in V_n}\int_0^Tf_n(r,r)\left(\frac{1}{ \pi(r)\mu_n(r,\beta_0)}\right)^2dN_{n,ij}(r), \\
&B_n:=4K^{(2)}\int_0^T\left(\frac{w(s)\alpha_0(s)}{\gamma\pi(s)\mu_n(s;\beta_0)}\right)^2ds.
\end{align*}
\end{theorem}
\begin{remark} \label{rem:timeaver}
We assume in Assumption (B) below that $r\mapsto\mu_n(r,\beta_0)$ is uniformly bounded from below. Furthermore $\gamma\to\int_0^Tw(t)/\pi(t)dt>0$ if $\pi$ is continuous. Therefore, $B_n$ is bounded from below and does not influence the rate of convergence. To illustrate the rate of convergence, we consider the case of $\pi$ constant. Note that in this case $\pi\equiv1$ and $p_n=p_n(t)$. Therefore, for $h>0$ small enough,
$$N=mp_n(t)\left(\int_0^Tw(t)dt\right)^{-1},$$
and $N$ is a scaled version of the expected number of edges.
\end{remark}
\begin{remark}
\label{rem:conf_int}
Note that Theorem \ref{thm:T1} is formulated for local alternatives, including the null hypothesis for $\Delta_n\equiv0$. This allows us to use the stated asymptotic normality to formulate confidence intervals for the $L^2$-norm of the smoothed and weighted $\Delta_n$ as follows: Let $q_{1-\alpha}$ denote the $(1-\alpha)$-quantile of a standard normal distribution. Then, asymptotically with probability at least $1-\alpha$,
$$\left\|K_h\star(c_n\Delta_n)\right\|_w\geq D_n:=\sqrt{\max\left(T_n-\frac{B_n^{\frac{1}{2}}q_{1-\alpha}}{N\sqrt{h}}-\frac{A_n}{Nh},\ 0\right)},$$
where $\|K_h\star(c_n\Delta_n)\|_w^2:=\int_0^T\left(\int_0^TK_{h,t}(s)c_n\Delta_n(s)ds\right)^2w(t)dt$. In other words, we can also test the hypotheses $H_{\epsilon}$: $\|K_h\star(c_n\Delta_n)\|_w\leq \epsilon$ for different values of $\epsilon$ (reject $H_{\epsilon}$ if $D_n>\epsilon$). The right-hand side, i.e. $D_n$, in the above displayed formula is the largest value of $\epsilon$ for which $H_{\epsilon}$ would be rejected. We can hence use $D_n$  as a lower bound for the unknown $\|K_h\star(c_n\Delta_n)\|_w$ (the deviation from the parametric model) and compare it to $\|K_h\star\hat{\alpha}_n(\cdot;\tilde{\beta}_n)\|_w$ to obtain the fraction of the baseline that deviates from the parametric model. Finally, note that all of this discussion depends on the choice of $w$, which can be used to include or exclude areas of interest. We will illustrate this in our data application in Section \ref{sec:data}.
\end{remark}

In other situations it has been pointed out that the convergence to a normal limit for $L^2$ type non-parametric tests is very slow. As a remedy, resampling techniques have been proposed for finding critical values. We are not aware of a resampling scheme that is suitable for our model of dependent counting processes and let this question open for future research. However, our simulations suggest that in this case the approximation is not too bad.

Our analysis of $T_n$ reveals 
that $T_n$ is asymptotically equivalent to a degenerate $U$-statistics with kernel depending on the bandwidth $h$. In this respect the statistic behaves like  test statistics in other models that are also based on $L^2$ distances between parametric and nonparametric curve estimators. Degenerate $U$-statistics (of order 1) can be represented as weighted sums of centered squares of linear statistics of the form $\sum_{k=1}^K \lambda_k(\zeta_k^2 - c_k)$, where the linear statistics $\zeta_k$ depend on the eigenfunctions of the integral operator whose kernel equals the kernel of the $U$-statistic, and the weights $\lambda_k$ are its eigenvalues. For tests of the form $T_n$ and other $L^2$-type tests one has to distinguish two asymptotic settings: the case of a fixed bandwidth and the case of a bandwidth converging to zero. In case of a fixed bandwidth, the representation of degenerate $U$-statistics can be used to show that the limit distribution of $U$-statistics is a weighted sum of independent, centered squared standard normal variables, see e.g.  \cite{LM013} and \cite{G77}. 
This is similar to the 
Cramer-von Mises test and other related goodness-of-fit tests. The linear statistics $\zeta_k$ can be interpreted as tests of one-dimensional deviations from the null hypothesis in one direction. Summed up the $\zeta_k^2$ result in an omnibus test. Also, for the bandwidth $h$ converging to zero, we obtain an approximation of squared linear statistics by weighted averages, with weights depending on the bandwidth $h$. For a Gaussian convolution kernel, eigenvalues and eigenfunctions corresponding to the directions are considered in \cite{SWY15}.  There, the dependence of the eigenvalues and eigenfunctions on the variance/squared bandwidth of the Gaussian kernel is explicitly stated. With decreasing bandwidth the mass of the weights moves to later summands, and this leads to an asymptotic normal distribution of these $L^2$ statistics.  Furthermore, it shows that the statistic puts more and more weight into later summands which increases the omnibus character of the test. For a comparison of $L^2$-type tests with decreasing bandwidth and with fixed bandwidth in nonparametric regression, see  \cite{FL00}. There it is shown that these tests  cover a class of tests with a large spectrum of properties. This was the reason why we have chosen this class of tests in our mathematical analysis. Furthermore, in  \cite{FL00}  it is also shown that asymptotics with decreasing bandwidth and with fixed bandwidth require different mathematical approaches. In this paper we only consider the case of statistics with bandwidths converging to zero. In our proof, we will also not explicitly use the representation of $U$-statistics given above, but we will use martingale  central limit theorems instead. Nevertheless, the discussion above is important for the basic understanding of the types of test considered in our work.

\subsection{Assumptions}
\label{subsec:assumptions}
In this section we collect the assumptions that we use to prove our main result. We begin with fairly standard assumptions on the data.

\textbf{Assumption (VX):} (Vertex Exchangeability) \\
\emph{The tupels $(C_{n,ij},X_{n,ij},N_{n,ij})$ are vertex exchangeable, i.e., their joint distribution does not change when the vertex labels are permuted.}

Since in observational studies names are normally uninformative, we regard this assumption as not very restrictive. (VX) implies that all quantities indexed by $(i,j)\in V_n\times V_n$ are identically distributed and hence that $\mu_n(t;\beta)$ and $p_n \pi(t)$ do not depend on $i,j$.

\textbf{Assumption (KBW):} (Kernel, Bandwidth, Weight) \\
\emph{The kernel $K$ is bounded, symmetric about $0$ and supported on $[-1,1]$. The weight function $w$ is continuous and bounded with $\IT:=\textrm{supp}\, w\subset(0,T)$. The bandwidth $h$ fulfills $\frac{\log m}{h N}\to0$, $\sqrt{h}\log m\to0$.}

These assumptions in particular imply that $h\to0$ and $Nh\to\infty$, which, when interpreting $N$ as the effective sample size, are standard assumptions for kernel smoothing. Choosing  $h$ of the form $O(N^{-r})$ for $r>0$ `small', with the standard choice $r=1/5$, seems reasonable. Also, the quantity $m = \binom{n}{2}$ can of course be replaced by $n^2$ in these assumptions. Using $m$ makes the origins of the stated assumptions more transparent.

In the following, let $\|f(\cdot)\|_{\infty}:=\sup_{x\in\mathcal{D}_f}\|f(x)\|$ denote the sup-norm, where $f:\mathcal{D}_f\to\mathcal{X}$ and $\|\cdot\|$ is a norm on $\mathcal{X}$.

\textbf{Assumption (C):} (Continuity and Boundedness of the Model) \\
\emph{The link functions $\Psi$ and $\alpha$ are bounded, i.e. $\|\Psi(\cdot;\beta_0)\|_{\infty}<\infty$ and $\|\alpha(\theta_0;\cdot)\|_{\infty}<\infty$, and fulfill the following Lipschitz property
$$\left|\Psi(x,\beta_1)-\Psi(x;\beta_2)\right|\leq L_{\Psi}\|\beta_1-\beta_2\|\,\textrm{ and }\,|\alpha(\theta_1,t)-\alpha(\theta_2,t)|\leq L_{\alpha}(t)\|\theta_1-\theta_2\|.$$
The function $L_{\alpha}$ satisfies $\|L_{\alpha}\|_2<\infty$. Moreover, $\mu_n(\cdot;\beta_0)$ is continuous.} 

Intuitively, Assumption (C) means that the model is not changing too rapidly neither over time (continuity of $t\mapsto\mu_n(t;\beta_0)$) nor for different parameters (Lipschitz continuity of $\Psi$ and $\alpha$). The latter is plausible, for instance, if we assume continuously differentiable link functions (with respect to the parameters) and bounded covariates.

\textbf{Assumption (P)} (Parametric Estimation) \\
\emph{The estimators $\hat{\theta}_n$ and $\tilde{\beta}_n$ are based on data independent of $(C_{n,ij}, X_{n,ij}, N_{n,ij})_{(i,j) \in V_n \times V_n}$ and satisfy $\|\hat{\theta}_n-\theta_0\|=O_P(N^{-1/2}),$ $\|\tilde{\beta}_n-\beta_0\|=O_P(N^{-1/2})$ and $\IE(\|\tilde{\beta}_n-\beta_0\|^2)=O(N^{-1})$. Moreover, there is a compact set $K(\beta_0)$ such that $\IP(\tilde{\beta}_n\in K(\beta_0))=1$.}

Since $N$ is the effective number of observations, the assumptions on the rates of consistency of the estimators are standard. The independence assumption is made for simplicity. It holds, for instance, if we use data splitting. We conjecture that the independence assumption can be replaced by a stochastic expansion of the estimators. This has been done in other discussions of test statistics based on the comparison of parametric and nonparametric fits, see e.g. Assumption (P1) in \citet{HM93}.

\textbf{Assumption (B):} (Local Boundedness) \\
{\em $\sup_{n\in\IN}\|\mu_n(\cdot;\beta_0)\|_{\infty}<\infty$ and $\sup_{n\in\IN}\sup_{\beta\in K(\beta_0)}\|\mu_n(\cdot;\beta)^{-1}\|_{\infty}=O(1)$.}

The upper bound holds, for example, if $\Psi$ is bounded, and the lower bound (an upper bound on the inverse) essentially only means that at no point in time the intensity converges to zero (provided that a link exists). Since we condition on the existence of a link in the definition of $\mu_n$ it might even be plausible to assume that $\mu_n$ does not even depend on $n$ as was argued in \citet{KMP19}.

For the following assumptions we let $B_n(c):=\left\{\beta: \|\beta-\beta_0\|\leq \frac{c}{\sqrt{N}}\right\}$ for an arbitrary constant $c>0$ and, for any $c_1,c_2>0,$ define the events
\begin{align*}
\mathcal{A}_n(c_1,c_2):=&\left\{\underset{\beta\in B_n(c_2)}{\sup_{t\in[0,T],}}\sqrt{p_n(t)}\left|\left(\frac{1}{mp_n(t)}\overline{\Psi}_n(t;\beta)\right)^{-1}-\mu_n(t;\beta)^{-1}\right|\leq c_1\sqrt{\frac{\log m}{m}}\right\} \\
\tilde{\mathcal{A}}_n(c_1,c_2):=&\left\{\underset{\beta\in B_n(c_2)}{\sup_{t\in[0,T],}}\sqrt{p_n(t)}\left|\frac{1}{m p_n(t)}\overline{\Psi}_n(t;\beta)-\mu_n(t;\beta)\right|\leq c_1\sqrt{\frac{\log m}{m}}\right\}.
\end{align*}
Recall that $\mu_n(t;\beta)=\IE\left(\Psi(X_{n,ij}(t),\beta)\big|C_{n,ij}(t)=1\right) = \IE \big[\frac{1}{mp_n(t)} \overline{\Psi}_n(t;\beta)\big]$ and $m = \binom{n}{2}$.

\textbf{Assumption (LL):} (Law of Large Numbers) \\
\emph{For any $c_2>0$ and any $\delta>0,$ there is $c_1>0$ such that, for all $n\in\IN$
$$\IP\left(\mathcal{A}_n(c_1,c_2)\right))\geq1-\delta\,\textrm{ and }\,\IP\left(\tilde{\mathcal{A}}_n(c_1,c_2)\right)\geq1-\delta.$$}

(LL) is essentially a law of large numbers for the hazards $\Psi(X_{n,ij}(t),\beta)$ in the sense that we require that their averages concentrate around their mean. As discussed in the Supplement (Section \ref{app:mixing}) this assumption holds under certain mixing conditions.

Denote $\tilde{C}_{n,ij}:=\sup_{t\in[0,T]}C_{n,ij}(t)$.

\textbf{Assumption (WC):} (Weak Correlation) \\
\emph{Suppose that
\begin{align}
&\frac{\IP(\tilde{C}_{n,12}=1)}{p_n}+\frac{\IP(\tilde{C}_{n,12}=1,\tilde{C}_{n,23}=1)}{np_n^2}+\frac{\IP(\tilde{C}_{n,12}=1,\tilde{C}_{n,34}=1)}{p_n^2}=O(1), \label{eq:swc1} \\
&\frac{1}{N^4}\underset{a=1,...,4}{\sum_{i_a,j_a\in V_n}}\IP\left(\prod_{a=1}^4\tilde{C}_{n,i_aj_a}=1\right)=O(1), \label{eq:swc2} \\
&\underset{s\in[0,T]}{\sup_{i,j\in V_n,}}\frac{1}{N^4}\underset{a=1,...,4}{\sum_{k_a,l_a,\in V_n}}\IP\left(\prod_{a=1}^4\tilde{C}_{n,k_al_a}=1\Big|C_{n,ij}(s)=1\right)=O(1), \label{eq:swc3} \\
&\sup_{k,l,k',l'}\frac{1}{N^4}\underset{a=1,...,4}{\sum_{i_a,j_a\in V_n}}\IP\left(\prod_{a=1}^4\tilde{C}_{n,i_aj_a}=1\bigg|\tilde{C}_{n,kl}=1,\tilde{C}_{n,k'l'}=1\right)=O(1), \label{eq:swc4} \\
&\underset{k,l,l'l'\in V_n}{\sup_{t,s\in[0,T],}}\frac{1}{N^3}\underset{a=1,...,3}{\sum_{i_a,j_a\in V_n}}\IP\left(\prod_{a=1}^3\tilde{C}_{n,i_aj_a}=1\big|C_{n,kl}(t)=1,C_{n,k'l'}(s)=1\right). \label{eq:swc5}
\end{align}}

To motivate this condition, note that the first term in \eqref{eq:swc1} requires that the probability of an edge being present at some time-point in the interval $[0,T]$ is of the same order as the probability being present at a fixed point $t\in[0,T]$. This appears to be a reasonable assumption if there are active edges, appearing and disappearing continuously, and inactive edges that never appear. This assumption excludes, for instance, the case where every edge appears exactly once at a random time point and exists for only a short time period. Furthermore, a sufficient condition for the third term in \eqref{eq:swc1} to be bounded is that $\tilde{C}_{n,12}$ and $\tilde{C}_{n,34}$ are independent. Note that this involves only disjoint pairs $(1,2)$ and $(3,4)$. If there is an overlap as for $(1,2)$ and $(2,3)$ appearing in the second term, the additional factor of $n$ in the denominator allows for strong dependence. This type of assumptions have also been used in \citet{KMP19} and \citet{K20}. For the remaining conditions similar interpretations can be found. The conditions on the conditional probabilities require that no single pair is indicative of the behaviour of the entire network. Assumption (WC) may be replaced by the weaker but more technical assumptions presented in Section \ref{app:lla} in the Supplement.

For the following technical assumption we consider, for each $t$ and $n$, a random distance function between pairs $d_t^{\,n}:V_n^2\times V_n^2\to[0,\infty)$, and for $I \subseteq V_n \times V_n$, set $d_t^{\,n}(ij,I):=\inf_{(k,l)\in I}d_t^{\,n}(ij,kl)$,  with the convention $d_t^{\,n}(ij,\emptyset):=\infty$.

\textbf{Assumption (mDep):} (Momentary-$m$-Dependence) \\
\emph{There is a number $M>0$ such that
\begin{align*}
&\forall n\in\IN,\forall t_0\in[0,T],\forall J\subseteq V_n\times V_n:\textrm{ Given }\mathcal{F}_{t_0}^{\,n}, \\
&\quad \left(N_{n,ij}(t),C_{n,ij}(t),X_{n,ij}(t)\right)_{(i,j)\in J,t\in[t_0,t_0+6h]}\textrm{ is conditionally independent of} \\
&\quad\sigma\Big(\left(N_{n,ij}(t),C_{n,ij}(t),X_{n,ij}(t)\right)\Ind\left(d_s^{\,n}(ij,J)\geq M\right): \\
&\quad\quad\quad\quad t_0-6h\leq s\leq t_0,t_0\leq t\leq s+6h, (i,j)\in V_n\times V_n\Big).
\end{align*}}

This assumption captures the situation of time developing networks where dependency structures may vary over time. The dependence structures are governed by the distance function $d_t^{\,n}$. An example for $d_t^{\,n}$ would be the length of the shortest path in the network between two edges (if existent, and $\infty$ otherwise). Intuitively, the above assumption means that, conditional on the past, the immediate short future of far apart processes can be treated as if they were independent. This is plausible if information needs some time to \emph{travel} through the network. See Supplement, section~\ref{app:mDep_mot} for a more detailed motivation for this assumption (see also \citet{K20}).

We call a pair $(i,j)\in V_n\times V_n$ a hub, if it has \emph{many close, active neighbouring pairs $(k,l)$ during a short period of time}. The number of close, active neighbouring pairs is defined as the number of pairs which are during any interval of length $6h$ simultaneously closer than $M$ (the constant from (mDep)) to $(i,j)$. In formulas this is (the choice of $[t-4h,t+2h]$ is somewhat arbitrary and made for later convenience)
$$K_M^{\,ij}:=\sup_{t\in[4h,T-2h]}\sum_{k,l\in V_n}\Ind\left(d_{t-4h}^{\,n}(kl,ij)<M\right)\Ind \big((k,l) \in E_{n,r} \text{ for some } r\in[t-4h,t+2h]\big).$$
For a given $n_{\rm hub}>0$, we call $(i,j)$ a hub if $K_M^{\,ij}\geq n_{\rm hub}$. Our weak correlation assumption (WC) allows for correlation between overlapping pairs $(i,j)$ and $(j,k)$. If the pair $(i,j)$ is a hub, this means there are many pairs $(j,k)$ in the edge set during some time interval. Our assumptions allow that during this time interval all such pairs $(j,k)$ are correlated with $(i,j)$. The existence of hubs poses therefore challenges when it comes to the behaviour of averages.

\textbf{Assumption (NH):} (No Hubs) \\
\emph{There is $n_{\rm hub}>0$ such that almost surely $K_M^{\,ij}\leq n_{\rm hub}$ for all $i,j\in V_n$}

(NH) is a simplifying assumption. We will prove our results under the weaker Assumption (HSR) which allows for the existence of hubs. (HSR) is given in the Supplement, in Section \ref{app:hubs}.

Our last assumption appears a bit clumsy. But the reader should note that all statements but the last would be trivially true if we had no conditional but regular expectations. Thus (BM) below excludes only pathologies in which single pairs react strongly to singular events. By $N_{n,ij}[a,b]$ we mean the number of jumps of the process $N_{n,ij}$ in the interval $[a,b]$ and we define for any process $X:[0,T]\to\IR$
\begin{align*}
\int |X(t)|d|M_{n,ij}|(t):=\int |X(t)|dN_{n,ij}(t)+\int |X(t)|\lambda_{n,ij}(t)dt.
\end{align*}

\textbf{Assumption (BM):} (Bounded Moments) \\
\emph{Let
\begin{align}
&\sup_{i,j,k,l\in V_n}\IE\left(\left(1+N_{n,ij}[0,T]+N_{n,ij}[0,T]N_{n,kl}[0,T]\right)^2\big|\tilde{C}_{n,ij}=1,\tilde{C}_{n,kl}=1\right)=O(1), \label{eq:BM1} \\
&\underset{s\in[0,T]}{\sup_{i,j,k_1,l_1,k_2,l_2\in V_n}}\IE\Big(N_{n,k_1l_1}[s-2h,s]N_{n,k_2l_2}[s-2h,s]\Big|C_{n,ij}(s)\tilde{C}_{n,k_1l_1}\tilde{C}_{n,k_2l_2}=1\Big)\to0, \label{eq:BM2} \\
&\underset{s\in[0,T]}{\sup_{i,j,k_1,l_1,k_2,l_2\in V_n}}\IE\Big(N_{n,k_1l_1}[s-2h,s]\Big|C_{n,ij}(s)\tilde{C}_{n,k_1l_1}\tilde{C}_{n,k_2l_2}=1\Big)=O(1), \label{eq:BM3} \\
&\underset{s\in[0,T]}{\sup_{i,j\in V_n,}}\sup_{k_1,l_1,...,k_4,l_4\in V_n}\IE\left(\prod_{a=1}^4N_{n,k_al_a}[0,T]\Bigg|\prod_{a=1}^4\tilde{C}_{n,k_al_a}C_{n,ij}(s)=1\right)=O(1), \label{eq:BM4} \\
&\sup_{i_1,\ldots,i_4\in V_n \atop j_1,\ldots,j_4\in V_n}\int_0^T\IE\left((N_{n,i_2j_2}[t-2h,t]+h)(N_{n,i_4j_4}[t-2h,t]+h)\Big|\prod_{a=1}^4\tilde{C}_{n,i_aj_a}=1\right)dt\to0, \label{eq:BM5} \\
&\sup_{i_1,\ldots,i_4\in V_n\atop j_1,\ldots,j_4\in V_n}\IE\Bigg[\int_0^T\left(h+N_{n,i_4j_4}[t-2h,t+2h]\right)\left(h+N_{n,i_2j_2}[t-2h,t+2h]\right) \nonumber \\
&\quad\times\underset{d_{t-2h}(uv,i_1j_1)<M}{\sup_{u,v\in V_n}}\left(h+N_{n,uv}[t-2h,t+2h]\right)d|M_{n,i_1j_1}|(t)\bigg|\prod_{a=1}^4\tilde{C}_{n,i_aj_a}=1\Bigg]\to0. \label{eq:BM6}
\end{align}}

The last statement requires also that the number of pairs $(u,v)$ which are closer than $M$ to $(i,j)$ is small. In that case the supremum is of finite order and the condition is reasonable. Assumption (BM) can be replaced by the assumptions given in Section \ref{app:lla} in the Supplement.

\subsection{Proof of Theorem \ref{thm:T1}}
\label{subsec:proof}
The general outline of the proof is simple. However there are many technical details that are very tedious to handle. We provide a comprehensive treatment of all the details in the Supplement, in Section \ref{supp:main_proof}, and show here the main steps only. Note that $C_n(t)=1$ if one of the processes $N_{n,ij}$ jumps at time $t$. We begin by rewriting the integrand in the test statistic in the following way (recall the definitions of $M_n$ and $N_n$ in the two paragraphs before \eqref{eq:ML})
\begin{align*}
&\hat{\alpha}_n(t;\tilde{\beta}_n)-\alpha_{\textrm{smooth}}(\hat{\theta}_n,t) \\
=&\int_0^TK_{h,t}(s)\frac{1}{\overline{\Psi}_n(s;\beta_0)}dN_n(s)-\int_0^TK_{h,t}(s)\alpha(\hat{\theta}_n,s)ds \\
&\quad\quad+\int_0^TK_{h,t}(s)\left(\frac{1}{\overline{\Psi}_n(s;\tilde{\beta}_n)}-\frac{1}{\overline{\Psi}_n(s;\beta_0)}\right)dN_n(s) \\
=&\int_0^TK_{h,t}(s)\frac{1}{\overline{\Psi}_n(s;\beta_0)}dM_n(s)+\int_0^TK_{h,t}(s)\left(\alpha(\theta_0,s)-\alpha(\hat{\theta}_n,s)+c_n\Delta_n(s)\right)ds \\
&+\int_0^TK_{h,t}(s)\left(\frac{1}{\overline{\Psi}_n(s;\tilde{\beta}_n)}-\frac{1}{\overline{\Psi}_n(s;\beta_0)}\right)dN_n(s) \\
=&\;I_1(t)+I_2(t)+I_3(t)+I_4(t),
\end{align*}
where
\begin{align}
I_1(t)&:=\int_0^TK_{h,t}(s)\frac{1}{\overline{\Psi}_n(s;\beta_0)}dM_n(s), \label{eq:defI1} \\
I_2(t)&:=c_n\int_0^TK_{h,t}(s)\Delta_n(s)ds, \label{eq:defI2} \\
I_3(t)&:=\int_0^TK_{h,t}(s)(\alpha(\theta_0,s)-\alpha(\hat{\theta}_n,s))ds, \label{eq:defI3} \\
I_4(t)&:=\int_0^TK_{h,t}(s)\frac{\overline{\Psi}_n(s;\beta_0)-\overline{\Psi}_n(s;\tilde{\beta}_n)}{\overline{\Psi}_n(s;\tilde{\beta}_n)\overline{\Psi}_n(s;\beta_0)}dN_n(s). \label{eq:defI4}
\end{align}
With this notation, we obtain that
\begin{align*}
&N\sqrt{h}T_n=N\sqrt{h}\int_0^T\left(\hat{\alpha}_n(t;\tilde{\beta}_n)-\alpha_{\textrm{smooth}}(\hat{\theta}_n,t)\right)^2w(t)dt \\
=&\;N\sqrt{h}\int_0^TI_1(t)^2w(t)dt+N\sqrt{h}\int_0^TI_2(t)^2w(t)dt+\sum_{i=3}^4N\sqrt{h}\int_0^TI_i(t)^2w(t)dt \\
&\quad+N\sqrt{h}\underset{i\neq j}{\sum_{i,j=1}^4}\int_0^TI_i(t)I_j(t)w(t)dt.
\end{align*}
Now, Lemmas \ref{lem:I3} and \ref{lem:I4} in the supplementary Section \ref{sec:support_lemmas} show that the terms involving $I_3(t)^2$ and $I_4(t)^2$ converge to zero. Furthermore, by using additionally the Cauchy-Schwarz Inequality and Lemmas \ref{lem:I2} and \ref{lem:I12}--\ref{lem:I14}, we see that also the cross-terms in the second row converge to zero. Thus, the asymptotic behavior of $T_n$ is entirely determined by the terms $I_1$ and $I_2$. The term involving $I_2$ equals exactly what we have to subtract in the formulation of Theorem \ref{thm:T1} (mind the definition of $c_n$)
\begin{align*}
N\sqrt{h}\int_0^TI_2(t)^2w(t)dt=\int_0^T\left(\int_0^TK_{h,t}(s)\Delta_n(s)ds\right)^2w(t)dt.
\end{align*}
Connecting all the previous results, we see that we have to show that
\begin{align*}
B_n^{-\frac{1}{2}}\left(N\sqrt{h}\int_0^TI_1(t)^2w(t)dt-h^{-\frac{1}{2}}A_n\right)\overset{d}{\to}\mathcal{N}(0,1)
\end{align*}
in order to complete the proof of Theorem \ref{thm:T1}. This is exactly the content of Proposition~\ref{prop:mmd} which will be shown Section \ref{app:proof} in the Supplement.

The main tool for the proof of Proposition \ref{prop:mmd} is Rebolledo's Martinale Central Limit Theorem (cf. Theorem \ref{thm:Rebolledo} in the Supplement). In Rebolledo's CLT a sequence of martingales converges to a Gaussian Limit Process. This requires two ingredients: Firstly, the sequence of processes must be a sequence of martingales and, secondly, the variation process of these martingales must stabilize. In our application the process of interest will be driven by $M_n$. $M_n$ is a martingale by definition of our counting process set-up. This is a time-wise property of the process and is as such not affected by the network structure (which is a spatial property). Since the asymptotic normality is a consequence of this time-wise martingale property, we may expect that the asymptotic distribution is the same as for the non-network case. The second ingredient, however, the stabilisation of the variance process, is a spatial property. This is naturally affected by the network structure. Thus, network dependence can possibly jeopardize a stabilisation of the variance process. We provide here assumptions under which this does not happen. Note that the speed of convergence is not of direct interest as long as the limit is correct. This task is, unfortunately, much more difficult than it might at first appear.

Therefore, we have to make assumptions on the network itself. Most importantly, this is Assumption (VX) which allows to reduce convergences of sums to correct behaviour of covariances as we detailed after Assumption (WC). For the network this is sufficient because we may use the notion of $\tilde{C}_{n,ij}$ which allows to disentangle the network from time. This is a strong assumption but we provide in Section \ref{app:lla} a set of alternative assumptions which may replace (WC) and does not make use of $\tilde{C}_{n,ij}$. For the covariates we do not have such a simple solution and therefore we require in Assumption (LL) concentration of the covariates.

\section{Numerical Results}
\label{sec:numerics}
\subsection{Implementation}
\label{subsec:implementation}
\footnote{The R-code is available at \url{https://github.com/akreiss/Baseline-Estimation.git}.} When it comes to applying the test in a real-world situation, one needs two independent sets of observation, and one also has to specify a model for the baseline function. The first set of data is used to compute the parametric estimates $(\hat{\theta}_n,\hat{\beta}_n)$ as in \eqref{eq:ML} as well as the partial likelihood estimate $\tilde{\beta}_n$ as in \eqref{eq:PL}. Then, the smoothed Nelson-Aalen estimator $\hat{\alpha}_n(t,\tilde{\beta}_n)$ will be computed based on the second set of data according to formula \eqref{eq:nelson-aalen}. After specifying a weight function, the test statistic $T_n$ can be computed as in \eqref{eq:Tn}. The critical value of the test is determined according to Theorem \ref{thm:T1}. Since we need the distribution of $T_n$ on the null-hypothesis, we set  $\Delta_n\equiv0$ but we have to estimate $A_n, B_n$ and $N$. All of the following estimates are computed based on the second dataset. For estimating $A_n$ we use its representation from Proposition \ref{prop:mmd}. To estimate $A_n$ we replace $\mu_n(r;\beta_0)$ by the corresponding average where $\beta_0$ is replaced by $\tilde{\beta}_n$. More precisely, $\mu_n(r;\beta_0)$ is estimated by
$$\left(\sum_{i,j\in V_n}C_{n,ij}(r)\right)^{-1}\left(\sum_{i,j\in V_n}C_{n,ij}(r)\exp\left(X_{n,ij}(r)'\tilde{\beta}_n\right)\right).$$
Furthermore, $mp_n(r)$ (the expected number of edges) is estimated by the observed number of edges at time $r$. The number $N$ can be estimated using the same estimates. In order to compute the integral in the definition of $N$ and in further quantities, we suppose that the network and the covariates remain constant over known time-intervals, e.g. they change at every hour but remain then constant for 60 minutes. Note that so far we have used only the partial likelihood estimator which works well on the hypothesis and on the alternative. The variance $B_n$ can be estimated using the same conventions with the following addition: To find critical values for the test, we may assume that we are on the hypothesis and hence we estimate the true global intensity $\alpha_0(t)$ by $\alpha(Z(t);\hat{\theta}_n)$. For finding the confidence area mentioned in Remark \ref{rem:conf_int} we use the nonparametric estimator $\hat{\alpha}_n(\cdot;\tilde{\beta}_n)$.

\subsection{Empirical Application}
\label{sec:data}
Here we apply the methodology from Section \ref{sec:model} to bike sharing data, details of the implementation are mentioned in Section \ref{subsec:implementation}. The data is based on $527$ bike stations in Washington D.C. and its surrounding. We consider these bike stations to be the vertices of a network. An interaction from bike station $i$ to bike station $j$ happens if someone rents a bike at $i$ and returns it at $j$. We consider such data from May 13, 2018 (Sunday) to May 26, 2018 (Saturday). The data is publicly available and can be downloaded from \url{https://www.capitalbikeshare.com/system-data}. Since many of the connections are rarely (or never) used, we attempt to model only the frequently used connections. To this end we construct a network of active pairs as follows: There is a link from bike station $i$ to station $j$ if there were at least ten bike rides in April 2018 from $i$ to $j$. So we only consider pairs that were active at least twice a week on average over a period of one month. This convention is somewhat arbitrary and a full analysis would require a sensitivity analysis with respect to this choice. Also note that this is a directed network. This is no problem because the result in Theorem \ref{thm:T1} holds analogously for directed networks.

\begin{figure}
\includegraphics[width=\textwidth]{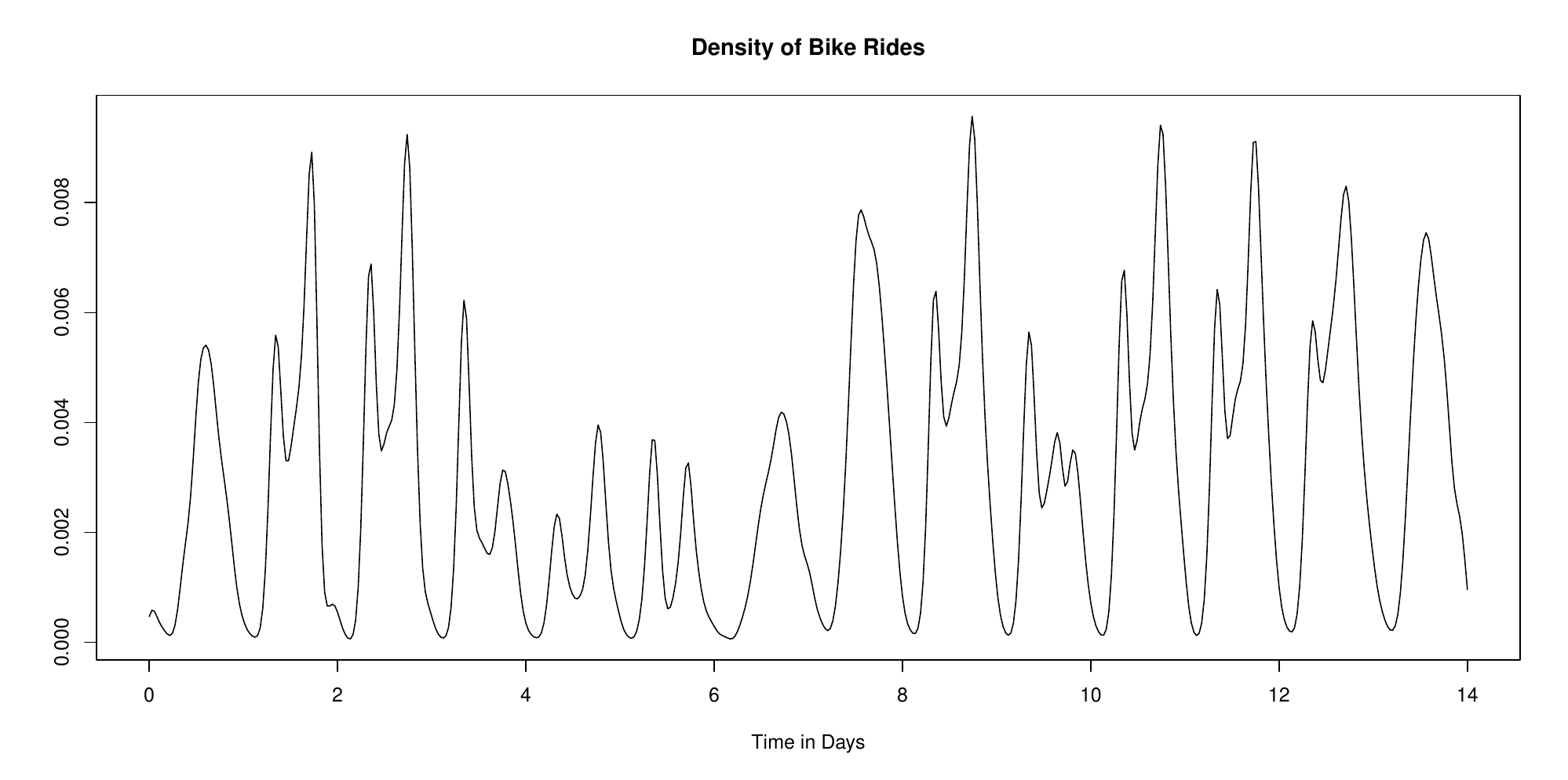}
\caption{Kernel density estimate for bike rides over a period of two weeks in 2018 (May 13 - May 26).}
\label{fig:density}
\end{figure}

Figure \ref{fig:density} shows a kernel density estimate of the times at which bike rides happened over the period from May 13 till May 26. The scale of the x-axis are days, with day $1$ and $8$ being Sundays and days $7$ and $14$ being Saturdays. One clearly sees a different pattern for working days and weekends. Moreover it appears that in the second week there were more bike rides than in the first week. One might suspect that this is due to the weather situation: Figure \ref{fig:weather} shows that there was more rain in the first week than in the second week. The weather information is publicly available from \url{https://www.wunderground.com/} and was collected at the Ronald Reagan Airport in Washington D.C. Even though the Washington D.C. area is large, it is plausible that this local weather information is a valid indicator for the weather in the entire region. Therefore, the temperature and the precipitation are system-wide covariates that are identical for all pairs. Let $T(t)$ denote the $\log$ of the temperature at time $t$ (in degree centigrade, there were no negative temperatures) and let $P(t)$ denote the precipitation (in centimeters). As a parametric model for $\alpha_0$ we consider the following ($\theta=(\theta_1,...,\theta_{17})$)

\begin{align*}
\alpha(\theta,t)=&\exp\Bigg(\theta_1+\begin{pmatrix}T(t) & T(t)^2\end{pmatrix}'\begin{pmatrix}\theta_2 \\ \theta_3\end{pmatrix}+\begin{pmatrix}P(t) & P(t)^2\end{pmatrix}'\begin{pmatrix}\theta_4 \\ \theta_5\end{pmatrix} \\
&\quad+\begin{pmatrix}\sin\left(t\frac{2\pi}{24}\right) \\\sin\left(2t\frac{2\pi}{24}\right) \\ \sin\left(3t\frac{2\pi}{24}\right)\end{pmatrix}'\begin{pmatrix}\theta_6 \\ \theta_7 \\ \theta_8\end{pmatrix}+\begin{pmatrix}\cos\left(t\frac{2\pi}{24}\right) \\ \cos\left(2t\frac{2\pi}{24}\right) \\ \cos\left(3t\frac{2\pi}{24}\right)\end{pmatrix}'\begin{pmatrix}\theta_9 \\ \theta_{10} \\ \theta_{11}\end{pmatrix} \\
&\quad+W(t)\begin{pmatrix}\sin\left(t\frac{2\pi}{24}\right) \\\sin\left(2t\frac{2\pi}{24}\right) \\ \sin\left(3t\frac{2\pi}{24}\right)\end{pmatrix}'\begin{pmatrix}\theta_{12} \\ \theta_{13} \\ \theta_{14}\end{pmatrix}+W(t)\begin{pmatrix}\cos\left(t\frac{2\pi}{24}\right) \\ \cos\left(2t\frac{2\pi}{24}\right) \\ \cos\left(3t\frac{2\pi}{24}\right)\end{pmatrix}'\begin{pmatrix}\theta_{15} \\ \theta_{16} \\ \theta_{17}\end{pmatrix}\Bigg),
\end{align*}

where $W(t)$ is an indicator function which equals $1$ if the $t$ lies on a weekend. Since the weather data is only available for every hour, we consider the functions $P(t)$ and $T(t)$ as piece-wise constant. Note that the essential assumption of continuity of the link function is not violated because we only require continuity in the parameter. The constant $L_{\alpha}(t)$ is bounded because temperature and precipitation as well as the other functions are bounded. The Assumption (VX) is fulfilled because the vertex-labels, i.e., the bike station IDs, were virtually randomly assigned to the bike stations. The continuity of $t\mapsto\mu_n(t;\beta_0)$ required in Assumption (C) is more of an issue. However, in the limit this is not a problem because the continuity is used for kernel approximations, thus these approximations are only problematic at the discontinuity. Since we integrate over the entire observation period, these approximation errors on short intervals are not of big importance.

\begin{figure}
\includegraphics[width=\textwidth]{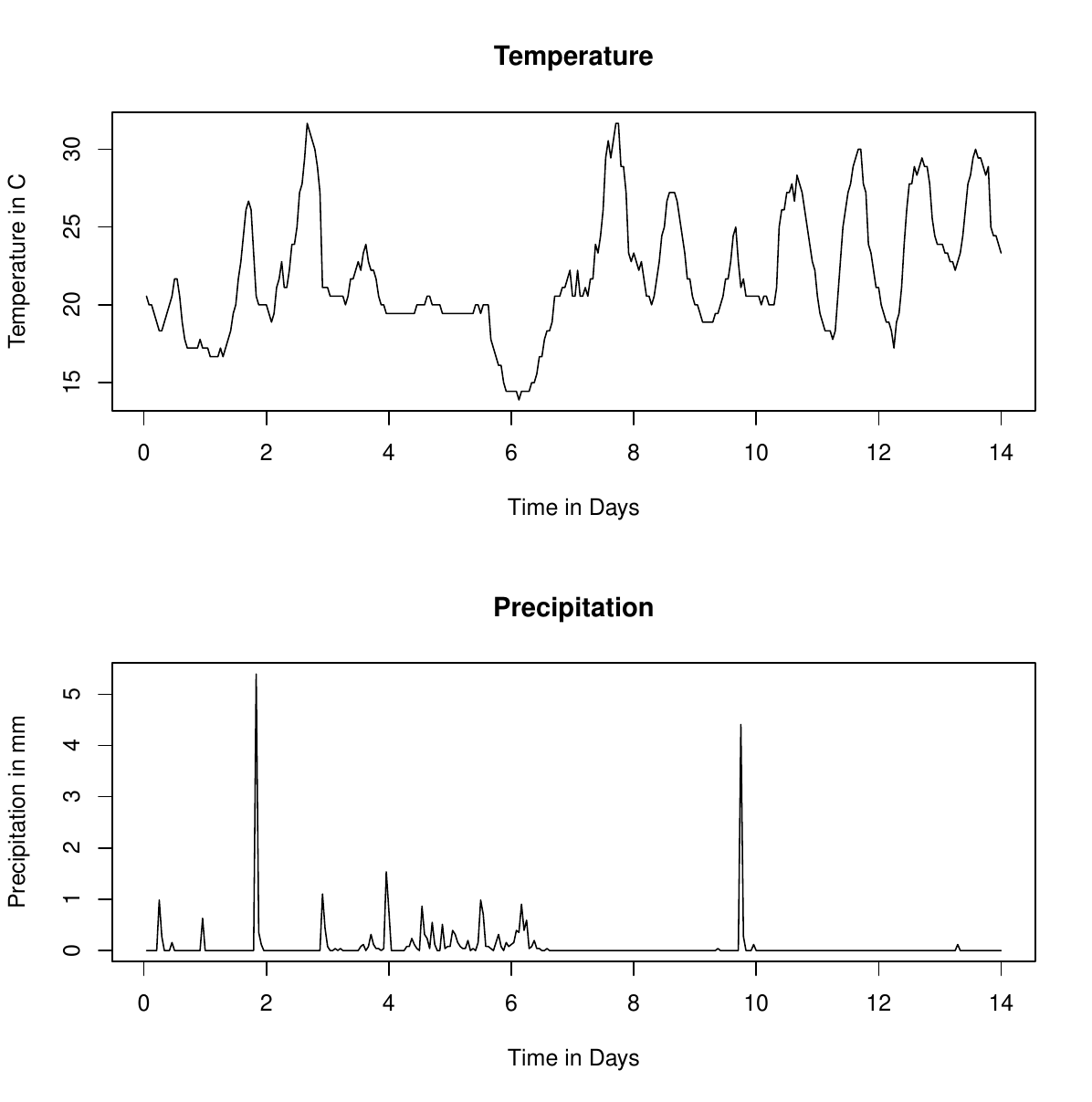}
\caption{Temperature and Precipitation in Washington D.C.}
\label{fig:weather}
\end{figure}

We consider covariates on pair-wise level based on distances. Let $d_{i,j}$ denote the logarithm of the biking time in minutes from station $i$ to station $j$ as returned by Google maps. Then, we consider a simple bi-variate covariate vector
$$X_{n,ij}=\binom{d_{i,j}}{ d_{i,j}^{\,2}}.$$
As a link function we consider $\Psi(x;\beta)=\exp(x'\beta)$ and hence we have a Cox proportional hazards model. We compute the estimates $\tilde{\beta}_n$ and $(\hat{\theta}_n,\hat{\beta}_n)$ as it was described in Section \ref{subsec:implementation} based on the first week, i.e, from May 13 to May 19. The estimates $\hat{\beta}_n$ and $\tilde{\beta}_n$ are almost identical and have roughly the values $(0.219,-0.147)'$. The negative sign of the parameter for $d_{i,j}^2$ shows that the intensity (as a function of the distance) has a maximum. This is plausible because people will likely not use the bikes for very short or very long tours. Now we compute the non-parametric Nelson-Aalen estimator of the baseline intensity based on the second week and compare the estimate with the parametric estimator. Here we have chosen the bandwidth to be $30$ min. This has been obtained by \emph{eye-balling}. The assumption of the independence of the two weeks is plausible because we expect that people do not base their biking decisions on the past week other than possibly indirectly through the weather which we control for.

Both estimators are shown in Figure \ref{fig:estimates}. We can see that in particular the drop in bike rides on Tuesday (Day 3) (possibly due to the rain) is well captured by the parametric model. Also the difference between weekends and working days is well visible in both estimators. However, there appears to be an overestimation of the activity on the weekends by the parametric estimator. We apply our test which is based on a modified $L^2$-distance (using smoothing and weighting) as described in \eqref{eq:Tn}. By using the weight function $w$ we can restrict the test to working days, weekends or the whole week. In all three cases, the test based on $T_n$, centralized and scaled according to Theorem \ref{thm:T1}, rejects the hypothesis of a baseline that is entirely driven by global covariates. The observed $p$-values are $1-\Phi(364)\approx 0$ (whole week), $1-\Phi(627)\approx 0$ (working days), and $1-\Phi(169)\approx 0$ (weekend). To investigate the deviation further, it is natural to compare the deviation between parametric and nonparametric estimator relative to the parametric estimator in the $L^2$-norm
$$\frac{\sqrt{\int_0^7\left(\hat{\alpha}_n(t;\tilde{\beta}_n)-\alpha(\hat{\theta}_n,t)\right)^2dt}}{\sqrt{\int_0^7\alpha(\hat{\theta}_n,t)^2dt}}\approx 0.32,$$
that is, around $32\%$ of the baseline $\alpha_0$ are contributed by the deviation $c_n\Delta_n$ and the remaining $68\%$ are explained by the global covariates. Our Theorem \ref{thm:T1} and the Remark \ref{rem:conf_int} allow us to make statements like this on, e.g., $95\%$ confidence level, as follows. According to Remark \ref{rem:conf_int} we can compute that, with probability $95\%$, $\|K_h\star(c_n\Delta_n)\|_w\geq0.22$. Dividing this number by $\|K_h\star\hat{\alpha}_n(\cdot;\tilde{\beta}_n)\|_w$ yields a contribution of around $29.4\%$ or more by $c_n\Delta_n$. Put differently, at most $70.6\%$ of $\alpha_0$ are explained through the global covariates. Using the weight function $w$ to restrict to working days only, yields a contribution of $c_n\Delta_n$ of $28.8\%$ to $\alpha_0$. Repeating the same exercise on weekends, yields a percentage of $31.2\%$. Hence, we see in both cases a relatively large misfit of the parametric model.

\begin{figure}
\includegraphics[width=\textwidth]{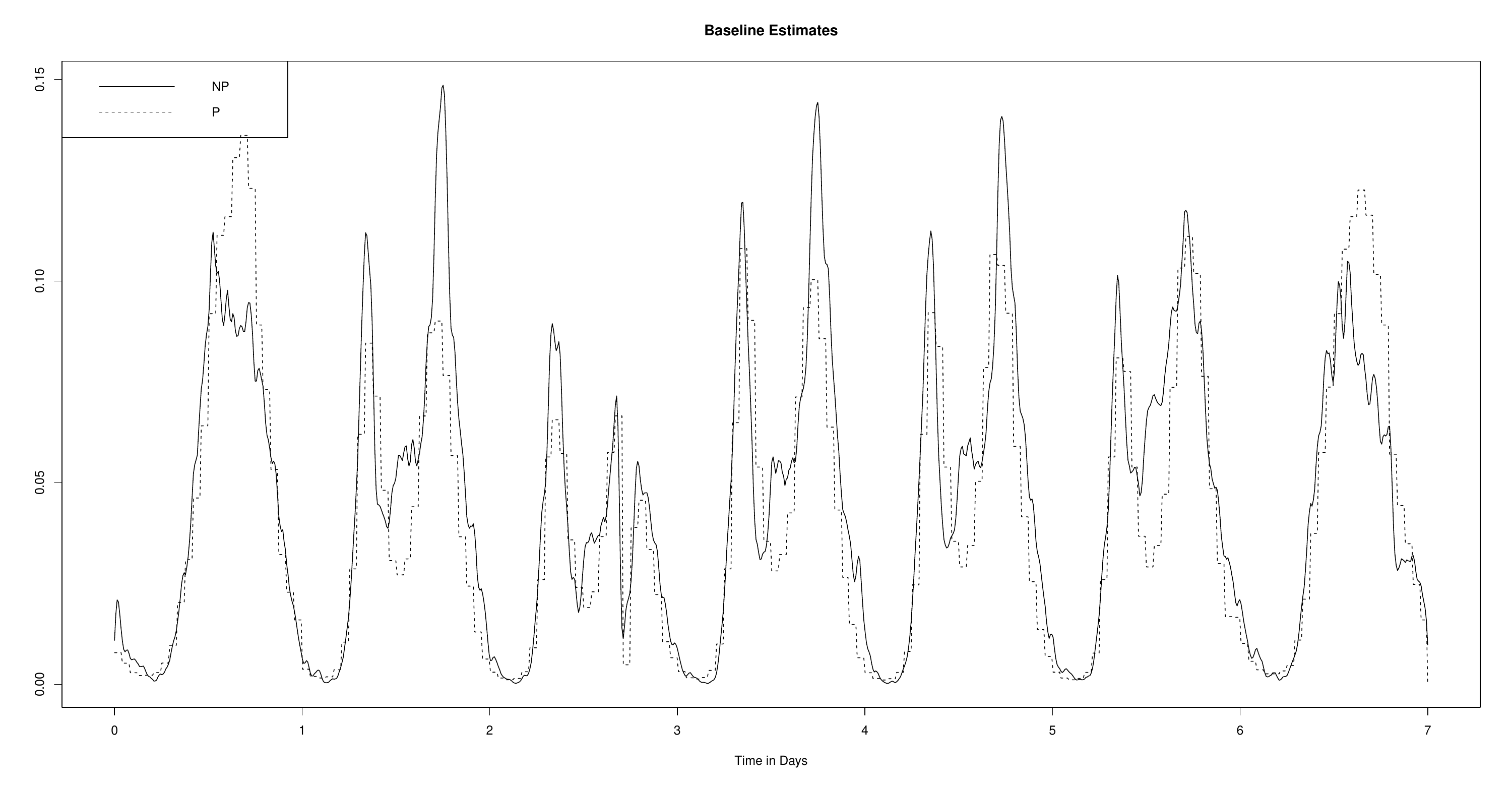}
\caption{Parametric (P) and Non-parametric (NP) estimates of the baseline intensity for the second week.}
\label{fig:estimates}
\end{figure}

\subsection{Simulation}
\label{subsec:simulation}
In this section we will simulate the power function of the test procedure. In order to simulate reasonable data, we fit the following model to the bike data from Section \ref{sec:data}: As pair-wise covariate $X_{ij}$ we use the number of joint neighbours of $i$ and $j$, that is, the number of directed paths from $i$ to $j$ of length $2$. For the global covariates we use the same parametric model as in Section \ref{sec:data}. This results in estimates $\hat{\theta}_n\in\IR^{17}$ and $\hat{\beta}_n\in\IR$. The fitted parameters are used in the subsequent simulations as true underlying parameters. In particular, we choose $\beta_0$ equal to $\hat{\beta}_n$. As baseline we use $\alpha_{\rho}(t):=\bar{\alpha}(t)+\rho(\alpha_1(t)-\bar{\alpha}(t))$, where $\rho\in[0,1]$ and $\alpha_1(t)=\alpha(\hat{\theta}_n,t)$ using the weather information from Section \ref{sec:data}. The function $\alpha_1$ changes hourly over a period of two weeks, thus it takes $14\cdot24$ different values. The function $\bar{\alpha}$ is constant and it is the average of $\alpha_1$:
$$\bar{\alpha}(t)=\log\left(\frac{1}{14\cdot24}\int_0^{14\cdot24}\alpha_1(s)ds\right)$$
Hence, for $\rho=0$ the null hypothesis is
$$H_0:\quad \alpha(t)=\exp(\theta)\textrm{ for some }\theta\in\IR,$$
The value of $\rho$ controls the deviation from the null hypothesis. Figure \ref{fig:baselines} shows $\alpha_{\rho}(t)$ for the choices $\rho=0, 0.05, 0.1, 0.2, 0.5$.

In order to obtain a network randomly changing in time, we adopt the model from \citet{JLY20} to our situation by allowing the parameters for weekends and working days to be different. We fit this model to the observed data from Section \ref{sec:data} and again we use the fitted parameters as model parameters for the simulations. We consider two simulation set-ups similar to Section \ref{sec:data}: One where the network has $n=527$ vertices (as in Section \ref{sec:data}) and one for a smaller network with $n=300$ vertices. In each instance of the simulation, we generate a period of two weeks and split the period in two segments of one week each. We then use the estimation procedure as explained in Section \ref{subsec:implementation}. We simulate $2,000$ such datasets from this model for the choices $\rho=0, 0.05, 0.1, 0.2, 0.5$ (cf. Figure \ref{fig:baselines} for a plot of the baselines). For each choice of $\rho$ we compute the percentage of rejections of our test when it is calibrated to level $\alpha=0.05$. The result is shown in Figure \ref{fig:power}. We see that for both network sizes the test complies with the desired level on the hypothesis $\rho=0$ and that it has power against the alternatives. We show in Section \ref{supp:simulation} of the supplement the finite sample distributions of the test statistics. It can be seen that the normal approximation is quite good justifying tests based on the asymptotic result.

\begin{figure}
\centering
\includegraphics[width=\textwidth]{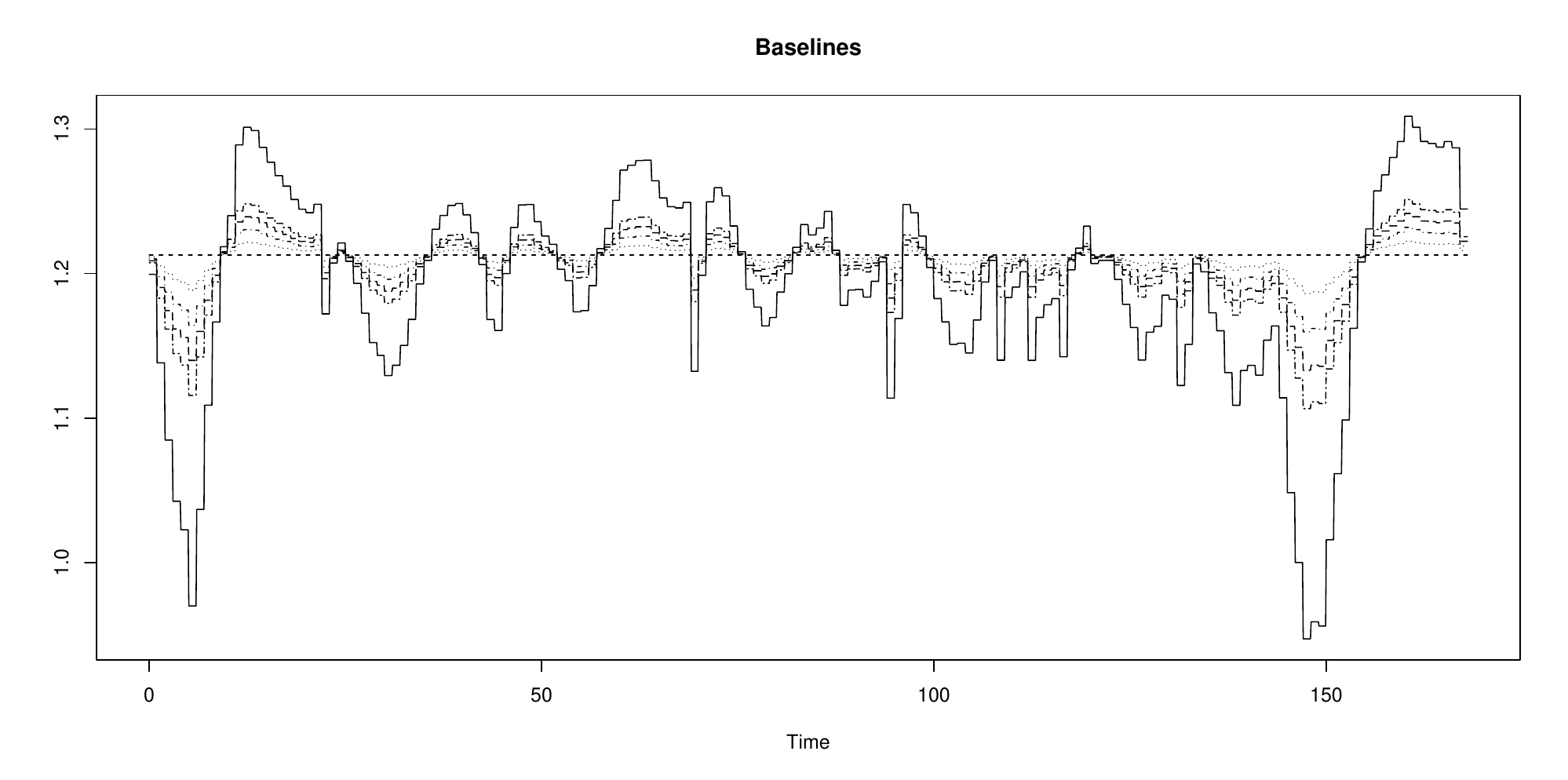}
\caption{Choices of baseline intensities used for the simulations (x-axis in hours)}
\label{fig:baselines}
\end{figure}

\begin{figure}
\centering
\includegraphics[width=0.6\textwidth]{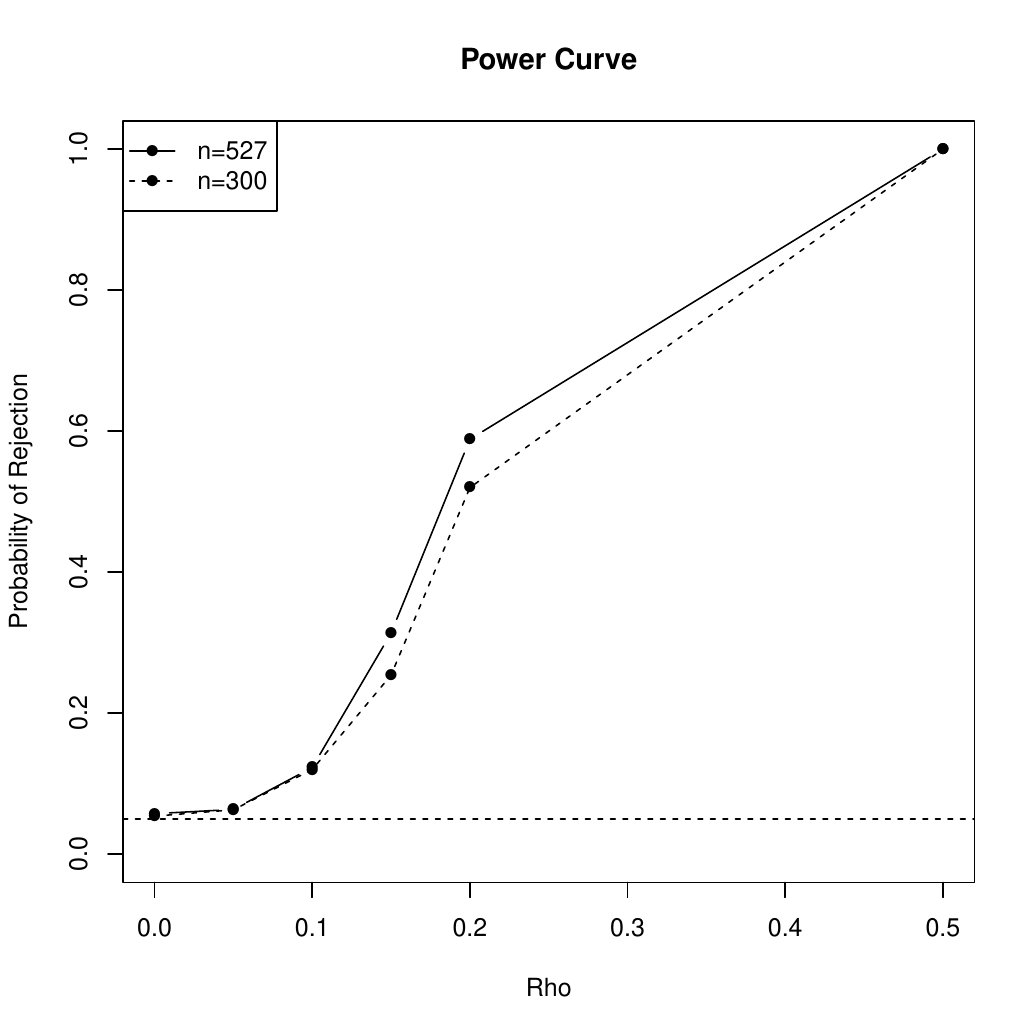}
\caption{Percentage of rejections of a test of level $\alpha=0.05$.}
\label{fig:power}
\end{figure}

\section{Conclusion}
\label{sec:conclusion}

In this paper we have proposed a test statistic for parametric network models of global covariates. We use weak dependence assumptions to derive the asymptotic behavior of the test statistic. This can be used to obtain critical values. Moreover, from a mathematical point of view, we have illustrated how weak dependence assumptions on networks can be developed in order to allow the derivation of asymptotic results. The assumptions allow for a certain dependence between neighbors in a network but require that the dependence is weak for vertices which are far apart in the network.

As an illustration we have applied the test to bike share data. We saw that even though the weather and time information used by the model appear to explain the data well, there is still some structure that is not unraveled by the particular model under study. This motivates further study in the direction of modeling interactions on networks. As a direct consequence, an interesting future research question would be to test if the fully non-parametric model from \citet{KMP19} offers needed flexibility or if a hybrid model like it was studied here offers enough flexibility. More generally, the link-functions $\alpha_0$ and $\Psi$ can be modelled parametrically or non-parametrically, yielding four possible models. In order to understand the situation in a real-world dataset well, all $\binom{4}{2}$ possible comparisons are of interest. In addition, a thorough treatment of the bandwidth would be desirable.

\newpage

\begin{center}
\textbf{\Large Supplement}
\end{center}

\section{Relation to Kreiss (2021) and Mixing}
\label{app:mixing}
The proofs of this paper are based on the concept of momentary-$m$-dependence which was introduced in \citet{K20} and which corresponds to our Assumption (mDep). We apply it in two places: In the proof of Proposition \ref{prop:mmd} we utilize Rebolledo's Martingale Convergence Theorem \ref{thm:Rebolledo} to obtain the desired convergence. We show the requirements of Proposition 2.9 and Theorem 2.11 from \citet{K20} to show that the quadratic variation process of a certain martingale converges which is a requirement of Rebolledo's Central Limit Theorem. So the main theoretical challenge here lies in finding a suitable expansion of the test statistic and to bound all nuisance terms. In addition, we have to prove in Lemma \ref{lem:I14} that a certain cross-term involving martingales converges to zero. After removing all nuisance terms the core of the problem is to bound the difference between two averages evaluated at $\beta_0$ and $\tilde{\beta}_n$ which is the sum of dependent random variables and therefore requires some weak dependence relation. We show in Lemma \ref{lem:r22} that momentary-$m$-dependence is a suitable concept to handle such a problem. Finally, the mixing concepts from \citet{K20} can be used to prove the Assumption (LL). We give more details now.

We show how Assumption (LL) in the main text can be shown to hold under certain mixing assumptions. Mixing on networks has been discussed in \citet{K20} and we give here just a quick overview. In order to formulate the mixing assumptions we need the following notation. Let $\delta_n=a \log n$ for some fixed $a>0$ and consider a random, time-varying partition $\left(G^t(k,b,\delta_n)\right)_{k,b}\subseteq E_{n,t}$ of the edge set, where $k=1,...,\mathcal{K}$ and $b=1,...,m$ (so some of the sets are empty). We call $G^t(k,b,\delta_n)$ the $b$-th group of type $k$. Note that $\mathcal{K}$ is fixed. This partition shall have the following properties
\begin{enumerate}
\item $(k,b)\neq (k',b')\Rightarrow G^t(k,b,\delta_n)\cap G^t(k',b',\delta_n)=\emptyset$,
\item $\bigcup_{k,b}G^t(k,b,\delta_n)=E_{n,t}$,
\item For $b\neq b'$: $(i,j)\in G^t(k,b,\delta_n)$, $(k,l)\in G^t(k,b',\delta_n)$ $\Rightarrow$ $d_t^n(ij,kl)\geq \delta_n$.
\end{enumerate}
Let $I_{n,ij}^{k,b,t}:=\Ind\left((i,j)\in G^t(k,b,\delta_n)\right)$ be the indicator which states if the pair $(i,j)$ is contained in the set $G^t(k,b,\delta_n)$. The idea is now to rewrite the sum of interest in $\mathcal{A}_n$ and $\tilde{\mathcal{A}}_n$ as follows
$$\sum_{i,j\in V_n}C_{n,ij}(t)\Psi(X_{n,ij}(t);\beta)=\sum_{k,b}U_{k,b}^t,\textrm{ where } U_{k,b}^t:=\sum_{i,j\in V_n}C_{n,ij}(t)\Psi(X_{n,ij}(t);\beta)I_{n,ij}^{k,b,t}.$$
The $\beta$-mixing coefficient of the network is then defined as follows
$$\beta_t(n):=\underset{k=1,...,\mathcal{K}}{\max_{B=1,...,m}}\beta\left(\left[U_{k,b}^t\right]_{b\leq B-1},U_{k,B}^{n,t}\right).$$
In order to make use of these concepts, we have to define the maximal size of a group of type $k$, $S_k(t):=\max_{b=1,...,m}|G^t(k,b,\delta_n)|$ and the indicator
$$\Gamma_n^t(c):=\left(\max_{k=1,...,\mathcal{K}}S_k(t)^2\leq c^2\frac{mp_n(t)}{\log mp_n(t)}\right)$$
for $c>0$ which indicates if no single group is too large. Moreover, we have to define the restriction of the random variables of interest to the event that the partition fulfills $\Gamma_n^t(c)=1$,
$$Y_{n,ij}(t):=\left(C_{n,ij}(t)\Psi(X_{n,ij}(t);\beta)-\IE\left(C_{n,ij}(t)\Psi(X_{n,ij}(t);\beta)\big| \Gamma_n^t=1\right)\right)\Gamma_n^t.$$
Let $E_{k,b}^{n,t}:=\IE\left(|G^t(k,b,\delta_n|\right)$ be the expected size of a group. We need to make the following assumptions.

\textbf{Assumption (Mixing)}\\
{\it Suppose that for any $c_0>0$ we can choose $a,c_1>0$ such that for $\delta_n=a\log n$, $\beta_t(n)\leq n^{-c_0}$ and $\IP(\Gamma_n^t(c_1)=0)\leq n^{-c_0}$ for all $t\in[0,T]$. Suppose furthermore that for each such $a$, there are constants $c_2,c_3>0$ such that for all $t\in[0,T]$
$$\frac{1}{mp_n(t)}\sum_{b=1}^mE_{k,b}^{n,t}\geq c_2,$$
and that for pairwise different vertices $v_1,...,v_4\in V_n$ and all $t\in[0,T]$
\begin{align*}
\textrm{Var}\left(Y_{n,v_1v_2}(t)I_{n,v_1v_2}^{k,b,t}\right)&\leq c_3\|\Psi(\cdot;\beta)\|_{\infty}^2E_{k,b}^{n,t}m^{-1}, \\
\textrm{Cov}\left(Y_{n,v_1v_2}(t)I_{n,v_1v_2}^{k,b,t},Y_{n,v_2v_3}(t)I_{n,v_2v_3}^{k,b,t}\right)&\leq c_3\|\Psi(\cdot;\beta)\|_{\infty}^2E_{k,b}^{n,t}m^{-\frac{3}{2}}, \\
\textrm{Cov}\left(Y_{n,v_1v_2}(t)I_{n,v_1v_2}^{k,b,t},Y_{n,v_3v_4}(t)I_{n,v_3v_4}^{k,b,t}\right)&\leq c_3\|\Psi(\cdot;\beta)\|_{\infty}^2E_{k,b}^{n,t}m^{-2}.
\end{align*}}
The existence of such a partition is in particular plausible in networks which are related to geographic locations (like the bike stations). The exponential decay of the $\beta$-mixing coefficient is motivated through an exponential growth of the $\delta_n$-neighborhoods, i.e., the number of pairs in a distance smaller than $\delta_n$. For an example we refer to \citet{K20}. For the assumptions on the correlations note that $\IE(Y_{n,ij}(t)^2)\approx \|\Psi(\cdot;\beta)\|_{\infty}^2p_n(t)$ and that, at least for reasonable partitions, $\IE(I_{n,ij}^{k,b,t})\approx E_{k,b}^{n,t}/mp_n(t)$.

Under these conditions we can prove Assumption (LL).
\begin{proposition}
\label{prop:supbound}
In addition to the assumptions of Section \ref{subsec:assumptions} (other than (LL)): Let $d_0>0$ be arbitrary but fixed. Suppose that the (Mixing) Assumption holds. Suppose furthermore that $\pi$ and $\mu_n$ are H\"{o}lder continuous with exponents $\alpha_1,\alpha_2$, respectively. Suppose that $\sup_{\beta\in B_n(d_0)}\|\Psi(\cdot;\beta)\|_{\infty}=O(1),$ and that
$$\left|\Psi(x_1,\beta)-\Psi(x_2,\beta)\right|\leq L_{\Psi}(\beta)\tilde{\omega}_{\Psi}\left(\|x_1-x_2\|\right),$$
where $\sup_{\beta\in B_n(d_0)}L_{\Psi}(\beta)=O(1)$. Let $k_0>0$ be such that $n^{-\alpha_1 k_0}mp_n/\log m\to0$ and $n^{-\alpha_2 k_0}\sqrt{m/\log m}\to0$ and $n^{-k_0}\sqrt{m/p_n\log m}\to0$. Let $\#\textrm{jumps}(C_n)[n^{-k_0}]$ denote the maximal number of jumps of $C_n$ in any interval of length $n^{-k_0}$. Assume that
\begin{align*}
&\frac{\#\textrm{jumps}(C_n)[n^{-k_0}]}{mp_n}=O_P\left(\sqrt{\frac{\log m}{mp_n}}\right),  \\
&\sup_{t_1,t_2:|t_1-t_2|\leq n^{-k_0}}\sum_{(i,j)\in V_n\times V_n}\frac{\tilde{\omega}_{\Psi}\left(\left\|X_{n,ij}(t_1)-X_{n,ij}(t_2)\right\|\right)}{mp_n}=O_P\left(\sqrt{\frac{\log m}{mp_n}}\right).
\end{align*}
Finally, let $\eta_0:=\inf_{\beta\in B_n(d_0)}\left\|\mu_n(\cdot;\beta)\right\|_{\infty}$ and assume that $\eta_0^{-1}=O(1)$. Then, for any $\delta>0$, there is $d_1>0$ such that
$$\IP\left(\mathcal{A}_n(d_1,d_0)\right)\geq1-\delta\textrm{ and }\IP\left(\tilde{\mathcal{A}}_n(d_1,d_0)\right)\geq1-\delta.$$
\end{proposition}
\begin{proof}
The proof uses a chaining type argument: Let
\begin{equation*}
T_{n,k_0}:=\left\{\left(\frac{a_j}{n^{k_0}}\right)_{j=0,...,q}:a_j\in\IZ\cdot\frac{2}{\sqrt{1+q}}, j=0,...,q\right\}\cap\left([0,T]\times \overline{B}_n(d_0)\right)
\end{equation*}
be a finite grid of $[0,T]\times B_n(d_0)$ of mesh $2n^{-k_0}/\sqrt{1+q}$. Then,
$$\#T_{n,k}=O\left(n^{(1+q)k_0}/(mp_n)^{q/2}\right).$$
Denote by $\pi_{n,k_0}:[0,T]\times B_n(d_0)\to T_{n,k_0}$ the projection, i.e, for any $(t,\beta)\in[0,T]\times B_n(d_0)$ it holds that $\left\|(t,\beta)-\pi_{n,k_0}(t,\beta)\right\|\leq n^{-k_0}$. In order to avoid too complicated notation we write $\pi_{n,k_0}(t,\beta)=\left(\pi_{n,k_0}(t),\pi_{n,k_0}(\beta)\right)$ such that $\pi_{n,k_0}(t)$ is the time grid point which lies closest to $t$ and $\pi_{n,k_0}(\beta)$ is the parameter grid point lying closest to $\beta$.

We begin by studying $\tilde{\mathcal{A}}_n(d_1,d_0)$ and rewrite the supremum in its definition in the following way (the following holds for all $d_0,d_1>0$)
\begin{align}
&\sup_{t\in[0,T],\beta\in B_n(d_0)}\sqrt{p_n(t)}\left|\frac{1}{mp_n(t)}\overline{\Psi}_n(t;\beta)-\mu_n(t;\beta)\right| \nonumber \\
\leq&\sup_{t\in[0,T],\beta\in B_n(d_0)}\left|\frac{1}{m\sqrt{p_n(t)}}\overline{\Psi}_n(t;\beta)-\frac{1}{m\sqrt{p_n(\pi_{n,k}(t))}}\overline{\Psi}_n(\pi_{n,k}(t,\beta))\right| \label{eq:cont1} \\
&\quad+\sup_{t\in[0,T],\beta\in B_n(d_0)}\left|\sqrt{p_n(\pi_{n,k}(t))}\mu_n(\pi_{n,k}(t,\beta))-\sqrt{p_n(t)}\mu_n(t;\beta)\right| \label{eq:cont2} \\
&\quad+\sup_{t\in[0,T],\beta\in B_n(d_0)}\left|\frac{1}{m\sqrt{p_n(\pi_{n,k}(t))}}\overline{\Psi}_n(\pi_{n,k}(t,\beta))-\sqrt{p_n(\pi_{n,k}(t))}\mu_n(\pi_{n,k}(t,\beta))\right|. \label{eq:finsup}
\end{align}
Let now $d_0>0$ arbitrary but fixed. We note firstly that, by H\"{o}lder continuity of $p_n$ and $\mu_n$ and the choice of $k=k_0$, we get that \eqref{eq:cont2} converges to zero with the prescribed rate for any $d_1>0$. For \eqref{eq:cont1}, we note that
\begin{align*}
&\left|\frac{1}{m\sqrt{p_n(t)}}\overline{\Psi}_n(t;\beta)-\frac{1}{m\sqrt{p_n(\pi_{n,k}(t))}}\overline{\Psi}_n(\pi_{n,k}(t,\beta))\right| \\
=&\frac{1}{m}\left|\overline{\Psi}_n(t;\beta)\right|\cdot\left|\frac{1}{\sqrt{p_n(t)}}-\frac{1}{\sqrt{p_n(\pi_{n,k}(t))}}\right|+\frac{1}{m\sqrt{p_n(\pi_{n,k}(t))}}\left|\overline{\Psi}_n(t;\beta)-\overline{\Psi}_n(\pi_{n,k}(t,\beta))\right|.
\end{align*}
The first part converges to zero when $k=k_0$ by the assumption of H\"{o}lder continuity of $p_n$ and since
$$\frac{1}{m}\overline{\Psi}_n\leq\sup_{\beta\in B_n(d_0)}\|\Psi(\cdot;\beta)\|_{\infty}$$
is bounded by assumption. For the second part we note that for general $t_1,t_2\in[0,T]$ and $\beta_1,\beta_2\in B_n(d_0)$ (below $\#\textrm{jumps}(C_n)[t_1,t_2]$ denotes the number of jumps of the process $C_n$ in the interval $[t_1,t_2]$)
\begin{align*}
&|\overline{\Psi}_n(t_1;\beta_1)-\overline{\Psi}_n(t_2,\beta_2)| \\
\leq&\sum_{(i,j)\in V_n\times V_n}\Bigg[\left|C_{n,ij}(t_1)-C_{n,ij}(t_2)\right|\Psi(X_{n,ij}(t_1),\beta_1) \\
&\quad+C_{n,ij}(t_2)\left|\Psi(X_{n,ij}(t_1),\beta_1)-\Psi(X_{n,ij}(t_1),\beta_2)+\Psi(X_{n,ij}(t_1),\beta_2)-\Psi(X_{n,ij}(t_2),\beta_2)\right|\Bigg] \\
\leq&\left\|\Psi(\cdot;\beta_1)\right\|_{\infty}\cdot\#\textrm{jumps}(C_n)[t_1,t_2]+mL_{\Psi}\|\beta_1-\beta_2\| \\
&\quad+\sum_{(i,j)\in V_n\times V_n}L_{\Psi}(\beta_2)\tilde{\omega}_{\Psi}\left(\left\|X_{n,ij}(t_1)-X_{n,ij}(t_2)\right\|\right).
\end{align*}
Hence we obtain
\begin{align*}
&\sup_{t\in [0,T],\beta\in B(d_0)}\frac{1}{m\sqrt{p_n(\pi_{n,k}(t))}}\left|\overline{\Psi}_n(t;\beta)-\overline{\Psi}_n(\pi_{n,k}(t,\beta))\right| \\
\leq&\sup_{\beta\in B_n(d_0)}\frac{\left\|\Psi(\cdot;\beta_1)\right\|_{\infty}\cdot\#\textrm{jumps}(C_n)[n^{-k}]}{m\sqrt{p_n}}+\frac{L_{\Psi}n^{-k}}{\sqrt{p_n}} \\
&\quad+\sup_{\beta\in B_n(d_0)}L_{\Psi}(\beta)\cdot\sup_{t_1,t_2:|t_1-t_2|\leq n^{-k}}\sum_{(i,j)\in V_n\times V_n}\frac{\tilde{\omega}_{\Psi}\left(\left\|X_{n,ij}(t_1)-X_{n,ij}(t_2)\right\|\right)}{m\sqrt{p_n}}.
\end{align*}
By assumption, we may choose for $k=k_0$ the number $d_1>0$ large enough to obtain that the above exceeds $d_1\sqrt{\log m/m}$ with probability at most $\delta$. The supremum \eqref{eq:finsup} can be handled by the exponential inequality in Lemma 2.16 in \citet{K20} (we can apply this Lemma by Assumption (Mixing)). We have that for $d_1$ (after possibly increasing it further) chosen large enough (note that $\log mp_n(t)\leq\log m$)
\begin{align*}
&\IP\left(\underset{\beta\in B_n(d_0)}{\sup_{t\in[0,T],}}\sqrt{p_n(\pi_{n,k}(t))}\left|\frac{1}{mp_n(\pi_{n,k}(t,\beta))}\overline{\Psi}_n(\pi_{n,k}(t))-\mu_n(\pi_{n,k}(t,\beta))\right|>d_1\sqrt{\frac{\log m}{m}}\right) \\
\leq&\# T_{n,k}\sup_{t\in[0,T],\beta\in B_n(d_0)}\IP\left(\left|\frac{1}{mp_n(t)}\overline{\Psi}_n(t,\beta)-\mu_n(t;\beta)\right|>d_1\sqrt{\frac{\log m}{mp_n(t)}}\right)\to0.
\end{align*}
So we have proved the convergence rate of the process $\overline{\Psi}_n$ itself. For the convergence of the reciprocals denote
$$\eta_0:=\inf_{\beta\in B_n(d_0)}\left\|\mu_n(\cdot;\beta)\right\|.$$
We make the observation that on $\tilde{\mathcal{A}}_n(d_1,d_0)$
\begin{align*}
&\sqrt{p_n(s)}\left|\left(\frac{1}{mp_n(s)}\overline{\Psi}_n(s;\beta)\right)^{-1}-\mu_n(s;\beta)^{-1}\right| \\
=&\left(\frac{1}{mp_n(s)}\overline{\Psi}_n(s;\beta)\right)^{-1}\mu_n(s;\beta)^{-1}\sqrt{p_n(s)}\left|\frac{1}{mp_n(s)}\overline{\Psi}_n(s;\beta)-\mu_n(s;\beta)\right| \\
\leq&\left(\eta_0-d_1\sqrt{\frac{\log m}{mp_n(s)}}\right)^{-1}\eta_0^{-1}\sqrt{p_n(s)}\left|\frac{1}{mp_n(s)}\overline{\Psi}_n(s;\beta)-\mu_n(s;\beta)\right| \\
\leq&\left(\eta_0-d_1\sqrt{\frac{\log m}{mp_n(s)}}\right)^{-1}\eta_0^{-1}d_1\sqrt{\frac{\log m}{m}}
\end{align*}
which implies the statement because $\eta_0^{-1}=O(1)$.
\end{proof}

\section{Discussion of Assumptions}
\label{supp:discussion_assumptions}
\subsection{Motivation for Assumption (mDep)}
\label{app:mDep_mot}
In order to motivate why (mDep) is necessary, we consider the following. Let $\epsilon_1,...,\epsilon_n$ be iid random variables with finite variance and zero expectation and let $X_{n,i}$ be further random variables. We are looking for conditions such that
\begin{equation}
\label{eq:mDep_mot}
\frac{1}{\sqrt{n}}\sum_{i=1}^nX_{n,i}\epsilon_i=o_P(1).
\end{equation}
Suppose that $(X_{n,i},\epsilon_i)_{i=1,...,n}$ is an iid sequence with $\IE(\epsilon_1|X_{n,1})=0$ and $\IE(\epsilon_1^2|X_{n,1})=\sigma^2$. Then, $\IE(X_{n,i}X_{n,j}\epsilon_i\epsilon_j)=\IE(X_{n,i}X_{n,j}\epsilon_i\IE(\epsilon_j|X_{n,j}))=0$ when $i\neq j$. Hence,
$$\IE\left(\left(\frac{1}{\sqrt{n}}\sum_{i=1}^n\epsilon_iX_{n,i}\right)^2\right)=\frac{1}{n}\sum_{i,j=1}^n\IE(X_{n,i}X_{n,j}\epsilon_i\epsilon_j)=\IE(X_{n,1}^2)\sigma^2.$$
Thus, $\IE(X_{n,1}^2)\to0$ is a sufficient condition for \eqref{eq:mDep_mot}. However, this arguments relies on $\IE(\epsilon_1|X_{n,1})=0$. If this is not the case, it is difficult to find a general treatment of the above problem. Suppose for example that $X_{n,i}:=f_i(\epsilon_1,...,\epsilon_n)$ is a function of all $\epsilon_1,...,\epsilon_n$. Hence, in the following we will no longer assume that $(X_{n,i},\epsilon_i)_{i=1,...,n}$ is an iid sequence, but $(\epsilon_i)_{i=1,...,n}$ remains an iid sequence. We will take the following route to remedy this issue: Let $\tilde{\epsilon}^{ij}_k:=\epsilon_k$ for $k\neq i,j$ and $\tilde{\epsilon}^{ij}_k:=0$ for $k=i,j$. With this we put
$$\tilde{X}_{n,k}^{ij}:=f_k(\tilde{\epsilon}^{ij}_1,...,\tilde{\epsilon}^{ij}_n).$$
Then, $\IE(\tilde{X}_{n,i}^{ij}\tilde{X}_{n,j}^{ij}\epsilon_i\epsilon_j)=0$ for $i\neq j$. Hence, \eqref{eq:mDep_mot} follows if $\sup_{i=1,...,n}\IE\left(\left(\tilde{X}_i^{ii}\right)^2\epsilon_i^2\right)\to0$ and $X_{n,i}X_{n,j}-\tilde{X}_{n,i}^{ij}X_{n,j}^{ij}$ is very small in a suitable sense. Such ideas have been transferred to counting processes by \citet{MN07} and to counting processes on networks in \citet{K20}. Suppose to this end that for each $t,n$, we have $d_t^n:V_n^2\times V_n^2\to[0,\infty)$, a random distance function between pairs. For $I \subseteq V_n \times V_n$, set $d_t^n(ij,I):=\inf_{(k,l)\in I}d_t^n(ij,kl)$,  with the convention $d_t^n(ij,\emptyset):=\infty$. In order to be able to define the correct analogue to uncorrelatedness as mentioned above we introduce the following augmented $\sigma$-field for any sets $I,J\subseteq V_n\times V_n$ and any $M>0$, where we denote by $\mathcal{F}_1\wedge\mathcal{F}_2$ the $\sigma$-field generated by the union of the $\sigma$-fields $\mathcal{F}_1$ and $\mathcal{F}_2$:
\begin{align*}
&\tilde{\mathcal{F}}_{I,t}^{n,J,M}:=\sigma\left(X_{n,ij}(\tau):(i,j)\in I, \tau\leq t\right) \\
&\quad\quad\wedge\sigma\left(\mathcal{N}_{n,ij}(r)\Ind(d_s^n((i,j),J)\geq M):s\leq \max(0,t-4h), r\leq s+6h, (i,j)\in V_n\times V_n\right), \\
&\mathcal{N}_{n,ij}(r):=\left(N_{n,ij}(r),X_{n,ij}(r),C_{n,ij}(r)\right).
\end{align*} 
Intuitively speaking the above filtration contains all information of the covariates belonging to pairs in $I$ and it contains extra (future) information about pairs that are far from $J$. We illustrate below how this can be applied. Consider to this end the functions $\phi_{n,ij,kl}^I(s_1,s_2)$ as defined in \eqref{eq:def_phi} and the sentence below (there $f_n(r,s)=\int_0^ThK_{h,t}(s)K_{h,t}(r)w(t)dt$ is symmetric). The exact meaning of these functions is at this point not important and the following discussion remains correct for all functions with the following measurability and summing properties: The functions $\phi_{n,ij,kl}(s_1,s_2)$ will play the role of $X_{n,i}X_{n,j}$ in the initial example and $\phi_{n,ij,kl}^{\{ij,kl\}}(s,r)$ will be the approximation $\tilde{X}_{n,i}^{ij}\tilde{X}_{n,j}^{ij}$. Rather than multiplying with an iid sequence we will later integrate with respect to the counting process martingales $M_{n,ij}$ and $M_{n,kl}$ and consider
$$\int_0^T\int_{s_1-2h}^{s_1-}2\phi_{n,ij,kl}(s_1,s_2)dM_{n,kl}(s_2)dM_{n,ij}(s_1).$$
When replacing the integrand $\phi_{n,ij,kl}(s_1,s_2)$ by the approximation $\phi_{n,ij,kl}^{\{ij,kl\}}(s_1,s_2)$, we see that $s_1\mapsto\int_{s_1-2h}^{s_1-}2\tilde{\phi}_{n,ij,kl}^{\{ij,kl\}}(s_1,s_2)dM_{n,kl}(s_2)$ is predictable with respect to $\tilde{\mathcal{F}}_{\{ij,kl\},t}^{m,\{ij,kl\},M}$ and Assumption (mDep) makes sure that $M_{n,ij}$ is a martingale with respect to a (slightly larger) filtration. Hence, we may use martingale results in order to conclude that the expectation of the stochastic integral equals zero.

To finish the discussion, we have to argue why the approximation is good enough. Consulting the definition of $\phi_{n,ij,kl}(s_1,s_2)$ we see that it is a sum and $\phi_{n,ij,kl}^I$ is obtained by removing those pairs from the sum which are close to $I$. Hence, the quality of the approximation is closely related to the appearance of hubs, i.e., pairs which have many close neighbours. Therefore we need to balance the size of hubs (the number of their close neighbouring pairs) and the likelihood of their appearance. This is the content of the next subsection.

\subsection{Hubs}
\label{app:hubs}

In this section we will be precise about the assumptions which a network with hubs needs to fulfill. Intuitively speaking, we will require that hub-size and hub-frequency need to balance each other. In order to make this precise we consider random variables $H_{\rm hub}^{ij}$ for $i,j\in V_n$ with
$$H_{\rm hub}^{ij}\geq \Ind(K_M^{ij} \geq n_{\rm hub}).$$
We call $H_{\rm hub}^{ij}$ the hub-ability of $(i,j)$. Note that we may simply put $H_{\rm hub}^{ij}=\Ind(K_M^{ij}\geq n_{\rm hub})$ or even $H_{\rm hub}^{ij}\equiv1$. However, in this case it might be difficult to justify the measurability and other properties which we will impose. We think therefore of $H_{\rm hub}^{ij}$ more like an external variable which determines whether the pair $(i,j)$ has the potential of becoming a hub rather than the indicator which says if $(i,j)$ has actually become a hub. The assumption (HSR) below should be understood as adding further constraints to $H_{\rm hub}^{ij}$ making trivial choices possibly infeasible. We define furthermore for $A\subseteq V_n\times V_n$ the variables $H_{\rm hub}^A:=\sup_{(i,j)\in A}H_{\rm hub}^{ij}$ and $K_M:=\max_{k,l\in V_n}K_M^{kl}$. With these definitions we denote
$$B^{ij}_n := n_{\rm hub} + H^{ij}_{\rm hub}K_M.$$
Under the no-hubs assumption (NH) from the main text we have $B_n^{ij}=n_{\rm hub}$ and $H_{\rm hub}^{ij}\equiv0$. Therefore the following condition (HSR) is trivially implied by (NH). In order to write down the condition, some notation is required:
\begin{align*}
\Phi_n(t,r)&:=\int_t^{r+2h}\frac{1}{hmp_n(s)}\bar{\Psi}_n(s;\beta_0)ds, \\
E_n(t)&:=\sup_{s\in[t,t+2h]}\frac{1}{N}\sum_{i,j\in V_n}C_{n,ij}(s),\qquad E_n:=\max\left(1,\sup_{t\in[0,T]}E_n(t)\right).
\end{align*}
$\Phi_n$ is a smoothed version of the overall hazard $\overline{\Psi}_n$ and $E_n(t)$ and $E_n$ are average numbers of observed edges. For simplicity of notation we will write
\begin{equation}
\label{eq:def_Jkl}
\mathcal{J}_{kl}(t):=C_0\left(4h\sup_{\rho\in[t-2h,t+2h]}C_{n,kl}(\rho)+N_{n,kl}[t-2h,t+2h]\right)
\end{equation}
with $C_0:=\|\alpha_0\|_{\infty}\|\Psi(\cdot;\beta_0)\|_{\infty}$. Let finally $F_{ij}(t):=\{(k,l): d_t^n(ij,kl)\geq M\}$ be the set of all pairs which are \emph{far away} from $(i,j)$ at time $t$.

\textbf{Assumption (HSR)}: (Hub Size Restriction) \\
\emph{Let $B_n^{ij}$ be measurable with respect to $\mathcal{F}_0^n$ and suppose that $n_{\rm hub}\leq K_M$. Suppose furthermore that there is a deterministic sequence $\kappa_n\geq K_M$ such that
\begin{align}
&\underset{i,j,k,l\in V_n}{\sup_{s,t\in[0,T]}}\frac{\kappa_n^2}{N^2}\IP\left(H_{\rm hub}^{\{ij,kl\}}=1\big|C_{n,ij}(t)=1,C_{n,kl}(s)=1\right)\to0, \label{eq:HSR1} \\
&\underset{i,j,k,l\in V_n}{\sup_{s,t\in[0,T]}}\kappa_n^3\IP\left(H_{\rm hub}^{ij}=1\big|C_{n,ij}(t)=1,C_{n,kl}(s)=1\right)=O(1), \label{eq:HSR2} \\
&\underset{i,j,k,l\in V_n}{\sup_{s,t\in[0,T]}}\kappa_n^6\IP\left(H_{\rm hub}^{ij}=1,H_{\rm hub}^{kl}=1\big|C_{n,ij}(t)=1,C_{n,kl}(s)=1\right)=O(1), \label{eq:HSR3} \\
&\left(\kappa_n^2+\frac{\kappa_n^4}{N}\right)\IE\left(\left(\sum_{k,l\in V_n}\int_0^T\frac{H_{\rm hub}^{kl}}{mp_n(s)}d|M_{n,kl}|(s)\right)^2\right)=O(1), \label{eq:HSR4} \\
&\underset{s\in[0,T]}{\sup_{i,j\in V_n}}\kappa_n^2\IE\left(\left(\underset{(k,l)\neq(i,j)}{\sum_{k,l\in V_n}}\int_{s-2h}^{s-}\frac{H_{\rm hub}^{\{ij,kl\}}}{mp_n(t)}d|M_{n,kl}|(t)\right)^2\Big|C_{n,ij}(s)=1\right)\to0, \label{eq:HSR5} \\
&\frac{\kappa_n^3}{m^4p_n^4}\underset{ij\neq i'j'}{\sum_{ij,i'j'}}\IE\Bigg[\sup_{s\in[0,T]}\sum_{k,l\in V_n}C_{n,kl}(s)\int_0^T\int_{t-2h}^{t-}\sum_{k',l'\in V_n}H_{\rm hub}^{k'l'}\mathcal{J}_{k'l'}(t)  \nonumber \\
&\qquad\times\sup_{(a,b)\notin F_{ij}(t-2h)}H_{\rm hub}^{ab}\mathcal{J}_{ab}(t)d|M_{n,i'j'}|(r)d|M_{n,ij}|(t)\Bigg]\to0. \label{eq:HSR6} \\
&\kappa_n^2\int_0^T\IE\left(\left(\sum_{k,l\in V_n}\int_{t-2h}^{t-}\frac{H_{\rm hub}^{kl}}{mp_n(r)}\Phi_n(t,r)d|M_{n,kl}|(r)\right)^2\right)dt\to0, \label{eq:HSR7} \\
&\kappa_n^2\underset{s\in[0,T]}{\sup_{i,j\in V_n}}\IE\left(\left(\underset{(k',l')\neq(k,l)}{\sum_{k,l,k',l'\in V_n}}\frac{H_{\rm hub}^{\{kl,k'l'\}}}{N^2}\int_0^T\int_{r-2h}^{r-}d|M_{n,k'l'}|(t)d|M_{n,kl}|(r)\right)^2\Big|C_{n,ij}(s)=1\right)=O(1), \label{eq:HSR8} \\
&\kappa_n^2\sup_{i,j,k,l\in V_n}\IE\Big(E_n^2\left(H_{\rm hub}^{ij}+H_{\rm hub}^{kl}+H_{\rm hub}^{ij}H_{\rm hub}^{kl}\right) \nonumber \\
&\qquad\left(1+N_{n,ij}[0,T]+N_{n,ij}[0,T]N_{n,kl}[0,T]\right)\big|\tilde{C}_{n,ij}=1,\tilde{C}_{n,kl}=1\Big)=O(1) \label{eq:HSR9}
\end{align}}

The above assumptions relate the frequency of hub appearance $\IP(H_{\rm hub}^{ij}=1)$ and the hub-size $\kappa_n$. In all assumptions we see a trade-off stating that larger hubs, i.e., large sequences $\kappa_n$, are permitted if $H_{\rm hub}^{ij}=0$ sufficiently often. The Assumption (NH) from the main text may be replaced by (HSR).

\subsection{Low Level Assumptions}
\label{app:lla}
In the main text we presented the Assumptions (WC), (BM) and (NH). In this section we present a set of assumptions which is in conjunction with (HSR) from Section \ref{app:hubs} sufficient for proving the main result. Hence, the Assumptions (WC), (BM), (NH) may be replaced by (HSR) and the assumptions from this section. We will state in the following the assumptions as equations and then prove that they follow from (WC) and (BM). Thus, (HSR), (WC), (BM) and all other assumptions from Section \ref{subsec:assumptions} with the exception of (NH) are supposed to hold true in this section.

Recall the notation $\tilde{C}_{n,ij}:=\sup_{t\in[0,T]}C_{n,ij}(t)$ and $\tilde{p}_n(ij,kl):=\IP(\tilde{C}_{n,ij}=1,\tilde{C}_{n,kl}=1)$.

\begin{equation}
\label{eq:wc1}
\IE\left(\left(\frac{1}{N}\sum_{i,j\in V_n}\int_0^TB_n^{ij}E_n(t)d|M_{n,ij}|(t)\right)^2\right)=O(1).
\end{equation}
\begin{proof}
In order to show \eqref{eq:wc1}, we compute by using that (VX) implies that there are only three different summands
\begin{align}
\eqref{eq:wc1}=&\frac{1}{N^2}\sum_{i,j,k,l\in V_n}\IE\left(\int_0^TB_n^{ij}E_n(t)d|M_{n,ij}|(t)\int_0^TB_n^{kl}E_n(t)d|M_{n,kl}|(t)\right) \nonumber \\
=&\frac{1}{N^2}O(n^2)\IE\left(\left(\int_0^TB_n^{12}E_n(t)d|M_{n,12}|(t)\right)^2\right) \label{eq:wc11} \\
&+\frac{1}{N^2}O(n^3)\IE\left(\int_0^TB_n^{12}E_n(t)d|M_{n,12}|(t)\int_0^TB_n^{23}E_n(t)d|M_{n,23}|(t)\right) \label{eq:wc12} \\
&+\frac{1}{N^2}O(n^4)\IE\left(\int_0^TB_n^{12}E_n(t)d|M_{n,12}|(t)\int_0^TB_n^{34}E_n(t)d|M_{n,34}|(t)\right). \label{eq:wc13}
\end{align}
It holds that
$$\int_0^TB_n^{ij}E_n(t)d|M_{n,ij}|(t)=\tilde{C}_{n,ij}\int_0^TB_n^{ij}E_n(t)d|M_{n,ij}|(t)$$
because $\tilde{C}_{n,ij}=0$ implies that the integral equals zero as well. Let
$$S:=T\|\alpha_0\|_{\infty}\|\Psi(\cdot;\beta_0)\|_{\infty}.$$
We get
\begin{align*}
&\IE\left(\int_0^TB_n^{ij}E_n(t)d|M_{n,ij}|(t)\int_0^TB_n^{kl}E_n(t)d|M_{n,kl}|(t)\right) \\    
=&\tilde{p}_n(ij,kl)\IE\left(\int_0^TB_n^{ij}E_n(t)d|M_{n,ij}|(t)\int_0^TB_n^{kl}E_n(t)d|M_{n,kl}|(t)\Big|\tilde{C}_{n,ij}=1,\tilde{C}_{n,kl}=1\right) \\
\leq&\tilde{p}_n(ij,kl)\IE\Bigg(B_n^{ij}B_n^{kl}E_n^2\big(N_{n,kl}[0,T]N_{n,ij}[0,T] \\
&\quad+S\left(S+N_{n,ij}[0,T]+N_{n,kl}[0,T]\right)\big)\Big|\tilde{C}_{n,ij}=1,\tilde{C}_{n,kl}=1\Bigg) \\
=&\tilde{p}_n(ij,kl)\Bigg(\IE\big(B_n^{ij}B_n^{kl}E_n^2N_{n,kl}[0,T]N_{n,ij}[0,T]\big|\tilde{C}_{n,ij}=1,\tilde{C}_{n,kl}=1\big)\\
&\qquad+2S\IE\big(B_n^{ij}B_n^{kl}E_n^2N_{n,ij}[0,T]\big|\tilde{C}_{n,ij}=1,\tilde{C}_{n,kl}=1\big) \\
&\qquad+S^2\IE\big(B_n^{ij}B_n^{kl}E_n^2\big|\tilde{C}_{n,ij}=1,\tilde{C}_{n,kl}=1\big)\Bigg).
\end{align*}
We argue now that we can find a constant $K_0>0$ such that the above is bounded by $K_0\tilde{p}_n(ij,kl)$. Firstly, by replacing the definition of $B_n^{ij}$, we get from Assumption (HSR, \ref{eq:HSR9}) that we may ignore all $B_n^{ij}$ and $B_n^{kl}$. Then we apply the Cauchy Schwarz Inequality and bound the terms $\IE\big(N_{n,kl}[0,T]^2N_{n,ij}[0,T]^2\big|\tilde{C}_{n,ij}=1,\tilde{C}_{n,kl}=1\big)$ and $\IE\big(N_{n,ij}[0,T]^2\big|\tilde{C}_{n,ij}=1,\tilde{C}_{n,kl}=1\big)$ uniformly in $i,j,k,l$ by (BM, \ref{eq:BM1}). The term $\IE(E_n^4|\tilde{C}_{n,ij}=1,\tilde{C}_{n,kl}=1)$ is bounded as follows by (WC, \ref{eq:swc4})
\begin{align*}
&\sup_{k,l,k',l'}\IE\left(E_n^4\big|\tilde{C}_{n,kl}=1,\tilde{C}_{n,k'l'}=1\right)\leq\sup_{k,l,k',l'}\IE\left(\left(\frac{1}{N}\sum_{i,j\in V_n}\tilde{C}_{n,ij}\right)^2\bigg|\tilde{C}_{n,kl}=1,\tilde{C}_{n,k'l'}=1\right) \\
=&\sup_{k,l,k',l'}\frac{1}{N^4}\underset{a=1,...,4}{\sum_{i_a,j_a\in V_n}}\IP\left(\prod_{a=1}^4\tilde{C}_{n,i_aj_a}=1\Big|\tilde{C}_{n,kl}=1,\tilde{C}_{n,k'l'}=1\right)=O(1).
\end{align*}
By (SP) we have that $mp_n/N=O(1)$. Therefore, we may replace $N$ by $mp_n=p_nn(n-1)/2$ and so we get by using the above bound and (WC, \ref{eq:swc1}) (below $\approx$ means equality up to a constant)
\begin{align*}
\eqref{eq:wc11}\approx&\frac{n^2}{n^4p_n^2}\tilde{p}_n(12,12)=\frac{\tilde{p}_n(12,12)}{n^2p_n^2}=O(1), \\
\eqref{eq:wc12}\approx&\frac{n^3}{n^4p_n^2}\tilde{p}_n(12,23)=\frac{\tilde{p}_n(12,23)}{np_n^2}=O(1), \\
\eqref{eq:wc13}\approx&\frac{n^4}{n^4p_n^2}\tilde{p}_n(12,34)=\frac{\tilde{p}_n(12,34)}{p_n^2}=O(1) 
\end{align*}
by the assumptions in (WC). Thus \eqref{eq:wc1} follows.
\end{proof}

\begin{equation}
\label{eq:wc2}
\IE\left(\left(\frac{1}{N}\sum_{i,j\in V_n}\int_0^Td|M_{n,ij}|(s)\right)^2\right)=O(1).
\end{equation}
\begin{proof}
\eqref{eq:wc2} follows by the same arguments as for \eqref{eq:wc1}. In fact, we simply have to remove all appearances of $E_n$ and $B_n^{ij}$ in the previous derivation. We can then use the same argument because, by definition, $E_n\geq1$.
\end{proof}

\begin{equation}
\label{eq:wc3}
\sum_{i,j,k,l\in V_n}\int_0^T\int_0^T\frac{\IP(C_{n,ij}(t)=1,\,C_{n,kl}(s)=1)}{m^2p_n(t)p_n(s)}dtds=O(1).
\end{equation}
\begin{proof}
The proof of \eqref{eq:wc3} is straight forward using the exchangeability from Assumption (VX) and the fact that
$$\IP(C_{n,ij}(t)=1,\,C_{n,kl}(s)=1)\leq\IP(\tilde{C}_{n,ij}=1,\,\tilde{C}_{n,kl}(s)=1)$$
for all $t,s\in[0,T]$. Hence, ($\approx$ means again equality up to a constant)
\begin{align*}
\eqref{eq:wc3}\leq&T^2\sum_{i,j,k,l\in V_n}\frac{\tilde{p}_n(ij,kl)}{m^2p_n^2}\approx\frac{\tilde{p}_n(12,12)}{n^2p_n^2}+\frac{\tilde{p}_n(12,23)}{np_n^2}+\frac{\tilde{p}_n(12,34)}{p_n^2}. \\
\end{align*}
\eqref{eq:wc3} follows now from (WC, \ref{eq:swc1}).
\end{proof}

\begin{equation}
\label{eq:wc4}
\sup_{i,j\in V_n,\,s\in[0,T]}\IE\left(\left(\frac{1}{N}\underset{(k,l)\neq(i,j)}{\sum_{k,l\in V_n}}\int_{s-2h}^{s-}d|M_{n,kl}|(t)\right)^2\Bigg|C_{n,ij}(s)=1\right)\to0.
\end{equation}
\begin{proof}
In order to prove \eqref{eq:wc4} we follow a similar strategy as for \eqref{eq:wc1}. We begin by bounding the conditional expectation
\begin{align}
&\IE\left(\int_{s-2h}^{s-}d|M_{n,k_1l_1}|(t)\int_{s-2h}^{s-}d|M_{n,k_2l_2}|(t)\Big|C_{n,ij}(s)=1,\tilde{C}_{n,k_1l_1}=1,\tilde{C}_{n,k_2l_2}=1\right) \nonumber \\
\leq&\IE\Big(\left(N_{n,k_1l_1}[s-2h,s]+2h\|\alpha_0\|_{\infty}\|\Psi(\cdot;\beta_0)\|_{\infty}\right) \nonumber \\
&\qquad\times\left(N_{n,k_2l_2}[s-2h,s]+2h\|\alpha_0\|_{\infty}\|\Psi(\cdot;\beta_0)\|_{\infty}\right)\Big|C_{n,ij}(s)\tilde{C}_{n,k_1l_1}\tilde{C}_{n,k_2l_2}=1\Big) \nonumber \\
\leq&\epsilon_n \label{eq:wc41}
\end{align}
for a suitable sequence $\epsilon_n\to0$ which exists by Assumptions (BM, \ref{eq:BM2}, \ref{eq:BM3}), (C) and $h\to0$ by (KBW). Introduce the notation $\tilde{p}_n(k_1l_1,k_2l_2|ij,s):=\IP(\tilde{C}_{n,k_1l_1}\tilde{C}_{n,k_2l_2}=1|C_{n,ij}(s)=1)$. Note that $\tilde{p}_n(k_1l_1,k_2l_2|ij,s)$ depends only on $|\{k_1,l_1\}\cap\{k_2,l_2\}|$ by the exchangeability assumption in (VX) as long as $(k_1,l_1),(k_2,l_2)\neq(i,j)$. We can now bound \eqref{eq:wc4} as follows by using \eqref{eq:wc41} (we use $\approx$ to indicate equality up to a constant):
\begin{align*}
\eqref{eq:wc4}=&\sup_{i,j\in V_n,\,s\in[0,T]}\frac{1}{N^2}\underset{(k_1,l_1),(k_2,l_2)\neq(i,j)}{\sum_{k_1,l_1,k_2,l_2\in V_n}}\IP(\tilde{C}_{n,k_1l_1}\tilde{C}_{n,k_2l_2}=1|C_{n,ij}(s)=1) \\
&\qquad\times\IE\left(\int_{s-2h}^{s-}d|M_{n,k_1l_1}|(t)\int_{s-2h}^{s-}d|M_{n,k_2l_2}|(t)\Bigg|C_{n,ij}(s)\tilde{C}_{n,k_1l_1}\tilde{C}_{n,k_2l_2}=1\right) \\
\leq&\sup_{i,j\in V_n,\,s\in[0,T]}\frac{1}{N^2}\underset{(k_1,l_1),(k_2,l_2)\neq(i,j)}{\sum_{k_1,l_1,k_2,l_2\in V_n}}\IP(\tilde{C}_{n,k_1l_1}\tilde{C}_{n,k_2l_2}=1|C_{n,ij}(s)=1)\epsilon_n \\
\leq&\sup_{i,j\in V_n,\,s\in[0,T]}\IE\left(\left(\frac{1}{N}\sum_{k,l\in V_n}\tilde{C}_{n,k_1l_1}\tilde{C}_{n,k_2l_2}\right)^2\bigg|C_{n,ij}(s)=1\right)\epsilon_n\to0
\end{align*}
since $\epsilon_n\to0$ and since the expectation is bounded by Assumption (WC, \ref{eq:swc3}).
\end{proof}

\begin{equation}
\label{eq:wc5}
\underset{s\in[0,T]}{\sup_{i,j\in V_n,}}\IE\left(\left(\frac{1}{N^2}\underset{(k',l')\neq(k,l)}{\sum_{k,l,k',l'\in V_n}}\int_0^T\int_{r-2h}^{r-}d|M_{n,k'l'}|(t)d|M_{n,kl}|(r)\right)^2\Bigg|C_{n,ij}(s)=1\right)=O(1).
\end{equation}
\begin{proof}
We can show \eqref{eq:wc5} by using relatively strong but simple bounds
\begin{align*}
&\eqref{eq:wc5} \\
\leq&\underset{s\in[0,T]}{\sup_{i,j\in V_n,}}\IE\left(\left(\frac{1}{N^2}\left(\sum_{k,l\in V_n}\int_0^Td|M_{n,kl}|(r)\tilde{C}_{n,kl}\right)^2\right)^2\Bigg|C_{n,ij}(s)=1\right) \\
\leq&\underset{s\in[0,T]}{\sup_{i,j\in V_n,}}\frac{1}{N^4}\underset{a=1,...,4}{\sum_{k_a,l_a,\in V_n}}\IE\left(\prod_{a=1}^4\left(N_{n,k_al_a}[0,T]+T\|\alpha_0\|_{\infty}\|\Psi(\cdot;\beta_0)\|_{\infty}\right)\tilde{C}_{n,k_al_a}\Bigg|C_{n,ij}(s)=1\right) \\
\leq&\underset{s\in[0,T]}{\sup_{i,j\in V_n,}}\sup_{k_1,l_1,...,k_4,l_4\in V_n}\IE\left(\prod_{a=1}^4\left(N_{n,k_al_a}[0,T]+T\|\alpha_0\|_{\infty}\|\Psi(\cdot;\beta_0)\|_{\infty}\right)\Bigg|\prod_{a=1}^4\tilde{C}_{n,k_al_a}C_{n,ij}(s)=1\right) \\
&\quad\times\underset{s\in[0,T]}{\sup_{i,j\in V_n,}}\frac{1}{N^4}\underset{a=1,...,4}{\sum_{k_a,l_a,\in V_n}}\IP\left(\prod_{a=1}^4\tilde{C}_{n,k_al_a}=1\Big|C_{n,ij}(s)=1\right).
\end{align*}
The above is bounded by Assumptions (C), (WC, \ref{eq:swc3}) and (BM, \ref{eq:BM4}).
\end{proof}

\begin{equation}
\label{eq:wc7}
\sup_{t,s\in[0,T]}\IE\left(\Phi_n(t,s)^2\frac{1}{m^2p_n(s)^2}\underset{(i,j)\neq(k,l)}{\sum_{i,j,k,l\in V_n}}C_{n,ij}(t)C_{n,kl}(s)\right)=O(1)
\end{equation}
\begin{proof}
In order to see that \eqref{eq:wc7} holds, we firstly note that by definition of $N$ and Assumption (SP) there is a constant $c>0$ such that $1/mp_n\leq c/N$. We can then show that
\begin{align}
&\underset{i,j,k,l\in V_n}{\sup_{t\in[0,T], s\in[t-2h,t],}}\IE\left(\Phi_n(t,s)^2\big|\tilde{C}_{n,ij}=1,\tilde{C}_{n,kl}=1\right) \nonumber \\
=&\underset{s\in[t-2h,t],}{\underset{i,j,k,l\in V_n}{\sup_{t\in[0,T],}}}\IE\left(\left(\int_t^{s+2h}\frac{1}{hmp_n(r)}\sum_{a,b\in V_n}C_{n,ab}(r)\Psi(X_{n,ab}(r);\beta_0)dr\right)^2\Big|\tilde{C}_{n,ij}=1,\tilde{C}_{n,kl}=1\right) \nonumber \\
\leq&4c^2\|\Psi(\cdot;\beta_0)\|_{\infty}^2\sup_{i,j,k,l\in V_n}\IE\left(\left(\frac{1}{N}\sum_{a,b\in V_n}\tilde{C}_{n,ab}\right)^2\Big|\tilde{C}_{n,ij}=1,\tilde{C}_{n,kl}=1\right) \nonumber \\
\leq&4c^2\|\Psi(\cdot;\beta_0)\|_{\infty}^2\sup_{i,j,k,l\in V_n}\IE\left(\left(\frac{1}{N}\sum_{a,b\in V_n}\tilde{C}_{n,ab}\right)^4\Big|\tilde{C}_{n,ij}=1,\tilde{C}_{n,kl}=1\right)^{\frac{1}{2}} \nonumber \\
=&4c^2\|\Psi(\cdot;\beta_0)\|_{\infty}^2\sup_{k,l,k',l'\in V_n}\left(\frac{1}{N^4}\underset{a=1,...,4}{\sum_{i_a,j_a\in V_n}}\IP\left(\prod_{a=1}^4\tilde{C}_{n,i_aj_a}=1\Big|\tilde{C}_{n,kl}=1,\tilde{C}_{n,k'l'}=1\right)\right)^{\frac{1}{2}} \label{eq:wc1+2}
\end{align}
which is bounded by Assumptions (C) and (WC, \ref{eq:swc4}). We can also compute
\begin{align*}
\eqref{eq:wc7}\leq&\sup_{t,s\in[0,T]}\frac{1}{m^2p_n^2}\sum_{i,j,k,l\in V_n}\IE\left(\Phi_n(t,s)^2\Big|\tilde{C}_{n,ij}=1,\tilde{C}_{n,kl}=1\right)\IP(\tilde{C}_{n,ij}=1,\tilde{C}_{n,kl}=1).
\end{align*}
The expectation is uniformly bounded by \eqref{eq:wc1+2} and the sum over the probabilities remains bounded by assumption (WC, \ref{eq:swc1}).
\end{proof}

\begin{equation}
\label{eq:wc8}
\sup_{t\in[0,T]}\IE\left(\left(\frac{1}{N}\overline{\Psi}_n(t,\beta_0)\right)^2\right)=O(1).
\end{equation}
\begin{proof}
We can show \eqref{eq:wc8} by simple bounds:
\begin{align*}
\eqref{eq:wc8}\leq&\IE\left(\left(\frac{1}{N}\sum_{i,j\in V_n}\tilde{C}_{n,ij}\right)^2\right)\|\Psi(\cdot;\beta_0)\|_{\infty}^2
\end{align*}
which is bounded by the assumptions (C), (SP) and (WC, \ref{eq:swc1}).
\end{proof}

\begin{equation}
\label{eq:wc9}
\int_0^T\IE\left(\left(\frac{1}{N}\sum_{k,l\in V_n}\int_{t-2h}^{t-}\Phi_n(t,r)d|M_{n,kl}|(r)\right)^2\right)dt\to0.
\end{equation}
\begin{proof}
In order to prove \eqref{eq:wc9}, we firstly bound $\Phi_n(t,r)$
\begin{align}
\Phi_n(t,r)=&\int_t^{r+2h}\frac{1}{hmp_n(s)}\sum_{i,j\in V_n}C_{n,ij}(s)\Psi(X_{n,ij}(s);\beta_0)ds \nonumber \\
\leq&\|\Psi(\cdot;\beta_0)\|_{\infty}\frac{r+2h-t}{h}\cdot\frac{N}{mp_n}\cdot\frac{1}{N}\sum_{i,j\in V_n}\tilde{C}_{n,ij}. \label{eq:phi_bound1}
\end{align}
By definition of $N$ and Assumption (SP) we may bound $N/mp_n$ from above by a constant $\alpha_2$. Using this bound, we obtain with $R:=\|\alpha_0\|_{\infty}\|\Psi(\cdot;\beta_0)\|_{\infty}$
\begin{align*}
\eqref{eq:wc9}\leq&\int_0^T\IE\left(\left(\frac{1}{N}\sum_{k,l\in V_n}\int_{t-2h}^{t-}\|\Psi(\cdot;\beta_0)\|_{\infty}\frac{r+2h-t}{h}\cdot\alpha_2\cdot\frac{1}{N}\sum_{i,j\in V_n}\tilde{C}_{n,ij}d|M_{n,kl}|(r)\right)^2\right)dt \\
\leq&4\|\Psi(\cdot;\beta_0)\|_{\infty}^2\alpha_2^2\int_0^T\IE\left(\left(\frac{1}{N^2}\sum_{i,j,k,l\in V_n}\int_{t-2h}^{t-}d|M_{n,kl}|(r)\tilde{C}_{n,ij}\tilde{C}_{n,kl}\right)^2\right)dt \\
=&\frac{4\|\Psi(\cdot;\beta_0)\|_{\infty}^2\alpha_2^2}{N^4} \\
&\times\underset{a=1,...,4}{\sum_{i_a,j_a\in V_n}}\int_0^T\IE\Big((N_{n,i_2j_2}[t-2h,t]+2hR)(N_{n,i_4j_4}[t-2h,t]+2hR) \\
&\qquad\qquad\times\tilde{C}_{n,i_1j_1}\tilde{C}_{n,i_2j_2}\tilde{C}_{n,i_3j_3}\tilde{C}_{n,i_4j_4}\Big)dt \\
\leq&\underset{a=1,...,4}{\sup_{i_a,j_a\in V_n}}\int_0^T\IE\left((N_{n,i_2j_2}[t-2h,t]+2hR)(N_{n,i_4j_4}[t-2h,t]+2hR)\Big|\prod_{a=1}^4\tilde{C}_{n,i_aj_a}=1\right)dt \\
&\times\frac{4\|\Psi(\cdot;\beta_0)\|_{\infty}^2\alpha_2^2}{N^4}\underset{a=1,...,4}{\sum_{i_a,j_a\in V_n}}\IP\left(\prod_{a=1}^4\tilde{C}_{n,i_aj_a}=1\right)
\end{align*}
The first line converges to zero by (BM, \ref{eq:BM5}) and the second line is bounded by (WC, \ref{eq:swc1}).
\end{proof}

\begin{equation}
\label{eq:wc1+1}
\int_0^T\IE\left(E_n(t)^2\Big|C_{n,12}(t)=1\right)+\IE\left(E_n(t)^3\right)dt=O(1).
\end{equation}
\begin{proof}
Assumption \eqref{eq:wc1+1} can be proven as follows:
\begin{align*}
\eqref{eq:wc1+1}\leq&\int_0^T\IE\left(\left(\frac{1}{N}\sum_{i,j\in V_n}\tilde{C}_{n,ij}\right)^2\Big|C_{n,12}(t)=1\right)+\IE\left(\left(\frac{1}{N}\sum_{i,j\in V_n}\tilde{C}_{n,ij}\right)^3\right)dt \\
\leq&T\sup_{t\in[0,T]}\IE\left(\left(\frac{1}{N}\sum_{i,j\in V_n}\tilde{C}_{n,ij}\right)^4\Big|C_{n,12}(t)=1\right)^{\frac{1}{2}}+T\IE\left(\left(\frac{1}{N}\sum_{i,j\in V_n}\tilde{C}_{n,ij}\right)^4\right)^{\frac{3}{4}}
\end{align*}
both of which remain bounded by Assumption (WC, \ref{eq:swc2}, \ref{eq:swc3}).
\end{proof}

\begin{equation}
\label{eq:wc1+3}
\underset{i,j,k,l\in V_n}{\sup_{t,s\in[0,T],}}\sup_{r\in[t,t+2h]}\IE\left(\left(\frac{1}{N}\overline{\Psi}_n(r;\beta_0)\right)^3\big|C_{n,ij}(t)C_{n,kl}(s)=1\right)=O(1).
\end{equation}
\begin{proof}
We have
\begin{align*}
\eqref{eq:wc1+3}=&\underset{k,l,l'l'\in V_n}{\sup_{t,s\in[0,T],}}\sup_{r\in[t,t+2h]}\IE\left(\left(\frac{1}{N}\sum_{i,j\in V_n}C_{n,ij}(r)\Psi(X_{n,ij}(r);\beta_0)\right)^3\big|C_{n,kl}(t)C_{n,k'l'}(s)=1\right) \\
\leq&\|\Psi\cdot;\beta_0)\|_{\infty}^3\underset{k,l,l'l'\in V_n}{\sup_{t,s\in[0,T],}}\frac{1}{N^3}\underset{a=1,...,3}{\sum_{i_a,j_a\in V_n}}\IP\left(\prod_{a=1}^3\tilde{C}_{n,i_aj_a}=1\big|C_{n,kl}(t)=1,C_{n,k'l'}(s)=1\right)
\end{align*}
which is bounded by Assumptions (C) and and (WC, \ref{eq:swc5}).
\end{proof}

Recall the definition of $\mathcal{J}_{ab}$ from \eqref{eq:def_Jkl}.
\begin{align}
&\frac{1}{m^4p_n^4}\underset{ij\neq i'j'}{\sum_{ij,i'j'}}\IE\Bigg[\underset{s\in[\xi,\xi+2h]}{\sup_{\xi\in[0,T]}}\sum_{k,l\in V_n}C_{n,kl}(s)\int_0^T\int_{t-2h}^{t-}\sum_{k',l'\in V_n}\mathcal{J}_{k'l'}(t) \nonumber \\
&\qquad\times\sup_{(a,b)\notin F_{ij}(t-2h)}\mathcal{J}_{ab}(t)d|M_{n,i'j'}|(r)d|M_{n,ij}|(t)\Bigg]\to0. \label{eq:wc1+4}
\end{align}
\begin{proof}
To show \eqref{eq:wc1+4} we note that $\mathcal{J}_{k'l'}(t)=\mathcal{J}_{k'l'}(t)\tilde{C}_{n,k'l'}$ because $\tilde{C}_{n,k'l'}=0$ implies $\mathcal{J}_{k'l'}(t)=0$. Then we bound \eqref{eq:wc1+4} by using bound $C_{n,kl}\leq\tilde{C}_{n,kl}$, collecting the sums in the front, and using that $\int_{t-2h}^{t-}d|M_{n,i'j'}|(r)\leq\mathcal{J}_{i'j}(t)$
\begin{align*}
\eqref{eq:wc1+4}\leq&\frac{1}{m^4p_n^4}\sum_{ij,i'j',kl,k'l'}\IE\Bigg[\int_0^T\int_{t-2h}^{t-}\mathcal{J}_{k'l'}(t)\sup_{(a,b)\notin F_{ij}(t-2h)}\mathcal{J}_{ab}(t)d|M_{n,i'j'}|(r)d|M_{n,ij}|(t) \\
&\qquad\tilde{C}_{n,kl}\tilde{C}_{n,ij}\tilde{C}_{n,i'j'}\tilde{C}_{n,k'l'}\Bigg] \\
\leq&\frac{1}{m^4p_n^4}\sum_{ij,i'j',kl,k'l'}\IE\Bigg[\int_0^T\mathcal{J}_{k'l'}(t)\mathcal{J}_{i'j'}(t)\sup_{(a,b)\notin F_{ij}(t-2h)}\mathcal{J}_{ab}(t)d|M_{n,ij}|(t) \\
&\qquad\tilde{C}_{n,kl}\tilde{C}_{n,ij}\tilde{C}_{n,i'j'}\tilde{C}_{n,k'l'}\Bigg] \\
\leq&\frac{1}{m^4p_n^4}\sum_{ij,i'j',kl,k'l'}\IP\left(\tilde{C}_{n,kl}\tilde{C}_{n,ij}\tilde{C}_{n,i'j'}\tilde{C}_{n,k'l'}=1\right) \\
&\times\underset{a=1,...,4}{\sup_{i_a,j_a\in V_n}}\IE\Bigg[\int_0^T\mathcal{J}_{i_4j_4}(t)\mathcal{J}_{i_2j_2}(t)\underset{d_{t-2h}(uv,i_1j_1)<M}{\sup_{u,v\in V_n}}\mathcal{J}_{uv}(t)d|M_{n,i_1j_1}|(t)\Big|\prod_{a=1}^4\tilde{C}_{n,i_aj_a}=1\Bigg].
\end{align*}
The first line remains bounded by Assumption (WC, \ref{eq:swc2}). For the second line we get
\begin{align*}
&\underset{a=1,...,4}{\sup_{i_a,j_a\in V_n}}\IE\Bigg[\int_0^T\mathcal{J}_{i_4j_4}(t)\mathcal{J}_{i_2j_2}(t)\underset{d_{t-2h}(uv,i_1j_1)<M}{\sup_{u,v\in V_n}}\mathcal{J}_{uv}(t)d|M_{n,i_1j_1}|(t)\Big|\prod_{a=1}^4\tilde{C}_{n,i_aj_a}=1\Bigg] \\
\leq&\underset{a=1,...,4}{\sup_{i_a,j_a\in V_n}}\IE\Bigg[\int_0^T\left(4h+N_{n,i_4j_4}[t-2h,t+2h]\right)\left(4h+N_{n,i_2j_2}[t-2h,t+2h]\right) \\
&\quad\times\underset{d_{t-2h}(uv,i_1j_1)<M}{\sup_{u,v\in V_n}}\left(4h+N_{n,uv}[t-2h,t+2h]\right)d|M_{n,i_1j_1}|(t)\bigg|\prod_{a=1}^4\tilde{C}_{n,i_aj_a}=1\Bigg]
\end{align*}
which is assumed to converge to zero in (BM, \ref{eq:BM6}).
\end{proof}

\section{Further Simulation Results}
\label{supp:simulation}

Note firstly, that according to the definition of $\Delta_n$ in Theorem \ref{thm:T1}, in our situation we have
$$\Delta_n(t)=\frac{1}{c_n}(\alpha_0(t)-\alpha(\theta_0;t))=\frac{\rho}{c_n}(\alpha_1(t)-\bar{\alpha}(t)).$$
Thus, here we violate our assumption of bounded $\Delta_n$. Figures \ref{fig:test_527} and \ref{fig:test_300} therefore show the histograms of
$$\frac{N\sqrt{h}\left(T_n-\frac{A_n}{Nh}-\int_0^T\left(\int_0^TK_{h,t}(s)\rho(\alpha_1(s)-\alpha_0(s))ds\right)^2w(t)dt\right)}{\sqrt{B_n}}.$$
According to Theorem \ref{thm:T1} this quantity should have asymptotically a standard normal distribution if we are in an asymptotic regime, where the deviation from the null hypothesis is of the order $c_n=(N\sqrt{h})^{-1/2}$. We can see that for small values of $\rho$ the empirically observed distribution is very close to the standard normal distribution. It seems that small values of $\rho$ can compensate for the unbounded $\Delta_n$. However, when $\rho$ increases, the approximation becomes less good which indicates that deviations of the corresponding size are too large to be adequately captured by the asymptotics. Note finally that for two sample sizes $n_1<n_2$, we have
$$\frac{\Delta_{n_2}(t)}{\Delta_{n_1}(t)}=\frac{c_{n_1}}{c_{n_2}}.$$
Thus, we see that the \emph{unboundedness} of $\Delta_n(t)$ gets worse when $n$ increases. In Figures \ref{fig:test_527} and \ref{fig:test_300} it appears at least visually that for larger $\rho$ the normal approximation is indeed slightly better for $n=300$. This is in-line with the previous discussion.

\begin{figure}
\includegraphics[width=\textwidth]{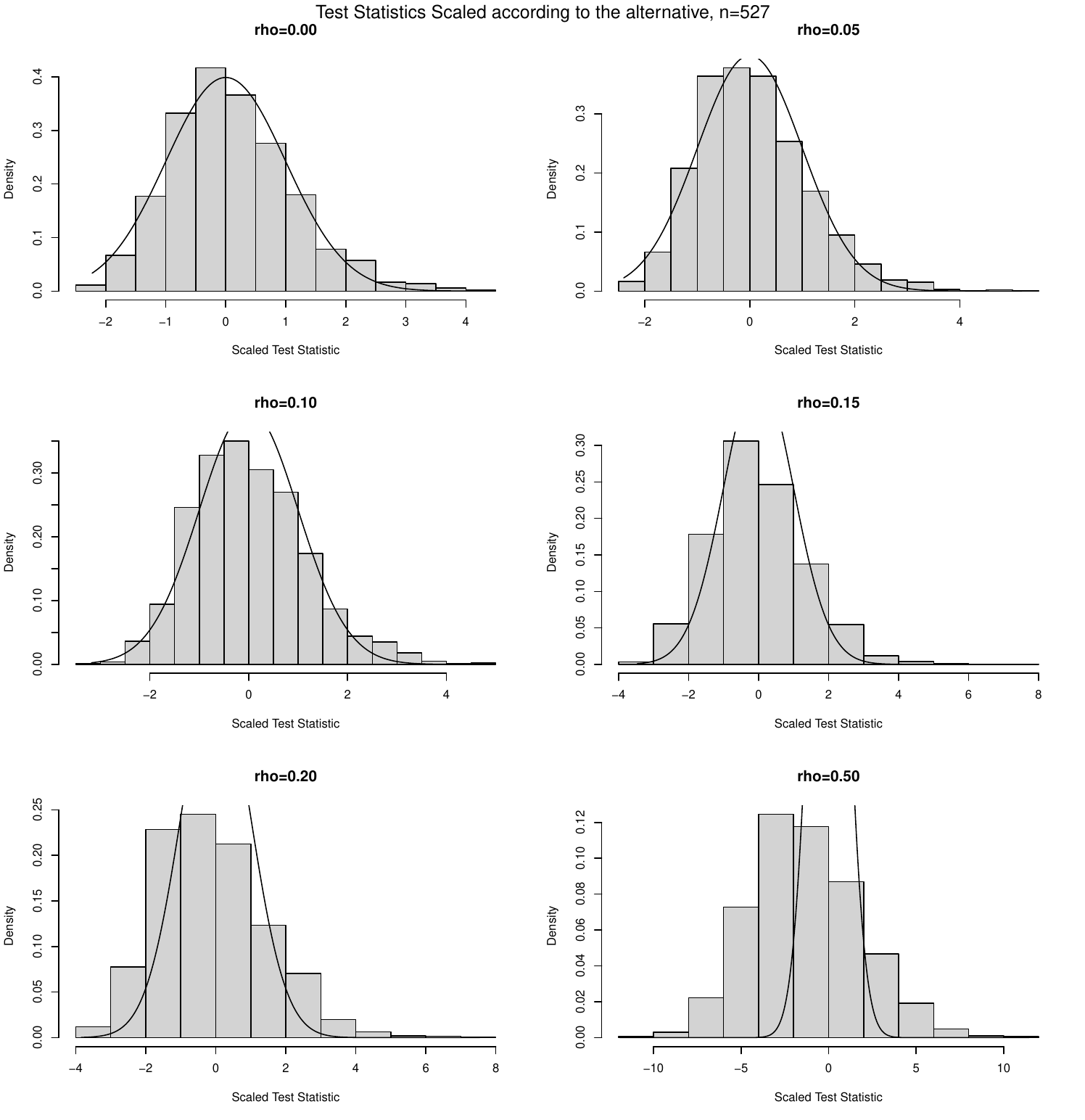}
\caption{Histograms of test statistics after scaling according to the corresponding alternative for networks of size $n=527$. The solid lines are densities of standard normal distributions.}
\label{fig:test_527}
\end{figure}

\begin{figure}
\includegraphics[width=\textwidth]{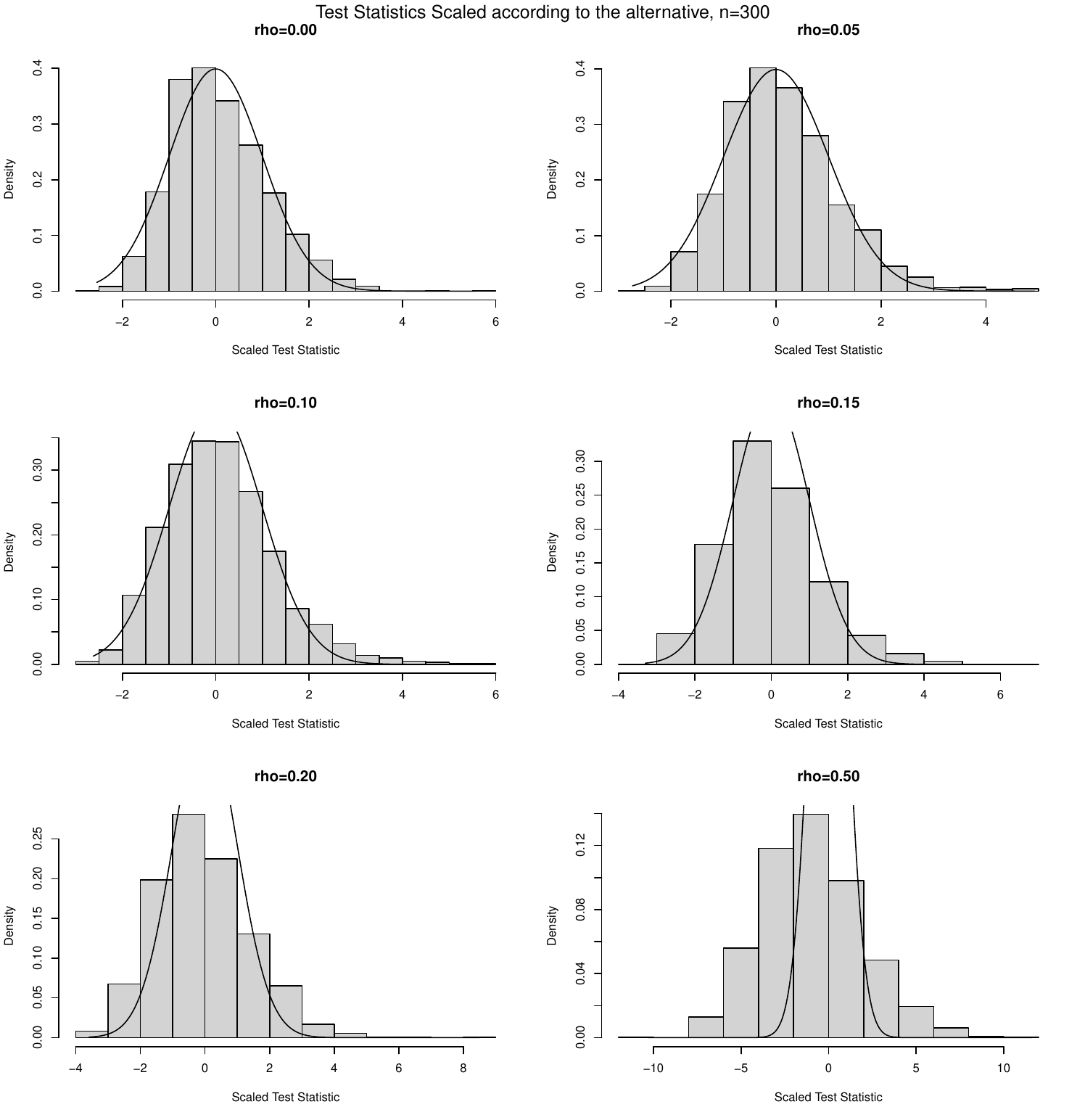}
\caption{Histograms of test statistics after scaling according to the corresponding alternative for networks of size $n=300$. The solid lines are densities of standard normal distributions.}
\label{fig:test_300}
\end{figure}

\section{Proofs}
\label{supp:main_proof}
\subsection{Prerequisites}
\label{app:prerequisites}
We collect here some consequences of the assumptions from Sections \ref{app:hubs} and \ref{app:lla}. Denote the neighbouring events ${\rm NE}_{ij}(t)$ of a pair $(i,j)$ and all events ${\rm AE}(t)$ as follows
\begin{align*}
{\rm NE}_{ij}(t):=&\sum_{k,l\in V_n}B_n^{kl}\mathcal{J}_{kl}(t)\Ind(kl\notin F_{ij}(t-2h)), \qquad{\rm AE}(t):=\sum_{k',l'\in V_n}B_n^{k'l'}\mathcal{J}_{k'l'}(t).
\end{align*}

\begin{lemma}
\label{lem:con_assump}
Let (KBW), (HSR), (SP), (C) and the assumptions from Section \ref{app:lla} hold. Then,
\begin{align}
&\frac{1}{N}\sup_{s,t\in[0,T],\,i,j,k,l\in V_n}\IE\left(\Phi_n(t,t)B_n^{\{ij,kl\}}\big|C_{n,ij}(t)=1,C_{n,kl}(s)=1\right)\to0, \label{eq:assump1} \\
&\frac{1}{N^2}\sup_{i,j\in V_n}\int_0^T\IE\left(B_n^{ij}\Phi_n(t,t)\big|C_{n,ij}(t)=1\right)dt\to0, \label{eq:assump2} \\
&\sup_{t,s\in[0,T],\,i,j,k,l\in V_n}\IE\left(\Phi_n(t,t)\Phi_n(s,s)B_n^{ij}B_n^{kl}\big|C_{n,ij}(t)=1,C_{n,kl}(s)=1\right)=O(1), \label{eq:assump3} \\
&\frac{1}{N}\sup_{t\in[0,T]}\IE\left(\left(B_n^{ij}\right)^2\Big|C_{n,ij}(t)=1\right)\to0, \label{eq:assump4} \\
&\frac{1}{N^3}\int_0^T\IE\left(\left(B_n^{ij}\right)^4\Big|C_{n,ij}(t)=1\right)dt\to0, \label{eq:assump5} \\
&\frac{1}{N^2}\sup_{i,j,k,l\in V_n,\,t,s\in[0,T]}\IE\left(\left(B_n^{ij}B_n^{kl}\right)^2\Big|C_{n,ij}(t)=1,\,C_{n,ij}(s)=1\right)\to0, \label{eq:assump6} \\
&\sup_{t,s\in[0,T],\,i,j,k,l\in V_n}\IE\left(B_n^{ij}B_n^{kl}\big|C_{n,ij}(t)=1,C_{n,kl}(s)=1\right)=O(1), \label{eq:assump7} \\
&\IE\left(\left(\sum_{k,l\in V_n}\int_0^T\frac{B_n^{kl}}{mp_n(s)}d|M_{n,kl}|(s)\right)^2\right)=O(1), \label{eq:assump8} \\
&\frac{1}{N}\IE\left(\left(\sum_{i,j\in V_n}\int_0^T\frac{\left(B_n^{ij}\right)^2}{mp_n(t)}d|M_{n,ij}|(t)\right)^2\right)=O(1) \label{eq:assump9} \\
&\sup_{i,j\in V_n,\,s\in[0,T]}\IE\left(\left(\frac{1}{N}\underset{(i',j')\neq(i,j)}{\sum_{i',j'\in V_n}}\int_{s-2h}^{s-}B_n^{\{ij,i'j'\}}d|M_{n,i'j'}|(t)\right)^2\Bigg|C_{n,ij}(s)=1\right)\to0, \label{eq:assump10} \\
&\frac{1}{m^3p_n^3}\underset{ij\neq i'j'}{\sum_{ij,i'j'}}\IE\Bigg[\underset{\rho\in[\xi-2h,\xi]}{\sup_{\xi\in[0,T]}}\Phi_n(\xi,\rho)\int_0^T\int_{t-2h}^{t-}{\rm AE}(t){\rm NE}_{ij}(t)d|M_{n,i'j'}|(r)d|M_{n,ij}|(t)\Bigg]\to0, \label{eq:assump11} \\
&\int_0^T\IE\left(\left(\sum_{k,l\in V_n}\int_{t-2h}^{t-}\frac{B_n^{kl}}{mp_n(r)}\Phi_n(t,r)d|M_{n,kl}|(r)\right)^2\right)dt\to0, \label{eq:assump12} \\
&\underset{s\in[0,T]}{\sup_{i,j\in V_n,}}\IE\left(\left(\underset{(k',l')\neq(k,l)}{\sum_{k,l,k',l'\in V_n}}\int_0^{s-}\int_{r-2h}^{r-}\frac{B_n^{\{kl,k'l'\}}}{N^2}d|M_{n,k'l'}|(t)d|M_{n,kl}|(r)\right)^2\Bigg|C_{n,ij}(s)=1\right)=O(1). \label{eq:assump13}
\end{align}
\end{lemma}
\begin{proof}
Throughout the proof we will use that by Assumption (SP), we have that $N/mp_n$ and $mp_n/N$ are bounded.

In order to show \eqref{eq:assump1}, we make use of the Cauchy-Schwarz Inequality as follows:
\begin{align*}
\eqref{eq:assump1}=&\frac{1}{N}\underset{i,j,k,l\in V_n}{\sup_{s,t\in[0,T],}}\int_t^{t+2h}\frac{1}{h}\IE\left(\frac{1}{mp_n(r)}\overline{\Psi}_n(r;\beta_0)B_n^{\{ij,kl\}}\big|C_{n,ij}(t)C_{n,kl}(s)=1\right)dr \\
\leq&2\underset{i,j,k,l\in V_n}{\sup_{s,t\in[0,T],}}\sup_{r\in[t,t+2h]}\IE\left(\left(\frac{1}{mp_n(r)}\overline{\Psi}_n(r;\beta_0)\right)^2\big|C_{n,ij}(t)C_{n,kl}(s)=1\right)^{\frac{1}{2}} \\
&\quad\times\underset{i,j,k,l\in V_n}{\sup_{s,t\in[0,T],}}\frac{1}{N}\IE\left(\left(B_n^{\{ij,kl\}}\right)^2\big|C_{n,ij}(t)C_{n,kl}(s)=1\right)^{\frac{1}{2}}.
\end{align*}
The first line remains bounded by \eqref{eq:wc1+3}. For the second line, we note that since $N\to\infty$ by (KBW), the $n_{\rm hub}$ in $B_n^{ij,kl}$ plays no role. For $H_{\rm hub}^{ij,kl}\kappa_n$ we obtain by (HSR, \ref{eq:HSR1})
\begin{align*}
&\underset{i,j,k,l\in V_n}{\sup_{s,t\in[0,T],}}\frac{1}{N}\IE\left(\left(B_n^{\{ij,kl\}}\right)^2\big|C_{n,ij}(t)C_{n,kl}(s)=1\right)^{\frac{1}{2}} \\
\leq&\underset{i,j,k,l\in V_n}{\sup_{s,t\in[0,T],}}\frac{\kappa_n}{N}\IP\left(H_{\rm hub}^{\{ij,kl\}}=1\big|C_{n,ij}(t)C_{n,kl}(s)=1\right)^{\frac{1}{2}}\to0.
\end{align*}
Statement \eqref{eq:assump2} is implied by \eqref{eq:assump1}. For \eqref{eq:assump3} we apply the Cauchy-Schwarz Inequality twice in the following way (in the first step with exponents $3/2$ and $3$ and in the second step with exponents $2$):
\begin{align*}
\eqref{eq:assump3}=&\underset{i,j,k,l\in V_n}{\sup_{t,s\in[0,T],}}\int_t^{t+2h}\int_s^{s+2h}\IE\left(\frac{\overline{\Psi}_n(r;\beta_0)\overline{\Psi}_n(q;\beta_0)B_n^{ij}B_n^{kl}}{h^2m^2p_n(r)p_n(q)}\big|C_{n,ij}(t)C_{n,kl}(s)=1\right)drdq \\
\leq&4\underset{i,j,k,l\in V_n}{\sup_{t,s\in[0,T],}}\underset{q\in[s,s+2h]}{\sup_{r\in[t,t+2h],}}\IE\left(\frac{\overline{\Psi}_n(r;\beta_0)\overline{\Psi}_n(q;\beta_0)B_n^{ij}B_n^{kl}}{m^2p_n(r)p_n(q)}\big|C_{n,ij}(t)C_{n,kl}(s)=1\right) \\
\leq&4\underset{i,j,k,l\in V_n}{\sup_{t,s\in[0,T],}}\underset{q\in[s,s+2h]}{\sup_{r\in[t,t+2h],}}\IE\left(\left(\frac{\overline{\Psi}_n(r;\beta_0)\overline{\Psi}_n(q;\beta_0)}{m^2p_n(r)p_n(q)}\right)^{\frac{3}{2}}\big|C_{n,ij}(t)C_{n,kl}(s)=1\right)^{\frac{2}{3}} \\
&\quad\times\IE\left(\left(B_n^{ij}B_n^{kl}\right)^3\big|C_{n,ij}(t)C_{n,kl}(s)=1\right)^{\frac{1}{3}} \\
\leq&4\underset{i,j,k,l\in V_n}{\sup_{t,s\in[0,T],}}\underset{q\in[s,s+2h]}{\sup_{r\in[t,t+2h],}}\IE\left(\left(\frac{1}{mp_n(r)}\overline{\Psi}_n(r;\beta_0)\right)^3\big|C_{n,ij}(t)C_{n,kl}(s)=1\right)^{\frac{1}{3}} \\
&\quad\times\IE\left(\left(\frac{1}{mp_n(q)}\overline{\Psi}_n(q;\beta_0)\right)^3\big|C_{n,ij}(t)C_{n,kl}(s)=1\right)^{\frac{1}{3}} \\
&\quad\times\IE\left(\left(B_n^{ij}B_n^{kl}\right)^3\big|C_{n,ij}(t)C_{n,kl}(s)=1\right)^{\frac{1}{3}}.
\end{align*}
The first two lines are identical and remain bounded by \eqref{eq:wc1+3}. The last line remains bounded by assumption (HSR, \ref{eq:HSR3}).

It is simple to see that \eqref{eq:assump4}-\eqref{eq:assump7} are direct consequences of (HSR \ref{eq:HSR2}) and (HSR \ref{eq:HSR3}).

For \eqref{eq:assump8}, we make the following decomposition based on the definition of $B_n^{ij,kl}$:
\begin{align*}
\eqref{eq:assump8}\leq&2n_{\rm hub}^2\IE\left(\left(\sum_{k,l\in V_n}\int_0^T\frac{1}{mp_n(s)}d|M_{n,kl}|(s)\right)^2\right) \\
&\quad+2\kappa_n^2\IE\left(\left(\sum_{k,l\in V_n}\int_0^T\frac{H_{\rm hub}^{kl}}{mp_n(s)}d|M_{n,kl}|(s)\right)^2\right)
\end{align*}
The first line is bounded by \eqref{eq:wc2} and the second line follows from (HSR, \ref{eq:HSR4}). \eqref{eq:assump9} can be seen through the same arguments. For \eqref{eq:assump10} we use a similar bound (and note that by (SP) we have $N\geq cmp_n(t)$ for a suitable $c>0$)
\begin{align*}
\eqref{eq:assump10}\leq&\frac{2n_{\rm hub}^2}{c}\sup_{i,j\in V_n,\,s\in[0,T]}\IE\left(\left(\underset{(i',j')\neq(i,j)}{\sum_{i',j'\in V_n}}\int_{s-2h}^{s-}\frac{1}{mp_n(t)}d|M_{n,i'j'}|(t)\right)^2\Bigg|C_{n,ij}(s)=1\right) \\
&\quad+\frac{\kappa_n^2}{c}\sup_{i,j\in V_n,\,s\in[0,T]}\IE\left(\left(\underset{(i',j')\neq(i,j)}{\sum_{i',j'\in V_n}}\int_{s-2h}^{s-}\frac{H_{\rm hub}^{\{ij,i'j'\}}}{mp_n(t)}d|M_{n,i'j'}|(t)\right)^2\Bigg|C_{n,ij}(s)=1\right).
\end{align*}
The first line converges to zero by \eqref{eq:wc4} and the second line by (HSR, \ref{eq:HSR5}).

In order to show \eqref{eq:assump11} we firstly note that $\mathcal{J}_{kl}(t)=\mathcal{J}_{kl}(t)\sup_{\rho\in[t-2h,t+2h]}C_{n,kl}(\rho)$ because $\sup_{\rho\in[t-2h,t+2h]}C_{n,kl}(\rho)=0$ implies that $\mathcal{J}_{kl}(t)=0$. Below we firstly replace the definitions of ${\rm NE}$ and ${\rm AE}$, then in the second step we take the supremum over all $B_n^{kl}\mathcal{J}_{kl}$ in the neighbourhood of $(i,j)$ and use the definition of $K_M^{ij}$, and finally we write out $\Phi_n$ and bound it by standard bounds
\begin{align*}
\eqref{eq:assump11}=&\frac{1}{m^3p_n^3}\underset{ij\neq i'j'}{\sum_{ij,i'j'}}\IE\Bigg[\underset{\rho\in[\xi-2h,\xi]}{\sup_{\xi\in[0,T]}}\Phi_n(\xi,\rho)\int_0^T\int_{t-2h}^{t-}\sum_{k',l'\in V_n}B_n^{k'l'}\mathcal{J}_{k'l'}(t) \\
&\qquad\times\sum_{k,l\in V_n}B_n^{kl}\mathcal{J}_{kl}(t)\sup_{\rho\in[t-2h,t+2h]}C_{n,kl}(\rho)\Ind(kl\notin F_{ij}(t-2h))d|M_{n,i'j'}|(r)d|M_{n,ij}|(t)\Bigg]  \\
\leq&\frac{1}{m^3p_n^3}\underset{ij\neq i'j'}{\sum_{ij,i'j'}}\IE\Bigg[\underset{\rho\in[\xi-2h,\xi]}{\sup_{\xi\in[0,T]}}\Phi_n(\xi,\rho)\int_0^T\int_{t-2h}^{t-}\sum_{k',l'\in V_n}B_n^{k'l'}\mathcal{J}_{k'l'}(t) \\
&\qquad\times\sup_{(a,b)\notin F_{ij}(t-2h)}B_n^{ab}\mathcal{J}_{ab}(t)K_M^{ij}d|M_{n,i'j'}|(r)d|M_{n,ij}|(t)\Bigg]  \\
\leq&\frac{1}{m^3p_n^3}\underset{ij\neq i'j'}{\sum_{ij,i'j'}}\IE\Bigg[\sup_{s\in[0,T]}\frac{2}{mp_n}\sum_{k,l\in V_n}C_{n,kl}(s)\|\Psi(\cdot;\beta_0)\|_{\infty}\int_0^T\int_{t-2h}^{t-}\sum_{k',l'\in V_n}B_n^{k'l'}\mathcal{J}_{k'l'}(t) \\
&\qquad\times\sup_{(a,b)\notin F_{ij}(t-2h)}B_n^{ab}\mathcal{J}_{ab}(t)K_M^{ij}d|M_{n,i'j'}|(r)d|M_{n,ij}|(t)\Bigg]
\end{align*}
which converges to zero by Assumption (C), \eqref{eq:wc1+4} and (HSR, \ref{eq:HSR6}).

It is clear that \eqref{eq:assump12} follows from \eqref{eq:wc9} and (HSR, \ref{eq:HSR7}), and that \eqref{eq:assump13} follows from \eqref{eq:wc5} and (HSR, \ref{eq:HSR8}) (note that the integrand in \eqref{eq:assump13} is non-negative).
\end{proof}

\subsection{Rebolledo's Martingale Central Limit Theorem}
\label{app:rebolledo}
For easier reference we state here the exact form of Rebolledo's Martingale Central Limit Theorem. We state here a version based on Theorem II.5.1 in \citet{ABGK93}, the original work is \citet{R80}.

Let $M^n=(M_1^n,...,M_k^n)$ be a vector of sequences of locally square integrable martingales on an interval $\mathcal{T}$. For $\epsilon>0$ we denote by $M_{\epsilon}^n$ a vector of locally square integrable martingales that contain all jumps of components of $M^n$ which are larger in absolute value than $\epsilon$, i.e., $M_i^n-M_{\epsilon,i}^n$ is a local square integrable martingale for all $i=1,...,k$ and $|\Delta M_i^n-\Delta M_{\epsilon,i}^n|\leq\epsilon$. Furthermore, we denote by $\langle M^n\rangle:=\left(\langle M_i^n,M_j^n\rangle\right)_{i,j=1,...,k}$ the $k\times k$ matrix of quadratic covariations.

Moreover, we denote by $M$ a multivariate, continuous Gaussian martingale with $\langle M\rangle_t=V_t$, where $V:\mathcal{T}\to\IR^{k\times k}$ is a continuous deterministic $k\times k$ positive semi-definite matrix valued function on $\mathcal{T}$ such that its increments $V_t-V_s$ are also positive semi-definite for $s\leq t$, then $M_t-M_s\sim \mathcal{N}(0,V_t-V_s)$ is independent of $(M_r:\,r\leq s)$. Given such a function $V$, such a Gaussian process $M$ always exists. We can now formulate the central limit theorem for martingales.

\begin{theorem}
\label{thm:Rebolledo}
Let $\mathcal{T}_0\subseteq\mathcal{T}$. Assume that for all $t\in\mathcal{T}_0$ as $n\to\infty$ we have $\langle M^n\rangle_t\overset{\IP}{\rightarrow} V_t$ and $\langle M_{\epsilon}^n\rangle_t\overset{\IP}{\rightarrow}0$ for all $\epsilon>0$. Then, $M^n_t\overset{d}{\rightarrow} M_t$ as $n\to\infty$ for all $t\in\mathcal{T}_0$.
\end{theorem}
 
\subsection{Details for the Proof of Theorem \ref{thm:T1}}
\label{app:proof}
\begin{proposition}
\label{prop:mmd}
In the setting of Theorem \ref{thm:T1} assume that Assumptions (VX), (B), (KBW), (C), (mDep), (HSR), (LL), (SP), (P) and all assumptions from Section \ref{app:lla} hold true and suppose that $\sup_{n\in\IN}\|\Delta_n\|_{\infty}<\infty$. Then,
\begin{equation}
\label{eq:I1}
B_n^{-\frac{1}{2}}\left(N\sqrt{h}\int_0^TI_1(t)^2w(t)dt-h^{-\frac{1}{2}}A_n\right)\overset{d}{\to}\mathcal{N}(0,1),
\end{equation}
where
\begin{align*}
A_n&:=N\sum_{i,j\in V_n}\int_0^T\frac{f_n(r,r)}{m^2p_n(r)^2\mu_n(r,\beta_0)^2}dN_{n,ij}(r), \\
B_n&:=4K^{(2)}\int_0^T\left(\frac{w(s)\alpha_0(s)}{\gamma\pi(s)\mu_n(s;\beta_0)}\right)^2ds, \\
K^{(2)}&:=\int_0^2\left(\int_{-1}^1K(u+v)K(u)du\right)^2dv.
\end{align*}
\end{proposition}

\begin{proof}[Proof of Proposition \ref{prop:mmd}]
By Assumption (B), we have that $\inf_{n\in\IN}B_n>0$. By Lemma \ref{lem:I1} from Section \ref{sec:support_lemmas} it is hence sufficient to prove
\begin{equation}
\label{eq:martconv}
B_n^{-\frac{1}{2}}\left(N\sqrt{h}\int_0^T\left(\int_0^TK_{h,t}(s)\frac{1}{mp_n(s)}\mu_n(s;\beta_0)^{-1}dM_n(s)\right)^2w(t)dt-h^{-\frac{1}{2}}A_n\right)\overset{d}{\to} N(0,1).
\end{equation}
Recall that
$$f_n(r,s):=\int_0^ThK_{h,t}(s)K_{h,t}(r)w(t)dt$$
is symmetric. Then, we can rewrite the integral in \eqref{eq:martconv} as follows
\begin{align}
&N\sqrt{h}\int_0^T\left(\int_0^TK_{h,t}(s)\frac{1}{mp_n(s)}\mu_n(s,\beta_0)^{-1}dM_n(s)\right)^2w(t)dt \nonumber \\
=&2Nh^{-\frac{1}{2}}\int_0^T\int_0^{r-}f_n(r,s)\frac{1}{mp_n(s)}\mu_n(s;\beta_0)^{-1}\frac{1}{mp_n(r)}\mu_n(r;\beta_0)^{-1}dM_n(s)dM_n(r) \label{eq:distr} \\
&\quad+Nh^{-\frac{1}{2}}\int_0^T\int_{\{s=r\}}f_n(r,s)\frac{1}{mp_n(r)}\mu_n(r;\beta_0)^{-1}\frac{1}{mp_n(s)}\mu_n(s;\beta_0)^{-1}dM_n(s)dM_n(r). \label{eq:bias}
\end{align}
Note that the two martingales $M_{n,ij}$ and $M_{n,kl}$ do not jump simultaneously when $(i,j)\neq (k,l)$ and that the height of each jump is exactly one. We hence obtain
\begin{align*}
\eqref{eq:bias}&=Nh^{-\frac{1}{2}}\sum_{(i,j)\in V_n\times V_n}\int_0^T\frac{f_n(r,r)}{m^2p_n(r)^2}\mu_n(r;\beta_0)^{-2}dN_{n,ij}(r)=h^{-\frac{1}{2}}A_n.
\end{align*}
In the next step, we split the double summation in \eqref{eq:distr} (from writing out the definition $M_n=\sum_{(i,j)\in V_n\times V_n}M_{n,ij}$) in two steps as well,
\begin{align}
&\eqref{eq:distr} \nonumber \\
=&2N\underset{(i,j)\neq (k,l)}{\sum_{(i,j),(k,l)\in V_n\times V_n}}h^{-\frac{1}{2}}\int_0^T\int_0^{r-}\frac{f_n(r,s)}{m^2p_n(s)p_n(r)}\mu_n(s;\beta_0)^{-1}\mu_n(r;\beta_0)^{-1}dM_{n,ij}(s)dM_{n,kl}(r) \label{eq:dummylabel} \\
&+2N\sum_{(i,j)\in V_n\times V_n}h^{-\frac{1}{2}}\int_0^T\int_0^{r-}\frac{f_n(r,s)}{m^2p_n(s)p_n(r)}\mu_n(s,\beta_0)^{-1}\mu_n(r,\beta_0)^{-1}dM_{n,ij}(s)dM_{n,ij}(r). \label{eq:distr2}
\end{align}
By assumption (KBW) $hK_{h,t}(s)\leq \|K\|_{\infty}<\infty$. Using this and the bounded support of $K$ we can prove the useful bound $\int_0^T\int_0^Tf_n(r,s)^2\frac{1}{m^2p_n(r)p_n(s)}dsdr\leq \|K\|_{\infty}^2N^{-2}$. Note, moreover, that we achieved by the splitting that \eqref{eq:distr2} is a martingale and thus we may compute the expectation of the square as the expectation of the quadratic variation. Therefore, we get for $n\to\infty$ with $\eta_0:=\inf_{t\in[0,T]}\mu_n(t;\beta_0)$
\begin{align*}
&\IE(\eqref{eq:distr2}^2) \\
=&\frac{4N^2}{h}\sum_{(i,j)\in V_n\times V_n}\int_0^T\IE\Bigg(\left(\int_0^{r-}\frac{f_n(r,s)}{m^2p_n(s)p_n(r)\mu_n(s,\beta_0)\mu_n(r,\beta_0)}dM_{n,ij}(s)\right)^2 \\
&\quad\quad\quad\quad\times C_{n,ij}(r)\alpha_0(r)\Psi(X_{n,ij}(r);\beta_0)\Bigg)dr \\
\leq&\frac{4N^2\|\Psi(\cdot;\beta_0)\|_{\infty}\|\alpha_0\|_{\infty}}{h} \\
&\qquad\times\sum_{(i,j)\in V_n\times V_n}\int_0^T\IE\Bigg(\left(\int_0^{r-}\frac{f_n(r,s)}{m^2p_n(s)p_n(r)\mu_n(s,\beta_0)\mu_n(r,\beta_0)}dM_{n,ij}(s)\right)^2\Bigg)dr \\
=&\frac{4N^2\|\Psi(\cdot;\beta_0)\|_{\infty}\|\alpha_0\|_{\infty}}{h}\sum_{(i,j)\in V_n\times V_n}\int_0^T\int_0^r\frac{f_n(r,s)^2\IE\left(\Psi(X_{n,ij}(s);\beta_0)C_{n,ij}(s)\right)\alpha_0(s)}{m^4p_n(s)^2p_n(r)^2\mu_n(s,\beta_0)^2\mu_n(r,\beta_0)^2}dsdr \\
\leq&\frac{4N^2\|\Psi(\cdot;\beta_0)\|_{\infty}\|\alpha_0\|_{\infty}^2}{hmp_n\eta_0^3}\int_0^T\int_0^r\frac{f_n(r,s)^2}{m^2p_n(s)p_n(r)}dsdr\leq\frac{4\|\Psi(\cdot;\beta_0)\|_{\infty}\|\alpha_0\|_{\infty}^2}{hmp_n\eta_0^3}\|K\|_{\infty}^2\to0
\end{align*}
because $hmp_n\to\infty$ and $\eta_0>0$ and all $\infty$-norms are bounded by (C). Recalling that $B_n>\epsilon>0$ for all $n\in\IN$ we see from the previous results that the proof of the proposition is complete if we can show that
\begin{equation}
\label{eq:distr1}
\frac{2N}{\sqrt{hB_n}}\underset{(i,j)\neq (k,l)}{\underset{(k,l)\in V_n\times V_n}{\sum_{(i,j)\in V_n\times V_n,}}}\int_0^T\int_0^{r-}\frac{f_n(r,s)}{m^2p_n(s)p_n(r)\mu_n(s;\beta_0)\mu_n(r;\beta_0)}dM_{n,ij}(s)dM_{n,kl}(r)\overset{d}{\to} N(0,1)
\end{equation}
We will do this by showing the assumptions of Rebolledo's martingale central limit theorem (cf. Theorem \ref{thm:Rebolledo}). We note to this end that \eqref{eq:distr1} is indeed a martingale evaluated at $T$. We therefore have to study the variation process as well as the jump process. This will be tedious. To simplify the notation, we introduce the following functions for $s_2\leq s_1$, $(i,j),(k,l)\in V_n\times V_n$ and $I\subseteq V_n\times V_n$ (the $M$ is the same as in (mDep)):
\begin{align}
&\phi_{n,ij,i'j'}^I(s_1,s_2) \label{eq:def_phi} \\
:=&\frac{4mN^2}{h}\int_{s_1}^{s_2+2h}f_n(r,s_1)f_n(r,s_2)\frac{\mu_n(r;\beta_0)^{-2}}{mp_n(r)}\alpha_0(r)\frac{\mu_n(s_1;\beta_0)^{-1}}{mp_n(s_1)}\frac{\mu_n(s_2;\beta_0)^{-1}}{mp_n(s_2)} \nonumber \\
&\quad\quad\quad\times\frac{1}{mp_n(r)}\underset{(k,l)\neq(i,j),(i',j')}{\sum_{(k,l)\in V_n\times V_n}}\Psi(X_{n,kl}(r);\beta_0)C_{n,kl}(r)\Ind(d_{s_1-4h}^n(kl,I)\geq M)dr \nonumber
\end{align}
We will write $\phi_{n,ij}^I(r)$ instead of $\phi_{n,ij,ij}^I(r,r)$ and let $\phi_{n,ij,kl}(s_1,s_2):=\phi^{\emptyset}_{n,ij,kl}(s_1,s_2)$. For completeness we let $\phi_{n,ij,i'j'}^ I(s_1,s_2)=0$ if $s_2>s_1$.  Note finally that $f_n(r,s)$ functions as a kernel: Since $K$ is supported on $[-1,1]$, $f_n(r,s)$ is different from zero only if $|r-s|\leq 2h$. This implies that $\phi_{n,ij,i'j'}^ I(s_1,s_2)=0$ whenever $s_2<s_1-2h$ With this notation the variation process can be written as follows:
\begin{align}
&\left\langle\eqref{eq:distr1}\right\rangle \nonumber \\
=&\frac{1}{mB_n}\sum_{(i,j),(i',j')\in V_n\times V_n}\int_0^T\int_0^T\phi_{n,ij,i'j'}(s_1,s_2)dM_{n,i'j'}(s_2)dM_{n,ij}(s_1) \nonumber \\
=&\frac{1}{mB_n}\sum_{(i,j),(i',j')\in V_n\times V_n}\int_0^T\int_{s_1-2h}^{s_1-}2\phi_{n,ij,i'j'}(s_1,s_2)dM_{n,i'j'}(s_2)dM_{n,ij}(s_1) \label{eq:zero_var} \\
&\quad+\frac{1}{mB_n}\sum_{(i,j),(i',j')\in V_n\times V_n}\int_0^T\int_{\{s_1\}}\phi_{n,ij,i'j'}(s_1,s_2)dM_{n,i'j'}(s_2)dM_{n,ij}(s_1). \label{eq:nonzero_var}
\end{align}
Note that the splitting in \eqref{eq:distr} and \eqref{eq:bias} relied on symmetry of $f_n$. Here we use that $\phi_{n,ij,i'j'}(s_1,s_2)=\phi_{n,i'j',ij}(s_1,s_2)$ and the fact that we sum over all $(i,j),(i',j')$. We will prove that \eqref{eq:zero_var} converges to zero and \eqref{eq:nonzero_var} converges to one which is going to be the asymptotic variance of the normal distribution. For the latter we note again that for $(i,j)\neq (i',j')$ the martingales $M_{n,ij}$ and $M_{n,i'j'}$ never jump simultaneously and that their jump height always equals one. So we get
\begin{align}
\eqref{eq:nonzero_var}=&\frac{1}{mB_n}\sum_{(i,j)\in V_n\times V_n}\int_0^T\phi_{n,ij}(s)dM_{n,ij}(s) \label{eq:abcdef} \\
&\quad\quad+\frac{1}{mB_n}\sum_{(i,j)\in V_n\times V_n}\int_0^T\phi_{n,ij}(s)\Psi(X_{n,ij}(s);\beta_0)C_{n,ij}(s)\alpha_0(s)ds. \label{eq:abcdefg}
\end{align}
We show that \eqref{eq:abcdef} converges to zero by Proposition 2.9 from \citet{K20}: We apply the proposition with $\tilde{\phi}_{n,ij}^J(s):=\frac{1}{mB_n}\phi_{n,ij}^J(s)$. These functions are of leave-$M$-out type, that is, they are measurable with respect to $\tilde{\mathcal{F}}_{ij,t}^{n,J,M}$, cf. Section \ref{app:mDep_mot}. In the following, we show that the three terms which appear in the upper bound there converge to zero. Firstly, we obtain by Lemma \ref{lem:phi_bound} and the definition of $\Phi_n$ (cf. Lemma \ref{lem:phi_bound})
\begin{align*}
&\frac{1}{m^2B_n^2}\sum_{ij}\int_0^T\IE\left(\phi_{n,ij}^{ij}(t)^2C_{n,ij}(t)\alpha_0(t)\Psi(X_{n,j}(t),\beta_0)\right)dt \\
\leq&\frac{C^2\|\Psi(\cdot;\beta_0)\|_{\infty}\|\alpha_0\|_{\infty}}{m^2B_n^2}\IE\left(\sum_{ij}\int_0^T\frac{1}{p_n(t)^2}\Phi_n(t,t)^2C_{n,ij}(t)dt\right) \\
=&\frac{C^2\|\Psi(\cdot;\beta_0)\|_{\infty}\|\alpha_0\|_{\infty}}{m^2B_n^2}\sum_{ij}\int_0^T\frac{1}{p_n(t)^2}\IE\left(\left(\int_t^{t+2h}\frac{1}{hmp_n(s)}\overline{\Psi}_n(s;\beta_0)ds\right)^2C_{n,ij}(t)\right)dt \\
\leq&\frac{4C^2\|\Psi(\cdot;\beta_0)\|_{\infty}^3\|\alpha_0\|_{\infty}}{mp_nB_n^2}\frac{1}{m}\sum_{ij}\int_0^T\sup_{s\in[t,t+2h]}\IE\left(\left(\frac{1}{mp_n(s)}\sum_{kl}C_{n,kl}(s)\right)^2\Big|C_{n,ij}(t)=1\right)dt
\end{align*}
which converges to zero since the $\infty$-norms are bounded by (C), $mp_n\to\infty$ by (KBW), $\inf B_n>0$ by (B) and \eqref{eq:wc1+1}. For the second part of the upper bound we introduce the notation $B_n^A:=\sup_{(a,b)\in A}B_n^{ab}$ for index sets $A\subseteq V_n\times V_n$ and where $B_n^{ab}$ is defined in Section \ref{app:hubs}. Since $B_n^{ab}\geq1$ for all pairs $(a,b)$ we have that $B_n^{\{ij,kl\}}\leq B_n^{ij}B_n^{kl}$. Moreover, we note that for any process $X:[0,T]\to\IR$ we have for the stochastic integral
\begin{align*}
\left|\int X(t)dM_{n,ij}(t)\right|\leq&\int |X(t)|dN_{n,ij}(t)+\int |X(t)|\lambda_{n,ij}(t)dt \\
=&\int |X(t)|dM_{n,ij}(t)+2\int |X(t)|\lambda_{n,ij}(t)dt
=:\int |X(t)|d|M_{n,ij}|(t).    
\end{align*}
We keep this notation in mind and bound the second term from Proposition 2.9 in \citet{K20} as follows by Lemma \ref{lem:phi_bound}:
\begin{align}
&\left|\frac{1}{m^2B_n^2}\sum_{ij,kl}\IE\left(\int_0^T\phi_{n,ij}^{ij,kl}(t)dM_{n,ij}(t)\int_0^T\left(\phi_{n,kl}(t)-\phi_{n,kl}^{ij,kl}(t)\right)dM_{n,kl}(t)\right)\right| \nonumber \\
\leq&\frac{1}{m^2B_n^2}\sum_{ij,kl}\IE\left(\int_0^T\left|\phi_{n,ij}^{ij,kl}(t)\right|d|M_{n,ij}|(t)\int_0^T\left|\phi_{n,kl}(t)-\phi_{n,kl}^{ij,kl}(t)\right|d|M_{n,kl}|(t)\right) \nonumber \\
\leq&\frac{1}{m^2B_n^2}\sum_{ij,kl}\IE\left(\int_0^T\left|\phi_{n,ij}^{ij,kl}(t)\right|d|M_{n,ij}|(t)\int_0^T\frac{2CB_n^{\{ij,kl\}}}{mP_n(t)^2}d|M_{n,kl}|(t)\right) \nonumber \\
=&\frac{1}{m^2B_n^2}\sum_{ij,kl}\IE\left(\int_0^T\left|\phi_{n,ij}^{ij,kl}(t)\right|dM_{n,ij}(t)\int_0^T\frac{2CB_n^{\{ij,kl\}}}{mp_n(t)^2}dM_{n,kl}(t)\right) \label{eq:bp1} \\
&+\frac{2}{m^2B_n^2}\sum_{ij,kl}\IE\left(\int_0^T\left|\phi_{n,ij}^{ij,kl}(t)\right|dM_{n,ij}(t)\int_0^T\frac{2CB_n^{\{ij,kl\}}}{mp_n(t)^2}\lambda_{n,kl}(t)dt\right) \label{eq:bp2} \\
&+\frac{2}{m^2B_n^2}\sum_{ij,kl}\IE\left(\int_0^T\frac{CB_n^{ij}}{p_n(t)}\Phi_n(t,t)\lambda_{n,ij}(t)dt\int_0^T\frac{2CB_n^{kl}}{mp_n(t)^2}dM_{n,kl}(t)\right) \label{eq:bp3} \\
&+\frac{4}{m^2B_n^2}\sum_{ij,kl}\IE\left(\int_0^T\frac{CB_n^{ij}}{p_n(t)}\Phi_n(t,t)\lambda_{n,ij}(t)dt\int_0^T\frac{2CB_n^{kl}}{mp_n(t)^2}\lambda_{n,kl}(t)dt\right). \label{eq:bp4}
\end{align}
In order to handle \eqref{eq:bp1}, we observe that by (mDep) the processes $\phi_{n,ij}^{ij,kl}(t)$ are predictable with respect to the augmented filtration $\tilde{\mathcal{F}}_{V_n\times V_n,t}^{n,\{ij,kl\},m}$ and that $M_{n,ij},M_{n,kl}$ are both uncorrelated martingales with respect to the same filtration. Therefore, we may compute (using also the measurability properties of $B_n^{\{ij,kl\}}$ and Lemma \ref{lem:phi_bound})
\begin{align*}
\eqref{eq:bp1}=&\frac{1}{m^2B_n^2}\sum_{ij}\IE\left(\int_0^T\left|\phi_{n,ij}^{ij}(t)\right|\frac{2CB_n^{ij}}{mp_n(t)^2}\lambda_{n,ij}(t)dt\right) \\
\leq&\frac{2C^2\|\alpha_0\|_{\infty}\|\Psi(\cdot;\beta_0)\|_{\infty}}{m^2p_n^2B_n^2}\int_0^T\IE\left(B_n^{ij}\Phi_n(t,t)|C_{n,ij}(t)=1\right)dt\to0
\end{align*}
because the $\infty$-norms  are bounded, $\inf_{n\in\IN}B_n>0$, $N^2/m^2p_n^2=O(1)$ by (SP) and the integral has the right order by \eqref{eq:assump2}. For the next term we have to split the integration and use that the expectation of the stochastic integral with respect to $M_{n,ij}$ is zero as well as (mDep) again:
\begin{align*}
\eqref{eq:bp2}=&\frac{2}{m^2B_n^2}\sum_{ij,kl}\IE\left(\int_0^T\int_0^t\left|\phi_{n,ij}^{ij,kl}(t)\right|\frac{2CB_n^{\{ij,kl\}}}{mp_n(s)^2}\lambda_{n,kl}(s)dsdM_{n,ij}(t)\right) \\
&+\frac{2}{m^2B_n^2}\sum_{ij,kl}\IE\left(\int_0^T\int_0^s\left|\phi_{n,ij}^{ij,kl}(t)\right|\frac{2CB_n^{\{ij,kl\}}}{mp_n(s)^2}\lambda_{n,kl}(s)dM_{n,ij}(t)ds\right) \\
\leq&\frac{4C\|\alpha_0\|_{\infty}\|\Psi(\cdot;\beta_0)\|_{\infty}}{m^3B_n^2}\sum_{ij,kl}\IE\left(\int_0^T\int_0^s\left|\phi_{n,ij}^{ij,kl}(t)\right|\frac{B_n^{\{ij,kl\}}C_{n,kl}(s)}{p_n(s)^2}dN_{n,ij}(t)ds\right) \\
\leq&\frac{4C\|\alpha_0\|_{\infty}^2\|\Psi(\cdot;\beta_0)\|_{\infty}^2}{m^3B_n^2}\sum_{ij,kl}\IE\left(\int_0^T\int_0^s\left|\phi_{n,ij}^{ij,kl}(t)\right|\frac{B_n^{\{ij,kl\}}C_{n,kl}(s)C_{n,ij}(t)}{p_n(s)^2}dtds\right) \\
\leq&\frac{4C^2\|\alpha_0\|_{\infty}^2\|\Psi(\cdot;\beta_0)\|_{\infty}^2}{mp_nB_n^2}\cdot\frac{1}{m^2}\sum_{ij,kl}\int_0^T\int_0^T\IE\left(\Phi_n(t,t)B_n^{\{ij,kl\}}\frac{C_{n,kl}(s)C_{n,ij}(t)}{p_n(t)p_n(s)}\right)dtds \\
\leq&\frac{4C^2\|\alpha_0\|_{\infty}^2\|\Psi(\cdot;\beta_0)\|_{\infty}^2}{mp_nB_n^2}\cdot\underset{t,s\in[0,T]}{\sup_{i,j,k,l\in V_n,}}\IE\left(\Phi_n(t,t)B_n^{\{ij,kl\}}|C_{n,ij}(t)C_{n,kl}(s)=1\right) \\
&\times\frac{1}{m^2}\sum_{ij,kl}\int_0^T\int_0^T\frac{\IP(C_{n,kl}(s)C_{n,ij}(t)=1)}{p_n(t)p_n(s)}dtds\to0
\end{align*}
again by boundedness of the $\infty$-norms by (C), since $N/mp_m=O(1)$ by (SP) and due to \eqref{eq:assump1} and \eqref{eq:wc3}. For the third term, we apply the Cauchy-Schwarz Inequality:
\begin{align*}
\eqref{eq:bp3}\leq&\frac{2}{mp_nB_n^2}\left(\IE\left(\sum_{ij}\int_0^T\frac{CB_n^{ij}}{mp_n(t)}\Phi_n(t,t)\lambda_{n,ij}(t)dt\right)^2\IE\left(\sum_{kl}\int_0^T\frac{2CB_n^{kl}}{mp_n(t)}dM_{n,kl}(t)\right)^2\right)^{\frac{1}{2}}
\end{align*}
For the first term we get
\begin{align*}
&\IE\left(\sum_{ij}\int_0^T\frac{CB_n^{ij}}{mp_n(t)}\Phi_n(t,t)\lambda_{n,ij}(t)dt\right)^2 \\
\leq&C^2\|\alpha_0\|_{\infty}^2\|\Psi(\cdot;\beta_0)\|_{\infty}^2\frac{1}{m^2}\sum_{ij,kl}\int_0^T\int_0^T\IE\left(B_n^{ij}B_n^{kl}\Phi_n(t,t)\Phi_n(s,s)\frac{C_{n,ij}(t)C_{n,kl}(s)}{p_n(t)p_n(s)}\right)dtds \\
\leq&C^2\|\alpha_0\|_{\infty}^2\|\Psi(\cdot;\beta_0)\|_{\infty}^2 \\
&\quad\times\frac{1}{m^2}\sum_{ij,kl}\int_0^T\int_0^T\IE\left(B_n^{ij}B_n^{kl}\Phi_n(t,t)\Phi_n(s,s)|C_{n,ij}(t)C_{n,kl}(s)=1\right) \\
&\quad\qquad\times\frac{\IP(C_{n,ij}(t)C_{n,kl}(s)=1)}{p_n(t)p_n(s)}dtds \\
=&O(1)
\end{align*}
by the boundedness assumptions (C), \eqref{eq:assump3}, \eqref{eq:wc3} and for the second part we can use the uncorrelated martingales to obtain
\begin{align*}
&\IE\left(\sum_{kl}\int_0^T\frac{2CB_n^{kl}}{mp_n(t)}dM_{n,kl}(t)\right)^2=\sum_{kl}\IE\left(\int_0^T\frac{4C^2\left(B_n^{kl}\right)^2}{m^2p_n(t)^2}\lambda_{n,kl}(t)dt\right) \\
\leq&\frac{4C^2\|\alpha_0\|_{\infty}\|\Psi(\cdot;\beta_0)\|_{\infty}}{mp_n}\int_0^T\IE\left(\left(B_n^{kl}\right)^2|C_{n,kl}(t)=1\right)dt\to0
\end{align*}
by the boundedness assumptions (C) and \eqref{eq:assump4}. Overall we conclude that $\eqref{eq:bp3}\to0$. Finally, for \eqref{eq:bp4} we see that
\begin{align*}
\eqref{eq:bp4}\leq&\frac{4}{mp_nB_n^2}\left(\IE\left(\sum_{ij}\int_0^T\frac{CB_n^{ij}}{mp_n(t)}\Phi_n(t,t)\lambda_{n,ij}(t)dt\right)^2\IE\left(\sum_{kl}\int_0^T\frac{2CB_n^{kl}}{mp_n(t)}\lambda_{n,kl}(t)dt\right)^2\right)^{\frac{1}{2}}.
\end{align*}
The first term above is identical to the first term in the bound of \eqref{eq:bp3} and we have just proven that it remains bounded. The second term remains bounded by (C), \eqref{eq:assump7} and \eqref{eq:wc3}. Convergence to zero follows since $mp_nB_n^2\to\infty$. Thus we have proven overall that also the second term from Proposition 2.9 in \citet{K20} converges to zero. Finally, we have to deal with the third expression. This can be estimated as follows using Lemma \ref{lem:phi_bound}
\begin{align}
&\frac{1}{m^2B_n^2}\sum_{ij,kl}\IE\left(\int_0^T\left(\phi_{n,ij}(t)-\phi_{n,ij}^{ij,kl}(t)\right)dM_{n,ij}(t)\int_0^T\left(\phi_{n,kl}(t)-\phi_{n,kl}^{ij,kl}(t)\right)dM_{n,kl}(t)\right) \nonumber \\
\leq&\frac{1}{m^2B_n^2}\sum_{ij,kl}\IE\left(\int_0^T\frac{2C B_n^{\{ij,kl\}}}{mp_n(t)^2}d|M_{n,ij}|(t)\int_0^T\frac{2C B_n^{\{ij,kl\}}}{r_np_n(t)^2}d|M_{n,kl}|(t)\right) \nonumber \\
\leq&\frac{4C^2}{m^2p_n^2}\IE\left(\left(\sum_{ij}\int_0^T\frac{\left(B_n^{ij}\right)^2}{mp_n(t)}d|M_{n,ij}|(t)\right)^2\right) \nonumber \\
\leq&\frac{8C^2}{m^2p_n^2}\IE\left(\left(\sum_{ij}\int_0^T\frac{\left(B_n^{ij}\right)^2}{mp_n(t)}dM_{n,ij}(t)\right)^2\right)+\frac{32C^2}{m^2p_n^2}\IE\left(\left(\sum_{ij}\int_0^T\frac{\left(B_n^{ij}\right)^2}{mp_n(t)}\lambda_{n,ij}(t)dt\right)^2\right). \label{eq:prop29cond3}
\end{align}
The two expressions above are identical to the two terms which appear in the upper-bound of \eqref{eq:bp3} but $B_n^{ij}$ has to be replaced by $\left(B_n^{ij}\right)^2$ and $\Phi_n(t,t)$ has to be removed. But these changes do not affect the general argument and so we can repeat the estimates following \eqref{eq:bp3} and obtain
\begin{align*}
\eqref{eq:prop29cond3}\leq&\frac{8C^2}{m^2p_n^2}\cdot\frac{\|\alpha_0\|_{\infty}\|\Psi(\cdot;\beta_0)\|_{\infty}}{mp_n}\int_0^T\IE\left(\left(B_n^{kl}\right)^4|C_{n,kl}(t)=1\right)dt \\
&+\frac{32C^2}{m^2p_n^2}\cdot \|\alpha_0\|_{\infty}^2\|\Psi(\cdot;\beta_0)\|_{\infty}^2 \\
&\quad\times\frac{1}{m^2}\sum_{ij,kl}\int_0^T\int_0^T\IE\left(\left(B_n^{ij}B_n^{kl}\right)^2|C_{n,ij}(t)C_{n,kl}(s)=1\right)\frac{\IP(C_{n,ij}(t)C_{n,kl}(s)=1)}{p_n(t)p_n(s)}dtds
\end{align*}
which converges to zero by boundedness of the $\infty$-norms in (C) and the supposed behavior of $B_n^{ij}$ in \eqref{eq:assump5} and \eqref{eq:assump6} as well as \eqref{eq:wc3}. Hence, we can employ Proposition 2.9 in \citet{K20} to have convergence of \eqref{eq:abcdef} to zero.

For \eqref{eq:abcdefg} we denote $\ell_n(s):=4K^{(2)}w(s)^2\mu_n(s;\beta_0)^{-3}\alpha_0(s)$ and rewrite \eqref{eq:abcdefg} as follows (the first line is the definition as a reminder)
\begin{align}
&\eqref{eq:abcdefg} \nonumber \\
=&\frac{1}{mB_n}\sum_{(i,j)\in V_n\times V_n}\int_0^T\phi_{n,ij}(s)\Psi(X_{n,ij}(s);\beta_0)C_{n,ij}(s)\alpha_0(s)ds \nonumber \\
=&\frac{1}{B_n}\sum_{(i,j)\in V_n\times V_n}\int_0^T\frac{N^2}{m^3p_n(s)^3}\left(\frac{m^2p_n(s)^3}{N^2}\phi_{n,ij}(s)-\ell_n(s)\right)\Psi(X_{n,ij}(s);\beta_0)C_{n,ij}(s)\alpha_0(s)ds \label{eq:pimlico1} \\
&+\frac{1}{B_n}\int_0^T\frac{N^2}{m^2p_n(s)^2}\left(\frac{1}{mp_n(s)}\overline{\Psi}_n(s;\beta_0)-\mu_n(s;\beta_0)\right)\ell_n(s)\alpha_0(s)ds \label{eq:pimlico2} \\
&+\frac{1}{B_n}\int_0^T\frac{N^2}{m^2p_n(s)^2}\mu_n(s;\beta_0)\ell_n(s)\alpha_0(s)ds. \label{eq:pimlico3}
\end{align}
The three parts above can be estimated as follows
\begin{align*}
|\eqref{eq:pimlico1}|\leq\frac{N^2}{B_nm^2p_n^2}\sup_{ij,s}\left|\frac{m^2p_n(s)^3}{N^2}\phi_{n,ij}(s)-\ell_n(s)\right|\int_0^T\frac{1}{mp_n(s)}\overline{\Psi}_n(s;\beta_0)\alpha_0(s)ds\to0
\end{align*}
because we have $\inf_{n\in\IN} B_n>0$, $N^2/m^2p_n^2=O(1)$ by (SP), by Lemma \ref{lem:phi_convergence} $\frac{m^2p_n(s)^3}{N^2}\phi_{n,ij}(s)$ converges uniformly in $s$ and $(i,j)$ to $\ell_n(s)$ and $\frac{1}{mp_n(s)}\overline{\Psi}_n(s;\beta_0)$ is uniformly bounded since $\IP(\mathcal{A}_n(c_n,0))\to1$ for some suitable $c_n\to\infty$. Next we estimate
$$|\eqref{eq:pimlico2}|\leq\frac{T\|\ell_n\|_{\infty}\|\alpha_0\|_{\infty}}{B_n}\frac{N^2}{m^2p_n^2}\left\|\frac{1}{mp_n(\cdot)}\overline{\Psi}_n(\cdot;\beta_0)-\mu_n(\cdot;\beta_0)\right\|_{\infty}=o_P(1)$$
by the same arguments. Finally, by definition $\eqref{eq:pimlico3}=1$ and we conclude that $\eqref{eq:abcdefg}\to1$ in probability. We have thus completed the proof of $\eqref{eq:nonzero_var}\to1$ in probability. We have to show now that \eqref{eq:zero_var} converges to zero.

To this end, we split \eqref{eq:zero_var} in two parts and obtain
\begin{align}
\eqref{eq:zero_var}=&\frac{1}{mB_n}\underset{(i,j)\neq (i',j')}{\sum_{(i,j),(i',j')\in V_n\times V_n}}\int_0^T\int_{s_1-2h}^{s_1-}2\phi_{n,ij,i'j'}(s_1,s_2)dM_{n,i'j'}(s_2)dM_{n,ij}(s_1) \label{eq:zero_var1} \\
&\quad+ \frac{1}{mB_n}\sum_{(i,j)\in V_n\times V_n}\int_0^T\int_{s_1-2h}^{s_1-}2\phi_{n,ij,ij}(s_1,s_2)dM_{n,ij}(s_2)dM_{n,ij}(s_1). \label{eq:zero_var2}
\end{align}
We show that \eqref{eq:zero_var1} converges to zero by Theorem 2.11 in \citet{K20} where we put $\tilde{\phi}_{n,ij,kl}^J=\frac{1}{B_n}\phi_{n,ij,kl}^J$ and $\delta_n=h$. The momentary m-Dependence holds by assumption and we have left to show the five conditions (4)-(8). We begin with condition (5): We apply firstly Lemma \ref{lem:phi_bound} and use then Itô's Lemma for semi-martingales (cf. e.g. \citet{CE15} Theorem 14.2.4 together with Definition 11.1.4) followed by some basic bounds and the Cauchy-Schwarz Inequality
\begin{align}
&|(5)| \nonumber \\
=&\Bigg|\IE\Bigg(\frac{1}{m^2B_n^2}\underset{ij\neq i'j', kl\neq k'l'}{\sum_{ij,i'j',kl,k'l'}}\int_0^T\int_{t-2h}^{t-}\left(\phi_{n,ij,i'j'}^{\{ij,i'j'\}}(t,r)-\phi_{n,ij,i'j'}^{\{ij,i'j',kl,k'l'\}}(t,r)\right)dM_{n,i'j'}(r)dM_{n,ij}(t) \nonumber \\
&\quad\times\int_0^T\int_{t-2h}^{t-}\left(\phi_{n,ij,i'j'}^{\{kl,k'l'\}}(t,r)-\phi_{n,kl,k'l'}^{\{ij,i'j',kl,k'l'\}}(t,r)\right)dM_{n,k'l'}(r)dM_{n,kl}(t)\Bigg)\Bigg| \nonumber \\
\leq&\IE\Bigg(\frac{1}{m^2B_n^2}\underset{ij\neq i'j', kl\neq k'l'}{\sum_{ij,i'j',kl,k'l'}}\int_0^T\int_{t-2h}^{t-}\frac{2C}{mp_n(t)p_n(r)}B_n^{\{kl,k'l'\}}d|M_{n,i'j'}|(r)d|M_{n,ij}|(t) \nonumber \\
&\quad\times\int_0^T\int_{t-2h}^{t-}\frac{2C}{mp_n(t)p_n(r)}B_n^{\{ij,i'j'\}}d|M_{n,k'l'}|(r)d|M_{n,kl}|(t)\Bigg) \nonumber \\
\leq&\IE\left(\left(\underset{ij\neq i'j'}{\sum_{ij,i'j'}}\int_0^T\int_{t-2h}^{t-}\frac{2C\cdot B_n^{\{ij,i'j'\}}}{m^2p_n(t)p_n(r)B_n}d|M_{n,i'j'}|(r)d|M_{n,ij}|(t)\right)^2\right) \label{eq:four} \\
=&\IE\Bigg(\sum_{ij,kl}\int_0^T2\int_0^{s-}\sum_{k'l'\neq kl}\int_{r-2h}^{r-}\frac{2C\cdot B_n^{\{kl,k'l'\}}}{m^2p_n(t)p_n(r)B_n}d|M_{n,k'l'}|(t)d|M_{n,kl}|(r) \nonumber \\
&\quad\times\sum_{i'j'\neq ij}\int_{s-2h}^{s-}\frac{2C\cdot B_n^{\{ij,i'j'\}}}{m^2p_n(t)p_n(r)B_n}d|M_{n,i'j'}|(t) \lambda_{n,ij}(s)ds\Bigg) \nonumber \\
&+\IE\left(\sum_{ij}\int_0^T\left(\sum_{i'j'\neq ij}\int_{r-2h}^{r-}\frac{2C\cdot B_n^{\{ij,i'j'\}}}{m^2p_n(r)p_n(t)B_n}d|M_{n,i'j'}|(t)\right)^2dN_{n,ij}(r)\right) \nonumber \\
\leq&\frac{4C^2\|\alpha_0\|_{\infty}\|\Psi(\cdot;\beta_0)\|_{\infty}}{B_n^2} \nonumber \\
&\times\Bigg\{\sum_{ij}\int_0^T\frac{2p_n(s)}{mp_n}\sqrt{\IE\left(\left(\frac{1}{mp_n}\underset{i'j'\neq ij}{\sum_{i'j'}}\int_{s-2h}^{s-}B_n^{\{ij,i'j'\}}d|M_{n,i'j'}|(t)\right)^2\Bigg|C_{n,ij}(s)=1\right)} \nonumber \\
&\quad\times\sqrt{\IE\left(\left(\frac{1}{m^2p_n^2}\underset{k'l'\neq kl}{\sum_{kl,k'l'}}\int_0^{s-}\int_{r-2h}^{r-}B_n^{\{kl,k'l'\}}d|M_{n,k'l'}|(t)d|M_{n,kl}|(r)\right)^2\Bigg|C_{n,ij}(s)=1\right)}ds \nonumber \\
&+\sum_{ij}\int_0^T\frac{p_n(r)}{m^2p_n^2}\IE\left(\left(\frac{1}{mp_n}\underset{i'j'\neq ij}{\sum_{i'j'}}\int_{r-2h}^{r-}B_n^{\{ij,i'j'\}}d|M_{n,i'j'}|(t)\right)^2\Bigg|C_{n,ij}(r)=1\right)dr\Bigg\} \nonumber
\end{align}
which converges to zero by Assumption (C), \eqref{eq:assump10}, and \eqref{eq:assump13}. We continue with (4): By Lemma \ref{lem:phi_bound} we may estimate for any $\epsilon>0$ by Markov's Inequality
\begin{align*}
&\IP\left(|(4)|>\epsilon\right) \\
=&\IP\left(\left|\frac{1}{mB_n}\underset{ij\neq kl}{\sum_{ij,kl}}\int_0^T\int_{t-2h}^{t-}\left(\phi_{n,ij,kl}(t,r)-\phi_{n,ij,kl}^{\{ij,kl\}}(t,r)\right)dM_{n,kl}(r)dM_{n,ij}(t)\right|>\epsilon\right) \\
\leq&\IP\left(\frac{1}{mB_n}\underset{ij\neq kl}{\sum_{ij,kl}}\int_0^T\int_{t-2h}^{t-}\frac{2C}{mp_n(t)p_n(r)}B_n^{\{ij,kl\}}d|M_{n,kl}|(r)d|M_{n,ij}|(t)>\epsilon\right) \\
\leq&\frac{1}{\epsilon^2}\IE\left(\left(\underset{ij\neq kl}{\sum_{ij,kl}}\int_0^T\int_{t-2h}^{t-}\frac{2CB_n^{\{ij,kl\}}}{m^2p_n(t)p_n(r)B_n}d|M_{n,kl}|(r)d|M_{n,ij}|(t)\right)^2\right).
\end{align*}
The above equals \eqref{eq:four} which we have just proven to converge to zero. We turn now to (6). Denote to this end by $F_{ij}(t)=\{(k,l): d_t^n(ij,kl)\geq M\}$ the set of all pairs which are \emph{far away} from $(i,j)$ at time $t$ and let $C_0:=\|\alpha_0\|_{\infty}\|\Psi(\cdot;\beta_0)\|_{\infty}$. For simplicity of notation we will write
$\mathcal{J}_{kl}(t):=C_0\left(4h\sup_{\rho\in[t-2h,t+2h]}C_{n,kl}(\rho)+N_{n,kl}[t-2h,t+2h]\right)$. With this we denote the neighbouring events ${\rm NE}_{ij}(t)$ of a pair $(i,j)$ and all events ${\rm AE}(t)$ as follows
\begin{align*}
{\rm NE}_{ij}(t):=&\sum_{k,l\in V_n}B_n^{kl}\mathcal{J}_{kl}(t)\Ind(kl\notin F_{ij}(t-2h)), \qquad{\rm AE}(t):=\sum_{k',l'\in V_n}B_n^{k'l'}\mathcal{J}_{k'l'}(t).
\end{align*}
Let moreover $\neg$ denote the negation operator. With this notation we obtain for all $t\in[0,T]$ and all $i,j\in V_n$ by using that $B_n^{kl,k'l'}\leq B_n^{kl}B_n^{k'l'}$
\begin{align*}
&\sum_{kl,k'l'}\int_t^{t+2h}\int_{\xi-2h}^{\xi}B_n^{kl,k'l'}\Ind(\neg kl,k'l'\in F_{ij}(t-2h))d|M_{n,k'l'}|(\rho)d|M_{n,kl}|(\xi) \\
\leq&C_0^2\sum_{kl,k'l'}B_n^{kl,k'l'}\mathcal{J}_{kl}(t)\mathcal{J}_{k'l'}(t)\Ind(\neg kl,k'l'\in F_{ij}(t-2h)) \\
\leq&C_0^2\sum_{kl}B_n^{kl}\mathcal{J}_{kl}(t)\Ind( kl\notin F_{ij}(t-2h))\sum_{k'l'}B_n^{k'l'}\mathcal{J}_{k'l'}(t)\Ind(k'l'\in F_{ij}(t-2h)) \\
&+C_0^2\sum_{kl}B_n^{kl}\mathcal{J}_{kl}(t)\Ind(kl\in F_{ij}(t-2h))\sum_{k'l'}B_n^{k'l'}\mathcal{J}_{k'l'}(t)\Ind(k'l'\notin F_{ij}(t-2h)) \\
&+C_0^2\sum_{kl}B_n^{kl}\mathcal{J}_{kl}(t)\Ind(kl\notin F_{ij}(t-2h))\sum_{k'l'}B_n^{k'l'}\mathcal{J}_{k'l'}(t)\Ind(k'l'\notin F_{ij}(t-2h)) \\
\leq&C_0^2\left({\rm NE}_{ij}(t){\rm AE}(t)+{\rm AE}(t){\rm NE}_{ij}(t)+{\rm NE}_{ij}(t)^2\right)\leq 3C_0^2{\rm AE}(t){\rm NE}_{ij}(t)
\end{align*}
because by definition ${\rm NE}_{ij}(t)\leq{\rm AE}(t)$. Using all of this we can bound expression (6) of Theorem 2.11 in \citet{K20} as follows by using firstly Lemma \ref{lem:phi_bound}
\begin{align*}
&|(6)| \\
=&\Bigg|\frac{2}{m^2B_n^2}\underset{ij\neq i'j', kl\neq k'l'}{\sum_{ij,i'j',kl,k'l'}}\IE\Bigg[\int_0^T\int_{t-2h}^{t-}\left(\phi_{n,ij,i'j'}^{\{ij,i'j'\}}(t,r)-\phi_{n,ij,i'j'}^{\{ij,i'j',kl,k'l'\}}(t,r)\right)dM_{n,i'j'}(r)\\
&\times\int_t^{t+2h}\int_{\xi-2h}^{\xi-}\phi_{n,kl,k'l'}^{\{ij,i'j',kl,k'l'\}}(\xi,\rho)dM_{n,k'l'}(\rho)dM_{n,kl}(\xi)\Ind(\neg kl,k'l'\in F_{ij}(t-2h))dM_{n,ij}(t)\Bigg]\Bigg| \\
\leq&\frac{4C^2}{m^3p_n^3B_n^2}\underset{ij\neq i'j', kl\neq k'l'}{\sum_{ij,i'j',kl,k'l'}}\IE\Bigg[\int_0^T\int_{t-2h}^{t-}B_n^{\{kl,k'l'\}}d|M_{n,i'j'}|(r) \\
&\times\int_t^{t+2h}\int_{\xi-2h}^{\xi-}\Phi_n(\xi,\rho)d|M_{n,k'l'}|(\rho)d|M_{n,kl}|(\xi)\cdot\Ind(\neg kl,k'l'\in F_{ij}(t-2h))d|M_{n,ij}|(t)\Bigg] \\
\leq&\frac{4C^2}{m^3p_n^3B_n^2}\underset{ij\neq i'j'}{\sum_{ij,i'j'}}\IE\Bigg[\underset{\rho\in[\xi-2h,\xi]}{\sup_{\xi\in[0,T]}}\Phi_n(\xi,\rho)\int_0^T\int_{t-2h}^{t-}d|M_{n,i'j'}|(r) \\
&\times\underset{kl\neq k'l'}{\sum_{kl,k'l'}}\int_t^{t+2h}\int_{\xi-2h}^{\xi-}B_n^{kl,k'l'}\Ind(\neg kl,k'l'\in F_{ij}(t-2h))d|M_{n,k'l'}|(\rho)d|M_{n,kl}|(\xi)d|M_{n,ij}|(t)\Bigg] \\
\leq&\frac{12C^2}{m^3p_n^3B_n^2}\underset{ij\neq i'j'}{\sum_{ij,i'j'}}\IE\Bigg[\underset{\rho\in[\xi-2h,\xi]}{\sup_{\xi\in[0,T]}}\Phi_n(\xi,\rho)\int_0^T\int_{t-2h}^{t-}d|M_{n,i'j'}|(r)C_0^2{\rm AE}(t){\rm NE}_{ij}(t)d|M_{n,ij}|(t)\Bigg].
\end{align*}
Hence, (6) converges to zero by \eqref{eq:assump11}. We turn now to expression (7) in Theorem 2.11 of \citet{K20}. We bound (7) firstly by noting that the indicator $\Ind(kl\in F_{ij}(t-2h))$ is often equal to one and hence we may simply ignore it and by applying the usual bounds on $\lambda_{n,ij}$ and from Lemma \ref{lem:phi_bound}
\begin{align*}
&|(7)| \\
=&\underset{ij\neq kl}{\sum_{ij,kl}}\int_0^T\int_{t-2h}^{t-}\IE\left[\frac{\phi_{n,ij,kl}^{\{ij,kl\}}(t,r)^2}{m^2B_n^2}C_{n,ij}(t)\lambda_{n,ij}(t)C_{n,kl}(r)\lambda_{n,kl}(r)\Ind(kl\in F_{ij}(t-2h))\right]drdt \\
\leq&\frac{C\|\alpha_0\|_{\infty}^2\|\Psi(\cdot;\beta_0)\|_{\infty}^2}{B_n^2}\int_0^T\int_{t-2h}^{t-}\IE\left[\Phi_n(t,r)^2\cdot\frac{1}{m^2p_n(r)^2}\underset{ij\neq kl}{\sum_{ij,kl}}C_{n,ij}(t)C_{n,kl}(r)\right]drdt.
\end{align*}
The above converges to zero because by assumption $\inf_{n\in\IN}B_n>0$ and the supremum over all $n,t,r$ of the expectation is bounded by \eqref{eq:wc7}. Lastly, we have to deal with (8) in Theorem 2.11 of \citet{K20}. To this end we note firstly that
$$\sum_{ij}C_{n,ij}(t)\Ind(\neg i'j',kl\in F_{ij}(t-2h))\leq K_M^{i'j'}+K_M^{kl}\leq B_n^{i'j'}B_n^{kl}.$$
Using this and Lemma \ref{lem:phi_bound} we obtain
\begin{align*}
&|(8)| \\
=&\Bigg|\frac{1}{m^2B_n^2}\underset{ij\neq i'j'}{\sum_{ij,i'j'}}\underset{kl\neq i'j'}{\sum_{kl}}\int_0^T\IE\Bigg[\int_{t-2h}^{t-}\phi_{n,ij,i'j'}^{\{ij,i'j',kl\}}(t,r)dM_{n,i'j'}(r)\int_{t-2h}^{t-}\phi_{n,ij,kl}^{\{ij,i'j',l\}}(t,r')dM_{n,kl}(r') \\
&\quad\quad\times C_{n,ij}(t)\lambda_{n,ij}(t)\Ind(\neg i'j',kl\in F_{ij}(t-2h))\Bigg]dt\Bigg| \\
\leq&\frac{C^2\|\alpha_0\|_{\infty}\|\Psi(\cdot;\beta_0)\|_{\infty}}{m^2B_n^2}\sum_{i'j'}\sum_{kl}\int_0^T\IE\Bigg[\sum_{ij}C_{n,ij}(t)\Ind(\neg i'j',kl\in F_{ij}(t-2h)) \\
&\quad\quad\times\int_{t-2h}^{t-}\frac{1}{p_n(r)}\Phi_n(t,r)d|M_{n,i'j'}|(r)\int_{t-2h}^{t-}\frac{1}{p_n(r')}\Phi_n(t,r')d|M_{n,kl}|(r')\Bigg]dt \\
\leq&\frac{C^2\|\alpha_0\|_{\infty}\|\Psi(\cdot;\beta_0)\|_{\infty}}{B_n^2}\int_0^T\IE\Bigg[\left(\sum_{kl}\int_{t-2h}^{t-}\frac{B_n^{kl}}{mp_n(r')}\Phi_n(t,r')d|M_{n,kl}|(r')\right)^2\Bigg]dt
\end{align*}
which converges to zero by Assumption (C) and \eqref{eq:assump12}. Thus we have shown all conditions of Theorem 2.11 in \citet{K20} and we can hence use this theorem to conclude that \eqref{eq:zero_var1} converges to zero in probability.

In order to see that \eqref{eq:zero_var2} converges to zero we apply the bounds from Lemma \ref{lem:phi_bound}. More precisely, let $\gamma,\delta>0$ be arbitrary and choose $\kappa>0$ such that $\IP(\sup_{t\in[0,T]}\Phi_n(t)>\kappa)\leq\delta$. Then, we firstly make crude bounds in the probability such that we can apply Markov's Inequality and use martingale properties:
\begin{align*}
&\IP(|\eqref{eq:zero_var2}|>\gamma) \\
=&\IP\left(\left|\frac{1}{mB_n}\sum_{ij}\int_0^T\int_{s_1-2h}^{s_1-}2\phi_{n,ij,ij}(s_1,s_2)dM_{n,ij}(s_2)dM_{n,ij}(s_1)\right|>\gamma\right) \\
\leq&\IP\left(\frac{1}{mB_n}\sum_{(i,j)}\int_0^T\int_{s_1-2h}^{s_1-}\frac{2C}{p_n(s_2)}\Phi_n(s_1,s_2)d|M_{n,ij}|(s_2)d|M_{n,ij}|(s_1)>\gamma\right) \\
\leq&\IP\left(\frac{2C\kappa}{mB_n}\sum_{(i,j)}\int_0^T\int_{s_1-2h}^{s_1-}\frac{1}{p_n(s_2)}d|M_{n,ij}|(s_2)d|M_{n,ij}|(s_1)>\gamma\right)+\delta \\
\leq&\frac{2C\kappa}{mB_n\gamma}\IE\left(\sum_{(i,j)}\int_0^T\int_{s_1-2h}^{s_1-}\frac{1}{p_n(s_2)}d|M_{n,ij}|(s_2)d|M_{n,ij}|(s_1)\right)+\delta \\
\leq&\frac{2C\kappa}{B_n\gamma}\IE\Bigg(\int_0^T\left(\int_{s_1-2h}^{s_1-}\frac{1}{p_n(s_2)}dM_{n,ij}(s_2)+\int_{s_1-2h}^{s_1}\frac{2C_{n,ij}(s_2)}{p_n(s_2)}\|\alpha_0\|_{\infty}\|\Psi(\cdot;\beta_0)\|_{\infty}ds_2\right) \\
&\quad\quad\times\|\alpha_0\|_{\infty}\|\Psi(\cdot;\beta_0)\|_{\infty}ds_1\Bigg)+\delta \\
=&\frac{8TC\kappa\|\alpha_0\|_{\infty}^2\|\Psi(\cdot;\beta_0)\|_{\infty}^2}{B_n\gamma}h+\delta.
\end{align*}
The above converges to $\delta$ since $h\to0$. Since $\delta$ is arbitrary we conclude that $\eqref{eq:zero_var2}=o_P(1)$ and hence we have finally shown that \eqref{eq:zero_var} converges to zero in probability. From this we see that the variation process of \eqref{eq:distr1} converges to one which is the first condition of Rebolledo's Martingale Centrale Limit Theorem (Theorem \ref{thm:Rebolledo}). The second condition is concerned with the jump process of \eqref{eq:distr1} and will be our concern next.

For an arbitrary but fixed $\eta>0$ the jump process associated with \eqref{eq:distr1} is given by
\begin{align*}
M_{\eta}(t):=&\sum_{kl}\int_0^T\Ind\left(\left|\frac{2N}{\sqrt{hB_n}}\underset{ij\neq kl}{\sum_{ij}}\int_0^{r-}\frac{f_n(r,s)}{m^2p_n(s)p_n(r)\mu_n(s;\beta_0)\mu_n(r;\beta_0)}dM_{n,ij}(s)\right|>\eta\right) \\
&\quad\quad\times\frac{2N}{\sqrt{hB_n}}\underset{ij\neq kl}{\sum_{ij}}\int_0^{r-}\frac{f_n(r,s)}{m^2p_n(s)p_n(r)\mu_n(s;\beta_0)\mu_n(r;\beta_0)}dM_{n,ij}(s)dM_{n,kl}(r).
\end{align*}
because the martingales do not jump simultaneously, its corresponding variation process equals
\begin{align*}
\langle M_{\eta}\rangle(t):=&\sum_{kl}\int_0^T\Ind\left(\left|\frac{2N}{\sqrt{hB_n}}\underset{ij\neq kl}{\sum_{ij}}\int_0^{r-}\frac{f_n(r,s)}{m^2p_n(s)p_n(r)\mu_n(s;\beta_0)\mu_n(r;\beta_0)}dM_{n,ij}(s)\right|>\eta\right) \\
&\quad\quad\times\frac{4N^2}{hB_n}\left(\underset{ij\neq kl}{\sum_{ij}}\int_0^{r-}\frac{f_n(r,s)}{m^2p_n(s)p_n(r)\mu_n(s;\beta_0)\mu_n(r;\beta_0)}dM_{n,ij}(s)\right)^2\lambda_{n,kl}(r)dr.
\end{align*}
Note that up to the indicator function the above expression is identical to the variation process of \eqref{eq:distr1} which we have just proven to converge. Hence, we have $\langle M_{\eta}\rangle_T=o_P(1)$ if we can prove that (mind the supremum in the indicator)
$$\Ind\left(\underset{r\in[0,T]}{\sup_{k,l\in V_n}}\left|\frac{2N}{\sqrt{hB_n}}\underset{ij\neq kl}{\sum_{ij}}\int_0^{r-}\frac{f_n(r,s)}{m^2p_n(s)p_n(r)\mu_n(s;\beta_0)\mu_n(r;\beta_0)}dM_{n,ij}(s)\right|>\eta\right)=o_P(1).$$
Proving this is significantly easier than handling the previous expressions because we have no double integral with respect to two martingales. We can make the following estimate: (recall the notation $\eta_0:=\inf_{t\in[0,T]}\mu_n(t;\beta_0)$)
\begin{align}
&\underset{r\in[0,T]}{\sup_{k,l\in V_n}}\left|\frac{2N}{\sqrt{hB_n}}\underset{ij\neq kl}{\sum_{ij}}\int_0^{r-}\frac{f_n(r,s)}{m^2p_n(s)p_n(r)\mu_n(s;\beta_0)\mu_n(r;\beta_0)}dM_{n,ij}(s)\right| \nonumber \\
\leq&\sup_{r\in[0,T]}\frac{2N}{\sqrt{hB_n}}\sum_{ij}\int_0^{r-}\frac{f_n(r,s)}{m^2p_n(s)p_n(r)\mu_n(s;\beta_0)\mu_n(r;\beta_0)}d|M_{n,ij}|(s) \nonumber \\
\leq&\sup_{r\in[0,T]}\frac{2N}{\sqrt{hB_n}}\sum_{ij}\int_0^{r-}\frac{f_n(r,s)}{m^2p_n(s)p_n(r)\mu_n(s;\beta_0)\mu_n(r;\beta_0)}dM_{n,ij}(s) \nonumber \\
&+\sup_{r\in[0,T]}\frac{2N}{\sqrt{hB_n}}\sum_{ij}\int_0^{r-}\frac{2\lambda_{n,ij}(s) f_n(r,s)}{m^2p_n(s)p_n(r)\mu_n(s;\beta_0)\mu_n(r;\beta_0)}ds \nonumber \\
=&\sup_{r\in[0,T]}\frac{2N}{\sqrt{hB_n}}\sum_{ij}\int_0^{r-}\frac{f_n(r,s)}{m^2p_n(s)p_n(r)\mu_n(s;\beta_0)\mu_n(r;\beta_0)}dM_{n,ij}(s) \label{eq:last1} \\
&+\frac{4N\|\alpha_0\|_{\infty}}{mp_n\cdot\eta_0^2\sqrt{hB_n}}\sup_{s\in[0,T]}\frac{\overline{\Psi}_n(s;\beta_0)}{mp_n(s)}\sup_{r\in[0,T]}\int_0^rf_n(r,s)ds \label{eq:last2}
\end{align}
We note firstly that \eqref{eq:last2} is $o_P(1)$ for two main reasons: Firstly, $N/mp_n=O(1)$ and, secondly, we assume that for any given $\delta>0$ there is $c>0$ such that $\IP(\tilde{\mathcal{A}}(c,0))\geq1-\delta$ which in turn implies
$$\sup_{s\in[0,T]}\frac{1}{mp_n(t)}\overline{\Psi}_n(s;\beta_0)\leq\|\mu_n(\cdot;\beta_0)\|_{\infty}+c\sqrt{\frac{\log m}{m}}$$
with probability larger than or equal to $1-\delta$. The boundedness of the infinity norm and $\log m/m\to0$ together with the fact that $f_n$ localizes like a kernel and causes the integral to be uniformly $O(h)$ shows that $\eqref{eq:last2}=o_P(1)$.

For \eqref{eq:last1} we note that we take the supremum of a martingale. Hence, we can apply Doobs Martingale inequality which applies to all right continuous martingales (cf. Theorem 5.1.3 in \citet{CE15}) to obtain:
\begin{align*}
&\IE\left(\left(\sup_{r\in[0,T]}\frac{2N}{\sqrt{hB_n}}\sum_{ij}\int_0^{r-}\frac{f_n(r,s)}{m^2p_n(s)p_n(r)\mu_n(s;\beta_0)\mu_n(r;\beta_0)}dM_{n,ij}(s)\right)^2\right) \\
\leq&4\sup_{r\in[0,T]}\IE\left(\left(\frac{2N}{\sqrt{hB_n}}\sum_{ij}\int_0^{r-}\frac{f_n(r,s)}{m^2p_n(s)p_n(r)\mu_n(s;\beta_0)\mu_n(r;\beta_0)}dM_{n,ij}(s)\right)^2\right) \\
\leq&\frac{16N^2\|\alpha_0\|_{\infty}\|\Psi(\cdot;\beta_0)\|_{\infty}}{m^2p_n^2hB_n\eta_0^4}\int_0^T\frac{f_n(r,s)^2}{mp_n(s)}ds.
\end{align*}
The above converges to zero by boundedness of $N^2/m^2p_n^2$ by (SP) and the infinity norms together with the fact that $f_n$ acts like a kernel and compensates $h^{-1}$. Moreover, $\inf_{n\in\IN}B_n>0$ and $mp_n\to\infty$. Hence, the proof of $\langle M_{\eta}\rangle_T=o_P(1)$ is complete and we may apply Rebolledo's Martingale Central Limit Theorem (Theorem \ref{thm:Rebolledo}) which concludes the proof of the proposition.
\end{proof}

\begin{lemma}
\label{lem:contXbar}
Suppose that $\left|\Psi(x;\beta_1)-\Psi(x;\beta_2)\right|\leq L_{\Psi}\|\beta_1-\beta_2\|$. Then,
$$\left|\mu_n(t;\beta_1)-\mu_n(t;\beta_2)\right|\leq L_{\Psi}\|\beta_1-\beta_2\|.$$
\end{lemma}
\begin{proof}[Proof of Lemma \ref{lem:contXbar}]
We have
\begin{align*}
&\left|\mu_n(t;\beta_1)-\mu_n(t;\beta_2)\right|\leq\IE\left(\left|\Psi(X_{n,ij}(t),\beta_1)-\Psi(X_{n,ij}(t),\beta_2)\right|\big|C_{n,ij}(t)=1\right) \\
\leq&L_{\Psi}\|\beta_1-\beta_2\|
\end{align*}
\end{proof}

\begin{lemma}
\label{lem:phi_bound}
Let (SP), (B), (KBW), (C) and (LL) hold and denote
$$\Phi_n(t,r):=\int_{t}^{r+2h}\frac{1}{hmp_n(s)}\overline{\Psi}_n(s;\beta_0)ds.$$
There is a constant $C>0$ such that for $I\subseteq J\subseteq V_n\times V_n$
\begin{align}
\left|\phi_{n,ij,i'j'}^I(t,r)-\phi_{n,ij,i'j'}^J(t,r)\right|&\leq C\frac{|J\setminus I|}{mp_n(t)p_n(r)}\sup_{(a,b)\in J\setminus I}B_n^{ab}, \label{eq:phibound1}\\
\left|\phi_{n,ij,i'j'}^I(t,r)\right|&\leq\frac{C}{p_n(r)}\Phi_n(t,r). \label{eq:phibound2}
\end{align}
Moreover, $\IP(\sup\Phi_n(t,r)>C)\to0$, where the $\sup$ is taken over all $t,r\in[0,T]$ with $t-2h\leq r\leq t$.
\end{lemma}
\begin{proof}
Since, $I\subseteq J$, we get for any $(k,l)\in V_n\times V_n$
\begin{align*}
&|\Ind(d_{t-4h}^n(kl,I)\geq M)-\Ind(d_{t-4h}^n(kl,J)\geq M)| \\
\leq& \Ind(\exists (i,j)\in J\setminus I: d_{t-4h}^n(kl,ij)\leq M)\leq\sum_{(i,j)\in J\setminus I}\Ind(d_{t-4h}^n(kl,ij)\leq M).
\end{align*}
Then, we obtain for all $t\in[4h,T-2h]$ and all $s\in[t-4h,t+2h]$
\begin{align*}
&\underset{(k,l)\neq (i,j),(i',j')}{\sum_{(k,l)\in V_n\times V_n}}C_{n,kl}(s)\left|\Ind(d_{t-4h}^n(kl,I)\geq M)-\Ind(d_{t-4h}^n(kl,J)\geq M)\right| \\
\leq&|J\setminus I|\sup_{(a,b)\in J\setminus I}\left(\Ind\left(H_{\rm hub}^{ab}=0\right)+\Ind\left(H_{\rm hub}^{ab}=1\right)\right) \\
&\quad\quad\quad\quad\times\underset{{k,l}\neq (i,j),(i',j')}{\sum_{(k,l)\in V_n\times V_n}}C_{n,kl}(s)\Ind(d_{t-4h}^n(kl,ab)\leq M) \\
\leq&|J\setminus I|\cdot \sup_{(a,b)\in J\setminus I}\left(n_{\rm hub}+H_{\rm hub}^{ab}K_m\right)=|J\setminus I|\sup_{(a,b)\in J\setminus I}B_n^{ab},
\end{align*}
recalling the definition of $B_n^{ab}$. Using this and $|f_n(r,s)|\leq\|K\|_{\infty}$, we get
\begin{align*}
&\left|\phi_{n,ij,i'j'}^I(t,r)-\phi_{n,ij,i'j'}^J(t,r)\right| \\
\leq&\frac{4N^2\|K\|_{\infty}^2\|\alpha_0\|_{\infty}\|\Psi(\cdot;\beta_0)\|_{\infty}}{hm^3p_n^2p_n(t)p_n(r)\eta_0^4} \\
&\quad\quad\quad\quad\times\int_{t}^{r+2h}\underset{(k,l)\neq(i,j),(i',j')}{\sum_{(k,l)\in V_n\times V_n}}C_{n,kl}(s)\left|\Ind(d_{t-4h}^n(kl,I)\geq M)-\Ind(d_{t-4h}^n(kl,J)\geq M)\right|ds \\
\leq&\frac{8N^2\|K\|_{\infty}^2\|\alpha_0\|_{\infty}\|\Psi(\cdot;\beta_0)\|_{\infty}}{m^3p_n^2p_n(t)p_n(r)\eta_0^4}|J\setminus I|\sup_{(a,b)\in J\setminus I}B_n^{ab},
\end{align*}
where $\eta_0:=\inf_{t\in[0,T]}\mu_n(t;\beta_0)>0$ by (B). Use now that by (SP) $N/m^2p_n^2=O(1)$ and that by (C) and (KBW) all $\infty$-norms are bounded to conclude \eqref{eq:phibound1}.

In order to show \eqref{eq:phibound2} we bound the $\phi_{n,ij,i'j'}$ by the direct bounds:
\begin{align*}
\left|\phi_{n,ij,i'j'}^I(t,r)\right|\leq\frac{4N^2\|K\|_{\infty}^2\|\alpha_0\|_{\infty}}{m^2p_n^2p_n(r)\eta_0^4}\int_{t}^{r+2h}\frac{1}{hmp_n(s)}\overline{\Psi}_n(s;\beta_0)ds.
\end{align*}
We conclude \eqref{eq:phibound2} by the same boundedness assumptions as for \eqref{eq:phi_tilde_bound1} and replacing the definition of $\Phi_n$.

By Assumption (LL), we find sequences $\delta_n\to0$ and $c_n\to\infty$ such that $\IP(\tilde{\mathcal{A}}_n(c_n,0))\geq1-\delta_n$ and $c_n\sqrt{\log m/mp_n}\to0$. We have on $\mathcal{A}_n(c,0)$ for all $t,r\in[0,T]$ with $t-2h\leq r\leq t$
\begin{align*}
\Phi_n(t,r)\leq 4\sup_{s\in[0,T]}\frac{1}{mp_n(s)}\overline{\Psi}_n(s;\beta_0)\leq4\left(\left\|\mu_n(\cdot;\beta_0)\right\|_{\infty}+c\sqrt{\frac{\log m}{mp_n}}\right)=O(1).
\end{align*}
This completes the proof.
\end{proof}

We also need a result on the asymptotic behavior of the $\phi$'s.
\begin{lemma}
\label{lem:phi_convergence}
Let (LL), (C), (KBW), (SP) and (B) hold. Then we have that
$$\sup_{s\in[0,T]}\sup_{(i,j)\in V_n\times V_n}\left|\frac{m^2p_n^3(s)}{N^2}\phi_{n,ij}(s)-\frac{4w(s)^2\alpha_0(s)}{\mu_n(s;\beta_0)^3}K^{(2)}\right|=o_P(1),$$
where
$$K^{(2)}:=\frac{1}{2}\left\|\left(K\star K\right)\right\|_2^2.$$
\end{lemma}
\begin{proof}
Let $\delta>0$ be arbitrary and let $c>0$ be so large such that $\IP(\tilde{\mathcal{A}}_n(c,0))\geq1-\delta$. Now, on $\tilde{\mathcal{A}}_n(c,0)$ we have for any $r\in[0,T]$ and any $(i,j)\in V_n\times V_n$
\begin{align*}
&\left|\frac{1}{mp_n(r)}\sum_{(k,l)\neq(i,j)}\Psi(X_{n,kl}(r);\beta_0)C_{n,kl}(r)-\mu_n(r;\beta_0)\right| \\
\leq&\left|\frac{1}{mp_n(r)}\overline{\Psi}_n(r;\beta_0)-\mu_n(r;\beta_0)\right|+\frac{C_{n,ij}(r)}{mp_n(r)}\Psi(X_{n,ij}(r);\beta_0) \\
\leq&c\sqrt{\frac{\log m}{mp_n(r)}}+\frac{1}{mp_n(r)}\left\|\Psi(\cdot;\beta_0)\right\|_{\infty}\leq c^*\sqrt{\frac{\log m}{mp_n}},
\end{align*}
where $c^*>0$ is a suitable constant depending on $c$ and $\|\Psi(\cdot;\beta_0)\|_{\infty}$ (recall that $mp_n\to\infty$). Recall moreover the notation $\eta_0:=\inf_{t\in[0,T]}\mu_n(t;\beta_0)$. Thus, by conditioning on $\tilde{\mathcal{A}}_n(c,0)$, for any $\gamma>0$
\begin{align*}
&\IP\left(\sup_{s\in[0,T]}\sup_{(i,j)\in V_n\times V_n}\left|\frac{m^2p_n^3(s)}{N^2}\phi_{n,ij}(s)-\frac{4w(s)^2\alpha_0(s)}{\mu_n(s;\beta_0)^3}K^{(2)}\right|>\gamma\right) \\
\leq&\IP\Bigg(\sup_{s\in[0,T]}\sup_{(i,j)\in V_n\times V_n}\frac{4p_n(s)}{h}\int_s^{s+2h}f_n(r,s)^2\frac{\mu_n(r;\beta_0)^{-2}\mu_n(s;\beta_0)^{-2}\alpha_0(r)}{p_n(r)} \\
&\quad\times\left|\frac{1}{mp_n(r)}\sum_{(k,l)\neq(i,j)}\Psi(X_{n,kl}(r);\beta_0)C_{n,kl}(r)-\mu_n(r;\beta_0)\right|dr \\
&+\sup_{s\in[0,T]}\Bigg|\frac{4p_n(s)}{h}\int_s^{s+2h}f_n(r,s)^2\frac{\mu_n(r;\beta_0)^{-1}\mu_n(s;\beta_0)^{-2}\alpha_0(r)}{p_n(r)}dr \\
&\quad\quad\quad\quad-\frac{4w(s)^2\alpha_0(s)}{\mu_n(s;\beta_0)^3}K^{(2)}\Bigg|>\gamma\Bigg) \\
\leq&\IP\Bigg(\sup_{s\in[0,T]}\frac{4\|\pi\|_{\infty}\|\alpha_0\|_{\infty}}{h\eta_0^4}\int_s^{s+2h}f_n(r,s)^2c^*\sqrt{\frac{\log m}{mp_n}}dr \\
&+\sup_{s\in[0,T]}\Bigg|\frac{4p_n(s)}{h}\int_s^{s+2h}f_n(r,s)^2\frac{\mu_n(r;\beta_0)^{-1}\mu_n(s;\beta_0)^{-2}\alpha_0(r)}{p_n(r)}dr \\
&\quad\quad\quad\quad-\frac{4w(s)^2\alpha_0(s)}{\mu_n(s;\beta_0)^3}K^{(2)}\Bigg|>\gamma\Bigg)+\delta.
\end{align*}
Both lines in the probability above are actually not random, so we finish the proof by proving convergence to zero of both lines. For the first line, use $|f_n(r,s)|\leq \|K\|_{\infty}<\infty$ as well as $\|\alpha_0\|_{\infty}<\infty$ and continuity of $\pi$ with $\log m/mp_n\to0$.

The second line follows from standard kernel arguments: Under (SP) and symmetrie of $K$ it holds that
\begin{align}
&\frac{4p_n(s)}{h}\int_s^{s+2h}\frac{f_n(r,s)^2\mu_n(r;\beta_0)^{-1}\mu_n(s;\beta_0)^{-2}\alpha_0(r)}{p_n(r)}dr \nonumber \\
=&\frac{4\pi(s)}{h\mu_n(s;\beta_0)^2}\int_s^{s+2h}f_n(r,s)^2\frac{\alpha_0(r)}{\pi(r)\mu_n(r;\beta_0)}dr \nonumber \\
=&\frac{4\pi(s)}{\mu_n(s;\beta_0)^2}\int_0^2\left(\int_{-\frac{s}{h}}^{\frac{T-s}{h}}K\left(v\right)K\left(u-v\right)w(s+vh)dv\right)^2\frac{\alpha_0(s+uh)}{\pi(s+uh)\mu_n(s+uh;\beta_0)}du \label{eq:kernelexp}
\end{align}
Since $w$ is continuous (and hence uniformly continuous on $[0,T]$) and in addition supported on $[\delta,T-\delta]$ for some $\delta>0$, we see that for $h\leq\delta$
\begin{align*}
&\int_{-\frac{s}{h}}^{\frac{T-s}{h}}K\left(v\right)K\left(u-v\right)w(s+vh)dv \\
=&w(s)\int_{-\frac{s}{h}}^{\frac{T-s}{h}}K\left(v\right)K\left(u-v\right)dv+\int_{-\frac{s}{h}}^{\frac{T-s}{h}}K\left(v\right)K\left(u-v\right)(w(s+vh)-w(s))dv \\
=&w(s)\int_{-1}^1K\left(v\right)K\left(u-v\right)dv+\int_{-\frac{s}{h}}^{\frac{T-s}{h}}K\left(v\right)K\left(u-v\right)(w(s+vh)-w(s))dv
\end{align*}
and the second term converges uniformly to zero by continuity of $w$ and $h\to0$. We also have
$$\left|\alpha_0(s+uh)-\alpha_0(s)\right|\leq\left|\alpha(\theta_0,s+uh)-\alpha(\theta_0,s)\right|+c_n\left|\Delta_n(s+uh)-\Delta_n(s)\right|.$$
By uniform continuity of $\alpha(\cdot,\theta_0)$, $\sup_{n\in\IN}\|\Delta_n\|_{\infty}<\infty$ and $c_n\to0$ we also have convergence of the above to zero uniformly in $u$ and $s$. Moreover, continuity of $\pi$ and $\mu_n(\cdot;\beta_0)$ on the compact interval $[0,T]$ implies uniform continuity also of $\pi(s)^{-1}$ and $\mu_n(s;\beta_0)^{-1}$ because both are bounded from below. This implies that the integrand in \eqref{eq:kernelexp} converges uniformly and we conclude in turn that uniformly over $s\in[0,T]$
$$\eqref{eq:kernelexp}-\frac{4w(s)^2\alpha_0(s)}{\mu_n(s;\beta_0)^3}\int_0^2\left(\int_{-1}^1K\left(v\right)K\left(u-v\right)dv\right)^2du\to0.$$
We finish the proof by noting that by symmetry of $K$
$$(K\star K)(u)=\int_{-1}^1K(v)K(u-v)dv=\int_{-1}^1K(-v)K(u+v)dv=(K\star K)(-u).$$
This, in turn, implies together with $(K\star K)(u)=0$ for $u>2$ that
$$\int_0^2\left(\int_{-1}^1K(v)K(u-v)dv\right)^2du=\int_0^2\left(K\star K\right)(u)^2du=\frac{1}{2}\int_{-\infty}^{\infty}\left(K\star K\right)(u)^2du=K^{(2)}.$$
\end{proof}

\subsection{Supporting Lemmas}
\label{sec:support_lemmas}
We assume throughout that the model from Theorem \ref{thm:T1} as well as Assumption (VX) hold and do not mention this explicitly every time.

\begin{lemma}
\label{lem:I1}
Let (LL), (B), (C), (KBW), (SP) and \eqref{eq:wc8} hold. Then,
$$N\sqrt{h}\int_0^TI_1(t)^2w(t)dt-N\sqrt{h}\int_0^T\left(\int_0^TK_{h,t}(s)\frac{1}{mp_n(s)}\mu_n(s;\beta_0)^{-1}dM_n(s)\right)^2w(t)dt=o_P(1).$$
\end{lemma}

\begin{lemma}
\label{lem:I2}
Let $\sup_{n\in\IN}\|\Delta_n\|_{\infty}<+\infty$ and let (KBW) hold. Then,
\begin{equation}
N\sqrt{h}\int_0^TI_2(t)^2w(t)dt=O(1). \label{eq:I2}
\end{equation}
\end{lemma}

\begin{lemma}
\label{lem:I3}
Let (C), (P) and (KBW) hold. Then,
\begin{equation}
N\sqrt{h}\int_0^TI_3(t)^2w(t)dt=o_P(1). \label{eq:I3}
\end{equation}
\end{lemma}

\begin{lemma}
\label{lem:I4}
Let (LL), (P), (B), (SP), (KBW), (C) and \eqref{eq:wc8} hold. Then,
\begin{align}
&N\sqrt{h}\int_0^TI_4(t)^2w(t)dt=o_P(1). \label{eq:I4}
\end{align}
\end{lemma}

\begin{lemma}
\label{lem:I12}
Let (LL), (SP), (KBW), (C), (B) hold and let $\sup_{n\in\IN}\|\Delta_n\|_{\infty}<\infty$. Then,
$$N\sqrt{h}\int_0^TI_1(t)I_2(t)w(t)dt=o_P(1).$$
\end{lemma}

\begin{lemma}
\label{lem:I13}
Let (P), (LL), (B), (C), (KBW) and (SP) hold. Then,
$$N\sqrt{h}\int_0^TI_1(t)I_3(t)w(t)dt=o_P(1).$$
\end{lemma}

\begin{lemma}
\label{lem:I14}
Let (P), (LL), (C), (KBW), (B), (SP) and \eqref{eq:wc8} hold. Then,
$$N\sqrt{h}\int_0^TI_1(t)I_4(t)w(t)dt=o_P(1).$$
\end{lemma}

\begin{proof}[Proof of Lemma \ref{lem:I1}]
We rewrite the quantity of interest as follows
\begin{align}
&N\sqrt{h}\int_0^TI_1(t)^2w(t)dt-N\sqrt{h}\int_0^T\left(\int_0^TK_{h,t}(s)\left(mp_n(s)\mu_n(s;\beta_0)\right)^{-1}dM_n(s)\right)^2w(t)dt \nonumber \\
=&N\sqrt{h}\int_0^T\left(\int_0^TK_{h,t}(s)\frac{1}{mp_n(s)}\left(\left(\frac{1}{mp_n(s)}\overline{\Psi}_n(s;\beta_0)\right)^{-1}-\mu_n(s;\beta_0)^{-1}\right)dM_n(s)\right)^2w(t)dt \label{eq:b2} \\
&\quad+2N\sqrt{h}\int_0^T\int_0^TK_{h,t}(s)\frac{1}{mp_n(s)}\mu_n(s;\beta_0)^{-1}dM_n(s) \nonumber \\
&\quad\quad\times\int_0^TK_{h,t}(s)\frac{1}{mp_n(s)}\left(\left(\frac{1}{mp_n(s)}\overline{\Psi}_n(s;\beta_0)\right)^{-1}-\mu_n(s;\beta_0)^{-1}\right)dM_n(s)w(t)dt. \label{eq:b3}
\end{align}
Part \eqref{eq:b2} converges to zero in probability as we show next. We estimate
\begin{align*}
&\eqref{eq:b2} \\
\leq&2N\sqrt{h}\int_0^T\left(\int_0^TK_{h,t}(s)\frac{1}{mp_n(s)}\left(\left(\frac{1}{mp_n(s)}\overline{\Psi}_n(s;\beta_0)\right)^{-1}-\mu_n(s;\beta_0)^{-1}\right)dN_n(s)\right)^2w(t)dt \\
+&2N\sqrt{h}\int_0^T\Bigg(\int_0^TK_{h,t}(s)\frac{1}{mp_n(s)}\left(\left(\frac{1}{mp_n(s)}\overline{\Psi}_n(s;\beta_0)\right)^{-1}-\mu_n(s;\beta_0)^{-1}\right) \\
&\quad\quad\quad\quad\quad\quad\quad\quad\quad\quad\quad\quad\quad\quad\quad\quad\times \overline{\Psi}_n(s;\beta_0)\alpha_0(s)ds\Bigg)^2w(t)dt.
\end{align*}
By assumption we can find for any $\delta>0$ a $c>0$ such that $\IP(\mathcal{A}_n(c,0))\geq1-\delta$ for all $n\in\IN$. This can be used as follows: For any $\gamma>0$
\begin{align}
&\IP(\eqref{eq:b2}>\gamma)\leq\IP(\eqref{eq:b2}>\gamma,\mathcal{A}_n(c,0))+\IP(\mathcal{A}_n(c,0)^c) \nonumber \\
\leq&\IP\Bigg(2N\sqrt{h}\int_0^T\Bigg[\left(\int_0^TK_{h,t}(s)\frac{c\sqrt{\log m}}{(mp_n(s))^{\frac{3}{2}}}dN_n(s)\right)^2 \nonumber \\
&\quad\quad\quad+\left(\int_0^TK_{h,t}(s)\frac{c\sqrt{\log m}}{(mp_n(s))^{\frac{3}{2}}}\overline{\Psi}_n(s;\beta_0)\alpha_0(s)ds\right)^2\Bigg]w(t)dt>\gamma\Bigg)+\delta \nonumber \\
\leq&\IP\Bigg(\int_0^T\Bigg[4N\sqrt{h}\left(\int_0^TK_{h,t}(s)\frac{\sqrt{\log m}}{(mp_n(s))^{\frac{3}{2}}}dM_n(s)\right)^2 \nonumber \\
&\quad\quad\quad+6N\sqrt{h}\left(\int_0^TK_{h,t}(s)\frac{\sqrt{\log m}}{(mp_n(s))^{\frac{3}{2}}}\overline{\Psi}_n(s;\beta_0)\alpha_0(s)ds\right)^2\Bigg]w(t)dt>\frac{\gamma}{c^2}\Bigg)+\delta \nonumber \\
\leq&\frac{2c^2}{\gamma}N\sqrt{h}\int_0^T\int_0^TK_{h,t}(s)\frac{\log m}{mp_n(s)}\alpha_0(s) \nonumber \\
&\quad\quad\times\left[2\frac{K_{h,t}(s)}{mp_n(s)}\underbrace{\IE\left(\frac{\overline{\Psi}_n(s;\beta_0)}{mp_n(s)}\right)}_{=\mu_n(s;\beta_0)}+3\IE\left(\left(\frac{\overline{\Psi}_n(s;\beta_0)}{mp_n(s)}\right)^2\right)\alpha_0(s)\right]dsdt+\delta. \label{eq:breakpoint4}
\end{align}
By Assumptions (B) and \eqref{eq:wc8} the expectations are uniformly bounded. Moreover, by assumption $\|\alpha_0\|_{\infty}+\|K\|_{\infty}<+\infty$. Hence, we find a constant $c^*>0$ such that
\begin{align*}
\eqref{eq:breakpoint4}\leq \frac{c^*}{\gamma}N\sqrt{h}\frac{\log m}{mp_n}\left[\frac{1}{hmp_n}+1\right]+\delta\to\delta
\end{align*}
since $\sqrt{h}\log m\to0$, $hmp_n\to\infty$ by assumption and $N/mp_n=O(1)$ by (SP). Since $\gamma$ and $\delta$ were arbitrary, we conclude that $\eqref{eq:b2}=o_P(1)$.

In order to show that $\eqref{eq:b3}=o_P(1)$, we make the following abbreviations
\begin{align*}
a_n(t,s):=&K_{h,t}(s)\frac{1}{mp_n(s)}\mu_n(s;\beta_0)^{-1}, \\
b_n(t,s):=&K_{h,t}(s)\frac{1}{mp_n(s)}\left(\left(\frac{1}{mp_n(s)}\overline{\Psi}_n(s;\beta_0)\right)^{-1}-\mu_n(s;\beta_0)^{-1}\right).
\end{align*}
Then, we get
\begin{align}
&\eqref{eq:b3}=2N\sqrt{h}\int_0^T\int_0^T\int_0^Ta_n(t,s)b_n(t,r)w(t)dtdM_n(s)dM_n(r) \nonumber \\
=&2N\sqrt{h}\int_0^T\int_0^{r-}\int_0^Ta_n(t,s)b_n(t,r)w(t)dtdM_n(s)dM_n(r) \label{eq:s1} \\
&+2N\sqrt{h}\int_0^T\int_{r+}^T\int_0^Ta_n(t,s)b_n(t,r)w(t)dtdM_n(s)dM_n(r) \label{eq:s2} \\
&+2N\sqrt{h}\int_0^T\int_{\{r\}}\int_0^Ta_n(t,s)b_n(t,r)w(t)dtdM_n(s)dM_n(r). \label{eq:s3} 
\end{align}
We begin by treating \eqref{eq:s1}. Note that the two inner integrals are a predictable function of $r$ and therefore, we may treat the stochastic integral like a martingale. From Lenglart's Inequality (see equation (2.15.18) in \citet{ABGK93}) we see that we may alternatively prove that the quadratic variation converges to zero in probability. Hence, we prove in the following that the quadratic variation of \eqref{eq:s1} converges to zero in probability. Let to this end $\delta>0$ be arbitrary and choose $c>0$ so large such that $\IP(\mathcal{A}_n(c,0)\cap\tilde{\mathcal{A}}_n(c,0))\geq1-\delta$ for all $n\in\IN$. For simplicity of notation we let $\xi_n(s):=c\sqrt{\log m /mp_n(s)}$ and note that on $\mathcal{A}_n(c,0)\cap\tilde{\mathcal{A}}_n(c,0)$ we have
\begin{align*}
&\left|\left(\frac{1}{mp_n(s)}\overline{\Psi}_n(s;\beta_0)\right)^{-1}-\mu_n(s;\beta_0)^{-1}\right|\leq \xi_n(s),
&\left|\frac{1}{mp_n(s)}\overline{\Psi}_n(s;\beta_0)-\mu_n(s;\beta_0)\right|\leq \xi_n(s)
\end{align*}
We use these considerations to obtain for every $\gamma>0$
\begin{align*}
&\IP\left(4N^2h\int_0^T\left(\int_0^{r-}\int_0^Ta_n(t,s)b_n(t,r)w(t)dtdM_n(s)\right)^2\alpha_0(r)\overline{\Psi}_n(r,\beta_0)dr>\gamma\right) \\
\leq&\IP\Bigg(4N^2h\int_0^T\left(\int_0^{r-}\int_0^Ta_n(t,s)b_n(t,r)w(t)dtdM_n(s)\right)^2\alpha_0(r)\overline{\Psi}_n(r,\beta_0)dr>\gamma, \\
&\quad\quad\quad\quad\quad\quad\quad\quad\quad\mathcal{A}_n(c,0)\cap\tilde{\mathcal{A}}_n(c,0)\Bigg)+\delta \\
\leq&\IP\Bigg(8N^2h\int_0^T\alpha_0(r)\frac{\overline{\Psi}_n(r,\beta_0)\xi_n(r)^2}{m^2p_n(r)^2}\Bigg[\left(\int_0^{r-}\int_0^T\frac{K_{h,t}(s)K_{h,t}(r)}{mp_n(s)\mu_n(s;\beta_0)}w(t)dtdN_n(s)\right)^2 \\
&+\left(\int_0^{r-}\int_0^T\frac{K_{h,t}(s)K_{h,t}(r)}{mp_n(s)\mu_n(s;\beta_0)}w(t)\alpha_0(s)\overline{\Psi}_n(s;\beta_0)dtds\right)^2\Bigg]dr>\gamma\Bigg)+\delta \\
\leq&\IP\Bigg(8N^2h\int_0^T\alpha_0(r)\frac{\left(\mu_n(r,\beta_0)+\xi_n(r)\right)\xi_n(r)^2}{mp_n(r)} \\
&\qquad\times\Bigg[2\left(\int_0^{r-}\int_0^T\frac{K_{h,t}(s)K_{h,t}(r)}{mp_n(s)\mu_n(s;\beta_0)}w(t)dtdM_n(s)\right)^2 \\
&\qquad\qquad+3\left(\int_0^{r-}\int_0^T\frac{K_{h,t}(s)K_{h,t}(r)}{mp_n(s)\mu_n(s;\beta_0)}w(t)\alpha_0(s)\overline{\Psi}_n(s;\beta_0)dtds\right)^2\Bigg]dr>\gamma\Bigg)+\delta \\
\leq&\frac{1}{\gamma}\Bigg(8N^2h\int_0^T\alpha_0(r)\frac{\left(\mu_n(r,\beta_0)+\xi_n(r)\right)\xi_n(r)^2}{mp_n(r)} \\
&\quad\quad\times\Bigg[2\IE\left(\left(\int_0^{r-}\int_0^T\frac{K_{h,t}(s)K_{h,t}(r)}{mp_n(s)\mu_n(s;\beta_0)}w(t)dtdM_n(s)\right)^2\right) \\
&\quad\quad\quad+\IE\left(3\left(\int_0^{r-}\int_0^T\frac{K_{h,t}(s)K_{h,t}(r)}{mp_n(s)\mu_n(s;\beta_0)}w(t)\alpha_0(s)\overline{\Psi}_n(s;\beta_0)dtds\right)^2\right)\Bigg]dr\Bigg)+\delta \\
\leq&\frac{1}{\gamma}\Bigg(8N^2h\int_0^T\alpha_0(r)\frac{\left(\mu_n(r,\beta_0)+\xi_n(r)\right)\xi_n(r)^2}{mp_n(r)} \\
&\quad\quad\times\Bigg[2\left(\int_0^r\left(\int_0^T\frac{K_{h,t}(s)K_{h,t}(r)}{mp_n(s)\mu_n(s;\beta_0)}w(t)dt\right)^2\alpha_0(s)\IE(\overline{\Psi}_n(s;\beta_0))ds\right) \\
&\quad\quad\quad+3\int_0^r\int_0^TK_{h,t}(s)K_{h,t}(r)\IE\left(\left(\frac{w(t)\alpha_0(s)\overline{\Psi}_n(s;\beta_0)}{mp_n(s)\mu_n(s;\beta_0)}\right)^2\right)dtds\Bigg]dr\Bigg)+\delta \\
\leq&\frac{1}{\gamma}\Bigg(8N^2h\|\alpha_0\|_{\infty}\frac{\left(\|\mu_n(\cdot;\beta_0)\|_{\infty}+C\right)c^2\log m}{m^2p_n^2} \\
&\times\int_0^T\Bigg[2\left(\int_0^r\left(\int_0^TK_{h,t}(s)K_{h,t}(r)dt\right)^2\frac{\|w\|_{\infty}^2\|\alpha_0\|_{\infty}}{mp_n}\|\mu_n(\cdot;\beta_0)^{-1}\|_{\infty}\|\Psi(\cdot;\beta_0)\|_{\infty}ds\right) \\
&\qquad\qquad+3\int_0^r\int_0^TK_{h,t}(s)K_{h,t}(r)\|w\|_{\infty}^2\|\alpha_0\|_{\infty}^2\|\mu_n(\cdot;\beta_0)^{-1}\|_{\infty}^2 \\
&\qquad\qquad\qquad\times\IE\left(\left(\frac{\overline{\Psi}_n(s;\beta_0)}{mp_n(s)}\right)^2\right)dtds\Bigg]dr\Bigg)+\delta,
\end{align*}
where $C>0$ is chosen such that $\xi_n(r)\leq C$. Note that all $\infty$-norms above are bounded by assumption, we assume in \eqref{eq:wc8} moreover that the expectation is bounded and by definition and (SP) $N^2/m^2p_n^2=O(1)$. Use finally $h\log m\to0$, $hmp_n\to\infty$ and
$$\int_{-\infty}^{\infty}\int_{-\infty}^{\infty}K_{h,t}(s)K_{h,t}(r)dtds=1\textrm{ and }\int_{-\infty}^{\infty}\left(\int_{-\infty}^{\infty}K_{h,t}(s)K_{h,t}(r)dt\right)^2ds=O\left(\frac{1}{h}\right)$$
to conclude that the inequality chain above converges to $\delta$ on the right hand side. Since $\gamma>0$ and $\delta>0$ were arbitrary, we have proven that $\eqref{eq:s1}=o_p(1)$. For \eqref{eq:s2} we note that interchanging the order of integration yields
$$\eqref{eq:s2}=2N\sqrt{h}\int_0^T\int_0^{s-}\int_0^Ta_n(t,s)b_n(t,r)w(t)dtdM_n(r)dM_n(s).$$
Hence, we have an expression which is almost identical to \eqref{eq:s1} and may therefore be treated by the same arguments. So we skip the proof of $\eqref{eq:s2}=o_P(1)$. Instead, we detail the arguments for handling \eqref{eq:s3}. Note firstly that in the inner integral with respect to $M_n(s)$ the integral over the intensity can be ignored and only the counting process integral remains. Similarly for the outer integral only the counting process integral remains. In consideration of this we get (let as usually $\gamma,\delta>0$ be arbitrary and $c>0$ such that $\IP(\mathcal{A}_n(c,0))\geq1-\delta$)
\begin{align}
 &\IP\left(\eqref{eq:s3})>\gamma\right)=\IP\left(2N\sqrt{h}\int_0^T\int_0^Ta_n(t,r)b_n(t,r)w(t)dtdN_n(r)>\gamma\right) \nonumber \\
=&\IP\left(2N\sqrt{h}\int_0^T\int_0^Ta_n(t,r)b_n(t,r)w(t)dtdN_n(r)>\gamma,\mathcal{A}_n(c,0)\right)+\delta \nonumber \\
\leq&\IP\left(2N\sqrt{h}\int_0^T\int_0^TK_{h,t}(r)^2\frac{\xi_n(r)w(t)}{m^2p_n(r)^2\mu_n(r;\beta_0)}dtdN_n(r)>\gamma\right)+\delta \nonumber \\
\leq&\frac{2N\sqrt{h}}{\gamma}\int_0^T\int_0^TK_{h,t}(r)^2\frac{\xi_n(r)w(t)}{m^2p_n(r)^2\mu_n(r;\beta_0)}\alpha_0(r)\IE\left(\overline{\Psi}_n(r;\beta_0)\right)dtdr+\delta \nonumber \\
\leq&\frac{2cT\|w\|_{\infty}\|\alpha_0\|_{\infty}\|K\|_{\infty}}{\gamma}\cdot\frac{N}{mp_n}\sqrt{\frac{\log m}{hmp_n}}+\delta \label{eq:breakpoint3}
\end{align}
Again, $N/mp_n=O(1)$ by (SP), the $\infty$-norms are finite and $\log m/hmp_n\to0$, thus $\eqref{eq:breakpoint3}\to\delta$. Since $\gamma>0$ and $\delta>0$ were chosen arbitrarily, we conclude $\eqref{eq:s3}=o_P(1)$ and the proof of the Lemma is complete.
\end{proof}

\begin{proof}[Proof of Lemma \ref{lem:I2}]
We may apply simple direct bounds
\begin{align*}
&N\sqrt{h}\int_0^TI_2(t)^2w(t)dt=N\sqrt{h}\int_0^T\left(c_n\int_0^TK_{h,t}(s)\Delta_n(s)ds\right)^2w(t)dt \\
\leq& N\sqrt{h}c_n^2T\|w\|_{\infty}\|\Delta_n\|_{\infty}^2.
\end{align*}
By assumption $\sup_{n\in\IN}\|\Delta_n\|_{\infty}<\infty$, $\|w\|_{\infty}<\infty$ and $N\sqrt{h}c_n^2=1$ by definition, thus \eqref{eq:I2} is shown.
\end{proof}

\begin{proof}[Proof of Lemma \ref{lem:I3}]
By the Cauchy-Schwarz Inequality and the definition of $L_{\alpha}$
\begin{align*}
&N\sqrt{h}\int_0^TI_3(t)^2w(t)dt=N\sqrt{h}\int_0^T\left(\int_0^TK_{h,t}(s)\left(\alpha(\theta_0,s)-\alpha(\hat{\theta}_n,s)\right)ds\right)^2w(t)dt \\
\leq&N\sqrt{h}\|w\|_{\infty}\int_0^T\int_0^TK_{h,t}(s)\left(\alpha(\theta_0,s)-\alpha(\hat{\theta}_n,s)\right)^2dsdt \\
\leq&N\sqrt{h}\|w\|_{\infty}\cdot\left\|\theta_0-\hat{\theta}_n\right\|^2\int_0^TL_{\alpha}(s)^2ds.
\end{align*}
Let now $\epsilon,\delta,c^*>0$ be arbitrary, it follows from the above that
\begin{align*}
&\IP\left(N\sqrt{h}\int_0^TI_3(t)^2w(t)dt>\epsilon\right) \\
\leq&\IP\left(\left\|\hat{\theta}_n-\theta_0\right\|>c^*N^{-\frac{1}{2}}\right)+\IP\left(N\sqrt{h}\|w\|_{\infty}\cdot\left(c^*N^{-\frac{1}{2}}\right)^2\|L_{\alpha}\|_2^2>\epsilon\right).
\end{align*}
By assumption, we may choose $c^*>0$ such that $\IP(\|\hat{\theta}_n-\theta_0\|>c^*N^{-1/2})\leq\delta$ for all $n\in\IN$. With this choice we obtain
\begin{align*}
&\IP\left(N\sqrt{h}\int_0^TI_3(t)^2w(t)dt>\epsilon\right)\leq\delta+\IP\left(h^{\frac{1}{4}}\cdot c^*N^{-\frac{1}{2}}>\sqrt{\frac{\epsilon}{N\|w\|_{\infty}\|L_{\alpha}\|_2^2}}\right).
\end{align*}
The probability on the right and side equals zero for $n$ large enough by the assumptions on $L_{\alpha}$ in (C) and since $h\to0$. Since $\epsilon,\delta>0$ where chosen arbitrarily, the proof is complete.
\end{proof}

\begin{proof}[Proof of Lemma \ref{lem:I4}]
We firstly insert $\mu_n$, the limit of $\bar{\Psi}_n$, as follows
\begin{align}
 &\frac{1}{3}N\sqrt{h}\int_0^TI_4(t)^2w(t)dt \nonumber \\
=&\frac{1}{3}N\sqrt{h}\int_0^T\left(\int_0^TK_{h,t}(s)\frac{1}{mp_n(s)}\left(\frac{1}{\frac{1}{mp_n(s)}\bar{\Psi}_n(s;\tilde{\beta}_n)}-\frac{1}{\frac{1}{mp_n(s)}\bar{\Psi}_n(s;\beta_0)}\right)dN_n(s)\right)^2w(t)dt \nonumber \\
\leq&        N\sqrt{h}\int_0^T\left(\int_0^TK_{h,t}(s)\frac{1}{mp_n(s)}\left(\frac{1}{\frac{1}{mp_n(s)}\bar{\Psi}_n(s;\tilde{\beta}_n)}-\frac{1}{\mu_n(s;\tilde{\beta}_n)}\right)dN_n(s)\right)^2w(t)dt \label{eq:step1} \\
&           +N\sqrt{h}\int_0^T\left(\int_0^TK_{h,t}(s)\frac{1}{mp_n(s)}\left(\frac{1}{\mu_n(s;\tilde{\beta}_n)}-\frac{1}{\mu_n(s;\beta_0)}\right)dN_n(s)\right)^2w(t)dt \label{eq:step2} \\
&           +N\sqrt{h}\int_0^T\left(\int_0^TK_{h,t}(s)\frac{1}{mp_n(s)}\left(\frac{1}{\frac{1}{mp_n(s)}\bar{\Psi}_n(s;\beta_0)}-\frac{1}{\mu_n(s;\beta_0)}\right)dN_n(s)\right)^2w(t)dt. \label{eq:step3}
\end{align}
The arguments for \eqref{eq:step1} and \eqref{eq:step3} are identical. Therefore, we only detail \eqref{eq:step1}. Recall to this end for $c>0$ the event
\begin{align*}
&B_n(c):=\left\{\left\|\tilde{\beta}_n-\beta\right\|\leq\frac{c}{\sqrt{N}}\right\}.
\end{align*}
Let $\delta>0$ be arbitrary. By the assumptions on $\tilde{\beta}_n$ and $\mathcal{A}(c_1,c_2)$ we can choose $c_1,c_2$ so large such that $\IP(\mathcal{A}_n(c_1,c_2)\cap\mathcal{B}_n(c_2))\geq1-\delta$ for all $n\in\IN$. We keep the choice of $\delta,c_1,c_2$ fixed for the remainder of the proof. Then, we obtain for any $\gamma>0$ (use the Cauchy-Schwarz Inequality in the last line)
\begin{align}
&\IP(\eqref{eq:step1}>\gamma)\leq\IP(\eqref{eq:step1}>\gamma,\mathcal{A}_n(c_1,c_2)\cap B_n(c_2))+\IP\left(\left(\mathcal{A}_n(c_1,c_2)\cap B_n(c_2)\right)^c\right) \nonumber \\
\leq&\IP\left(N\sqrt{h}\int_0^T\left(\int_0^TK_{h,t}(s)\frac{c_1\sqrt{\log m}}{(mp_n(s))^{\frac{3}{2}}}dN_n(s)\right)^2w(t)dt>\gamma\right)+\delta \nonumber \\
\leq&\IP\Bigg(\int_0^T\Bigg[2N\sqrt{h}\left(\int_0^TK_{h,t}(s)\frac{\sqrt{\log m}}{(mp_n(s))^{\frac{3}{2}}}dM_n(s)\right)^2 
\nonumber \\
&\quad\quad\quad+2N\sqrt{h}\left(\int_0^TK_{h,t}(s)\frac{\sqrt{\log m}}{(mp_n(s))^{\frac{3}{2}}}\alpha_0(s)\bar{\Psi}_n(s;\beta_0)ds\right)^2\Bigg]w(t)dt>\frac{\gamma}{c_1^2}\Bigg)+\delta \nonumber \\
\leq&\frac{c_1^2}{\gamma}\int_0^T2N\sqrt{h}\int_0^TK_{h,t}(s)\frac{\log m}{mp_n(s)}\alpha_0(s) \nonumber \\
&\quad\quad\times\left[\frac{K_{h,t}(s)}{mp_n(s)}\IE\left(\frac{\bar{\Psi}_n(s;\beta_0)}{mp_n(s)}\right)+\IE\left(\left(\frac{\bar{\Psi}_n(s;\beta_0)}{mp_n(s)}\right)^2\right)\alpha_0(s)\right]w(t)dsdt+\delta. \label{eq:breakpoint}
\end{align}
By Assumption (B) there is a constant $c^*\in(0,\infty)$ such that $\IE\left(\bar{\Psi}_n(s;\beta_0)/mp_n(s)\right)<c^*$ and by Assumption \eqref{eq:wc8} together with Assumption (SP) (after possibly increasing $c^*$) we also have $\IE\left(\left(\bar{\Psi}_n(s;\beta_0)/mp_n(s)\right)^2\right)<c^*$. Then,
\begin{align*}
\eqref{eq:breakpoint}\leq \frac{2c_1^2}{\gamma}N\sqrt{h}\frac{\log m\cdot\|\alpha_0\|_{\infty}\|w\|_{\infty}}{mp_n}\left[\frac{T\|K\|_{\infty}c^*}{mp_nh}+Tc^*\|\alpha_0\|_{\infty}\right]+\delta\to\delta
\end{align*}
since $N/mp_n=O(1)$ by (SP) and, by Assumption (KBW), $\sqrt{h}\log m\to0$ and $hmp_n\to\infty$ as well as $\|\alpha_0\|_{\infty}, \|K\|_{\infty}, \|w\|_{\infty}<\infty$. Since $\delta,\gamma>0$ were chosen arbitrarily, we have shown that $\eqref{eq:step1}=o_P(1)$. The statement $\eqref{eq:step3}=o_P(1)$ can be shown along the same lines: The only change is that one does not need to condition on the event $B_n(c_2)$ but may choose $c_2=0$ instead.

We have left to prove that $\eqref{eq:step2}=o_P(1)$. Let to this end $\delta,\gamma>0$ be arbitrary and choose $c_2>0$ such that $\IP(B_n(c_2))\geq1-\delta$. By Lemma \ref{lem:contXbar} we get
\begin{align}
&\IP(\eqref{eq:step2}>\gamma)\leq \IP(\eqref{eq:step2}>\gamma,B_n(c_2))+\delta \nonumber \\
\leq& \IP\left(N\sqrt{h}\int_0^T\left(\int_0^TK_{h,t}(s)\cdot\frac{L_{\Psi}c_2N^{-\frac{1}{2}}}{mp_n(s)\cdot\inf_{\beta\in B_n(c_2)}\mu_n(s;\beta)^2}dN_n(s)\right)^2w(t)dt>\gamma\right)+\delta \nonumber \\
\leq&\IP\Bigg(2N\sqrt{h}\left(c_2N^{-\frac{1}{2}}\right)^2\int_0^T\Bigg[\left(\int_0^TK_{h,t}(s)\frac{L_{\Psi}}{mp_n(s)\cdot\inf_{\beta\in B(c_2)}\mu_n(s;\beta)^2}dM_n(s)\right)^2 
\nonumber \\
&\quad\quad\quad+\left(\int_0^TK_{h,t}(s)\frac{L_{\Psi}}{mp_n(s)\cdot\inf_{\beta\in B_n(c_2)}\mu_n(s;\beta)^2}\alpha_0(s)\bar{\Psi}_n(s;\beta_0)ds\right)^2\Bigg]w(t)dt>\gamma\Bigg)+\delta \nonumber \\
\leq&\frac{2}{\gamma}N\sqrt{h}\left(c_2N^{-\frac{1}{2}}\right)^2\int_0^T\int_0^TK_{h,t}(s)\alpha_0(s)\left(\frac{L_{\Psi}}{\inf_{\beta\in B_n(c_2)}\mu_n(s;\beta)^2}\right)^2 \nonumber \\
&\quad\quad\times\Bigg[\frac{K_{h,t}(s)}{mp_n(s)}\IE\left(\frac{\bar{\Psi}_n(s;\beta_0)}{mp_n(s)}\right)+\IE\left(\left(\frac{\bar{\Psi}_n(s;\beta_0)}{mp_n(s)}\right)^2\alpha_0(s)\right)\Bigg]w(t)dsdt+\delta. \label{eq:breakpoint2}
\end{align}
Similarly as before, we may bound the expectations by $c^*$ by assumption. We get
\begin{align*}
\eqref{eq:breakpoint2}\leq&\frac{2\|\alpha_0\|_{\infty}\|w\|_{\infty}}{\gamma}N\sqrt{h}\left(c_2N^{-\frac{1}{2}}\right)^2 \\
&\quad\quad\quad\times\int_0^T\left(\frac{L_{\Psi}}{\inf_{\beta\in B_n(c_2)}\mu_n(s;\beta)^2}\right)^2ds\cdot\left[\frac{\|K\|_{\infty}}{hmp_n}+c^*\|\alpha_0\|_{\infty}\right]+\delta \\
\to&\delta
\end{align*}
because, by Assumption (KBW), $h\to0$, $hmp_n\to\infty$ and since
$$\int_0^T\left(\frac{L_{\Psi}}{\inf_{\beta\in B_n(c_2)}\mu_n(s;\beta)^2}\right)^2ds=O(1)$$
together with the finiteness of all infinity norms. This completes the proof.
\end{proof}

\begin{proof}[Proof of Lemma \ref{lem:I12}]
We have
\begin{align*}
&N\sqrt{h}\int_0^TI_1(t)I_2(t)w(t)dt=N\sqrt{h}c_n\int_0^T\int_0^T\int_0^TK_{h,t}(s)K_{h,t}(r)\frac{\Delta_n(r)w(t)}{\bar{\Psi}_n(s;\beta_0)}dtdrdM_n(s)
\end{align*}
Since the integrand is predictable, the stochastic integral is a martingale. Therefore, we may use Lenglart's Inequality (cf. (2.5.18) in \citet{ABGK93}) to argue that it is enough to prove that the quadratic variation of the process converges to zero. Suppose to this end that $\delta>0$ is arbitrary and that $c>0$ is chosen such that $\IP(\mathcal{A}_n(c,0))\geq1-\delta$. Let $\xi_n(s)=\sqrt{\log m/mp_n(s)}$. On the set $\mathcal{A}_n(c,0)$ it holds for all $s\in[0,T]$ that
$$\left|\left(\frac{1}{mp_n(s)}\bar{\Psi}_n(s;\beta_0)\right)^{-1}-\mu_n(s;\beta_0)^{-1}\right|\leq\xi_n(s).$$
Then we obtain for the quadratic variation and for any $\gamma>0$ (replace also the definition of $c_n$ and use the Cauchy-Schwarz Inequality together with $\int_0^T\int_0^TK_{h,t}(s)K_{h,t}(r)dtdr\leq1$)
\begin{align*}
&\IP\left(N^2hc_n^2\int_0^T\left(\int_0^T\int_0^TK_{h,t}(s)K_{h,t}(r)\frac{\Delta_n(r)w(t)}{\bar{\Psi}_n(s;\beta_0)}dtdr\right)^2\alpha_0(s)\bar{\Psi}_n(s;\beta_0)ds>\gamma\right) \\
\leq&\IP\left(N\sqrt{h}\int_0^T\left(\int_0^T\int_0^TK_{h,t}(s)K_{h,t}(r)\frac{\Delta_n(r)w(t)}{\bar{\Psi}_n(s;\beta_0)}dtdr\right)^2\alpha_0(s)\bar{\Psi}_n(s;\beta_0)ds>\gamma\right) \\
\leq&\IP\left(N\sqrt{h}\int_0^T\int_0^T\int_0^TK_{h,t}(s)K_{h,t}(r)\frac{\Delta_n(r)^2w(t)^2\alpha_0(s)}{\bar{\Psi}_n(s;\beta_0)}dtdrds>\gamma,\mathcal{A}_n(c,0)\right)+\delta \\
\leq&\IP\Bigg(N\sqrt{h}\frac{T\|\Delta_n\|_{\infty}^2\|w\|_{\infty}^2\|\alpha_0\|_{\infty}}{mp_n}\left(\left\|\mu_n(\cdot;\beta_0)^{-1}+\xi_n(\cdot)\right\|_{\infty}\right)>\gamma\Bigg)+\delta.
\end{align*}
Now we see that the expression in the probability is not random and hence the probability equals zero if the inequality inside is violated. This is the case for large $n$ since $N/mp_n=O(1)$ by (SP), $h\to0$, $\log m/mp_n\to0$ and the boundedness assumptions on $\|\Delta_n\|_{\infty}$, $\|\alpha_0\|$ and $\left\|\mu_n(\cdot;\beta_0)^{-1}\right\|_{\infty}$. Thus, we have that the quadratic variation of $N\sqrt{h}\int_0^TI_1(t)I_2(t)w(t)dt$ converges to zero in probability and thus also the martingale itself.
\end{proof}

\begin{proof}[Proof of Lemma \ref{lem:I13}]
Let $\delta>0$ be arbitrary and let $c_1,c_2$ be such that
$$\IP\left(\left\|\hat{\theta}_n-\theta_0\right\|\leq c_2N^{-\frac{1}{2}},\,\mathcal{A}_n(c_1,c_2)\right)\geq1-\delta.$$
This is possible by (P) in conjunction with (LL). By the independence of $\hat{\theta}_n$ from the remaining random quantities as in Assumption (P), we have that
\begin{align*}
&N\sqrt{h}\int_0^TI_1(t)I_3(t)w(t)dt \\
=&N\sqrt{h}\int_0^T\int_0^T\int_0^TK_{h,t}(r)K_{h,t}(s)\left(\alpha_0(\theta_0,r)-\alpha(\hat{\theta}_n,r)\right)\frac{w(t)}{\bar{\Psi}_n(s;\beta_0)}drdtdM_n(s)
\end{align*}
is a martingale. As a consequence we can prove the convergence to zero by applying Lenglart's Inequality (cf. (2.5.18) in \citet{ABGK93}) and prove that the quadratic variation converges to zero in probability. We do this in the following. Let $\gamma>0$ be arbitrary and recall the notation $\xi_n(s)=\sqrt{\log m/mp_n(s)}$. By using the same arguments as in the end of the proof of Lemma \ref{lem:I12} we obtain
\begin{align*}
&\IP\left(N^2h\int_0^T\Bigg(\int_0^T\int_0^TK_{h,t}(r)K_{h,t}(s)\left(\alpha_0(\theta_0,r)-\alpha(\hat{\theta}_n,r)\right)\frac{w(t)}{\bar{\Psi}_n(s;\beta_0)}drdt\right)^2 \\
&\quad\quad\quad\quad\quad\quad\quad\quad\quad\quad\times\alpha_0(s)\bar{\Psi}_n(s;\beta_0)ds>\gamma\Bigg) \\
\leq&\IP\left(N^2h\int_0^T\int_0^T\int_0^TK_{h,t}(r)K_{h,t}(s)\left(\alpha_0(\theta_0,r)-\alpha(\hat{\theta}_n,r)\right)^2w(t)^2\frac{\alpha_0(s)}{\bar{\Psi}_n(s;\beta_0)}drdtds>\gamma\right) \\
\leq&\IP\Bigg(N^2h\int_{[0,T]^3}\frac{K_{h,t}(r)K_{h,t}(s)L_{\alpha}(r)^2\left(c_2N^{-\frac{1}{2}}\right)^2w(t)^2\alpha_0(s)}{mp_n(s)} \\
&\quad\quad\quad\quad\quad\times\left(\mu_n(s;\beta_0)^{-1}+\xi_n(s)\right)drdtds>\gamma\Bigg)+\delta \\
\leq&\IP\Bigg(h\cdot N^2\left(c_2N^{-\frac{1}{2}}\right)^2\cdot\frac{\|w\|_{\infty}^2\|\alpha_0\|_{\infty}}{mp_n}\cdot\left\|\mu_n(\cdot;\beta_0)^{-1}+\xi_n(\cdot)\right\|_{\infty}\|L_{\alpha}\|_2^2>\gamma\Bigg)+\delta.
\end{align*}
The probability equals zero for $n$ large enough because by Assumptions (B) and (C) all infinity norms are bounded, by (KBW) $h\to0$, by (SP) $N/mp_n=O(1)$ and by (C) $\|L_{\alpha}\|_2<\infty$. Since $\delta>0$ was chosen arbitrarily, we conclude that the quadratic variation process converges to zero in probability and the proof is complete.
\end{proof}

\begin{proof}[Proof of Lemma \ref{lem:I14}]
We firstly write
\begin{align}
&N\sqrt{h}\int_0^TI_1(t)I_4(t)w(t)dt \nonumber \\
=&N\sqrt{h}\int_0^T\int_0^T\int_0^TK_{h,t}(s)K_{h,t}(r)\frac{\bar{\Psi}_n(r;\beta_0)-\bar{\Psi}_n(r;\tilde{\beta}_n)}{\bar{\Psi}_n(r;\beta_0)\bar{\Psi}_n(r;\tilde{\beta}_n)\bar{\Psi}_n(s;\beta_0)}w(t)dN_n(r)dtdM_n(s) \nonumber \\
=&N\sqrt{h}\int_0^T\int_0^T\int_0^TK_{h,t}(s)K_{h,t}(r)\frac{\bar{\Psi}_n(r;\beta_0)-\bar{\Psi}_n(r;\tilde{\beta}_n)}{\bar{\Psi}_n(r;\beta_0)\bar{\Psi}_n(r;\tilde{\beta}_n)\bar{\Psi}_n(s;\beta_0)}w(t)dM_n(r)dtdM_n(s) \label{eq:r1} \\
&+N\sqrt{h}\int_0^T\int_0^T\int_0^TK_{h,t}(s)K_{h,t}(r)\frac{\bar{\Psi}_n(r;\beta_0)-\bar{\Psi}_n(r;\tilde{\beta}_n)}{\bar{\Psi}_n(r;\tilde{\beta}_n)\bar{\Psi}_n(s;\beta_0)}\alpha_0(r)w(t)drdtdM_n(s). \label{eq:r2}
\end{align}
We begin with \eqref{eq:r2}. We write
\begin{align}
&\eqref{eq:r2} \nonumber \\
=&N\sqrt{h}\int_0^T\int_0^s\int_0^TK_{h,t}(s)K_{h,t}(r)\frac{\bar{\Psi}_n(r;\beta_0)-\bar{\Psi}_n(r;\tilde{\beta}_n)}{\bar{\Psi}_n(r;\tilde{\beta}_n)\bar{\Psi}_n(s;\beta_0)}\alpha_0(r)w(t)dtdrdM_n(s) \label{eq:r21} \\
&+N\sqrt{h}\int_0^T\int_0^r\int_0^TK_{h,t}(s)K_{h,t}(r)\frac{\bar{\Psi}_n(r;\beta_0)-\bar{\Psi}_n(r;\tilde{\beta}_n)}{\bar{\Psi}_n(r;\tilde{\beta}_n)\bar{\Psi}_n(s;\beta_0)}\alpha_0(r)w(t)dtdM_n(s)dr. \label{eq:r22}
\end{align}
We see that \eqref{eq:r21} is a martingale because the integrand is predictable because by (P) $\tilde{\beta}_n$ is obtained independently of the martingales. Thus, by Lenglart's Inequality (cf. (2.5.18) in \citet{ABGK93}) it is sufficient to prove that the quadratic variation converges to zero. We do this in the following. Let as usual $\delta,\gamma>0$ and choose $c_1,c_2>0$ such that the probability of the event $\|\tilde{\beta}_n-\beta_0\|\leq c_2N^{-1/2}\cap\mathcal{A}_n(c_1,c_2)\cap\tilde{\mathcal{A}}_n(c_1,c_2)$ is at least $1-\delta$ (this possible by the Assumptions (P) and (LL)). On this event, we have by Lemma \ref{lem:contXbar}
\begin{align*}
&\left|\left(\frac{1}{mp_n(r)}\bar{\Psi}_n(r;\tilde{\beta}_n)\right)^{-1}-\mu_n(r,\tilde{\beta}_n)^{-1}\right|\leq c_1\sqrt{\frac{\log m}{mp_n(r)}}, \\
&\left|\frac{1}{mp_n(r)}\bar{\Psi}_n(r;\beta_0)-\mu_n(r,\beta_0)\right|\leq c_1\sqrt{\frac{\log m}{mp_n(r)}}, \\
&\left|\frac{\bar{\Psi}_n(r;\beta_0)-\bar{\Psi}_n(r;\tilde{\beta}_n)}{\bar{\Psi}_n(r;\tilde{\beta}_n)\bar{\Psi}_n(r;\beta_0)}\right|\leq \frac{1}{mp_n}\left(2c_1\sqrt{\frac{\log m}{mp_n}}+\eta_0(c_2)^{-2}L_{\Psi}c_2N^{-\frac{1}{2}}\right),
\end{align*}
where $\eta_0:=\inf_{r\in[0,T],\beta\in K(\beta_0)}\mu_n(r;\beta)$. Then we obtain for the quadratic variation of \eqref{eq:r21}
\begin{align*}
&\IP\Bigg(N^2h\int_0^T\left(\int_0^s\int_0^TK_{h,t}(s)K_{h,t}(r)\frac{\bar{\Psi}_n(r;\beta_0)-\bar{\Psi}_n(r;\tilde{\beta}_n)}{\bar{\Psi}_n(r;\tilde{\beta}_n)\bar{\Psi}_n(s;\beta_0)}\alpha_0(r)w(t)dtdr\right)^2 \\
&\quad\quad\quad\quad\quad\quad\quad\quad\quad\quad\times\alpha_0(s)\bar{\Psi}_n(s;\beta_0)ds>\gamma\Bigg) \\
\leq&\IP\Bigg(N^2h\int_0^T\int_0^s\int_0^TK_{h,t}(s)K_{h,t}(r)\left(\frac{\bar{\Psi}_n(r;\beta_0)-\bar{\Psi}_n(r;\tilde{\beta}_n)}{\bar{\Psi}_n(r;\tilde{\beta}_n)\bar{\Psi}_n(s;\beta_0)}\right)^2\alpha_0(r)^2w(t)^2 \\
&\quad\quad\quad\quad\quad\quad\quad\quad\quad\quad\times\alpha_0(s)\bar{\Psi}_n(s;\beta_0)dtdrds>\gamma\Bigg) \\
\leq&\IP\Bigg(N^2h\|\alpha_0\|_{\infty}^3\|w\|_{\infty}^2\int_0^T\frac{1}{mp_n}\left(2c_1\sqrt{\frac{\log m}{mp_n}}+\eta_0^{-2}L_{\Psi}c_2N^{-\frac{1}{2}}\right)^2dr \\
&\quad\quad\times\left(\left\|\mu_n(\cdot;\beta_0)\right\|_{\infty}+c_1\sqrt{\frac{\log m}{mp_n}}\right)>\gamma\Bigg)+\delta
\end{align*}
The probability above equals zero for $n$ large enough because, by Assumptions (C), (B) and (KBW) all norms and $\eta_0$ are bounded, by Assumption (SP) $N/mp_n=O(1)$, and by (KBW) $h\log m\to0$, $\log m/mp_n\to0$. So we have $\eqref{eq:r21}=o_P(1)$ because $\delta,\gamma>0$ where chosen arbitrarily. The proof of $\eqref{eq:r22}=o_P(1)$ requires momentary-$m$-dependence and is given in Lemma \ref{lem:r22}.

Let us now turn to \eqref{eq:r1}. When handling \eqref{eq:b3} earlier, we defined functions $a_n(t,s)$ and $b_n(t,s)$. We will show that \eqref{eq:r1} and \eqref{eq:b3} have a similar structure. We redefine therefore the functions $a_n(t,s)$ and $b_n(t,s)$ as follows
\begin{align*}
a_n(t,s):=K_{h,t}(s)\bar{\Psi}_n(s;\beta_0)^{-1}\textrm{ and }b_n(t,r):=K_{h,t}(r)\left(\bar{\Psi}_n(r;\tilde{\beta}_n)^{-1}-\bar{\Psi}_n(r;\beta_0)^{-1}\right).
\end{align*}
With these new definitions we obtain
$$2\eqref{eq:r1}=2N\sqrt{h}\int_0^T\int_0^T\int_0^Ta_n(t,s)b_n(t,r)w(t)dM_n(r)dtdM_n(s) =\eqref{eq:s1}+\eqref{eq:s2}+\eqref{eq:s3}$$
because for this step no conditions on $a_n(t,s)$ and $b_n(t,r)$ are required. Thus in the following we will show \eqref{eq:s1}-\eqref{eq:s3} converge to zero in probability for the new definitions of $a_n$ and $b_n$. For handling \eqref{eq:s1} we required that \eqref{eq:s1} is a martingale in $T$. This is the case also with the new definitions of $a_n$ and $b_n$ because $\tilde{\beta}_n$ is independent of the counting processes by (P). When showing that $\eqref{eq:s1}=o_P(1)$ we conditioned on the set $\mathcal{A}(c,0)\cap\tilde{\mathcal{A}}_n(c,0)$. We will now condition on $\mathcal{A}(c_1,c_2)\cap\tilde{\mathcal{A}}_n(c_1,c_2)\cap\{\|\tilde{\beta}_n-\beta_0\|\leq c_2N^{-1/2}\}$ because its probability is larger than $1-\delta$ for any $\delta>0$ after choosing $c_1,c_2>0$ appropriately by (P) and (LL). By following the same arguments as for \eqref{eq:s1}, it is enough to show that for every $c_1,c_2,\gamma>0$ we have that
\begin{align}
&\IP\Bigg(N^2h\int_0^T\left(\int_0^{r-}\int_0^Ta_n(t,s)b_n(t,r)w(t)dtdM_n(s)\right)^2\alpha_0(r)\bar{\Psi}_n(r;\beta_0)dr>\gamma, \nonumber \\
&\quad\quad\quad\quad\mathcal{A}(c_1,c_2),\tilde{\mathcal{A}}_n(c_1,c_2),\{\|\tilde{\beta}_n-\beta_0\|\leq c_2N^{-1/2}\}\Bigg)\to0. \label{eq:eq1}
\end{align}
Note that on $\mathcal{A}(c_1,c_2)\cap\tilde{\mathcal{A}}_n(c_1,c_2)\cap\{\|\tilde{\beta}_n-\beta_0\|\leq c_2N^{-1/2}\}$ we have by Lemma \ref{lem:contXbar}
\begin{align*}
a_n(t,s)\leq&K_{h,t}(s)\frac{1}{mp_n(s)}\left(\mu_n(s;\beta_0)^{-1}+c_1\sqrt{\frac{\log m}{mp_n(s)}}\right)=:\frac{\bar{a}_n(t,s)}{mp_n(s)}, \\
b_n(t,s)\leq&K_{h,t}(s)\frac{1}{mp_n(s)}\left(2c_1\sqrt{\frac{\log m}{mp_n(s)}}+\eta_0^{-2}L_{\Psi}c_2N^{-\frac{1}{2}}\right)=:\frac{\bar{b}_n(t,s)}{mp_n(s)}.
\end{align*}
Recall the notation $\xi_n(s)=c_1\sqrt{\log m/mp_n(s)}$ and let $C>0$ be such that $\xi_n(s)\leq C$. The boundedness of $\|\mu(\cdot;\beta)^{-1}\|_{\infty}$ by Assumption (B) allows to increase $C>0$ such that $\bar{a}_n(t,s)\leq CK_{h,t}(s)$. We may hence bound the probability in \eqref{eq:eq1} by very similar arguments as for \eqref{eq:s1}. We give the main steps below (recall the notation $\xi_n(s)=c_1\sqrt{\log m/mp_n(s)}$ and let $C>0$ be such that $\xi_n(s)\leq C$)
\begin{align*}
&\IP\Bigg(N^2h\int_0^T\left(\int_0^{r-}\int_0^T\frac{\bar{a}_n(t,s)\bar{b}_n(t,r)}{mp_n(s)}w(t)dtd|M_n|(s)\right)^2\frac{\alpha_0(r)\left(\mu_n(r;\beta_0)+\xi_n(r)\right)}{mp_n(r)}dr>\gamma\Bigg) \\
\leq&\frac{2N^2h\|\alpha_0\|_{\infty}\left(\|\mu_n(\cdot;\beta_0)\|_{\infty}+C\right)}{\gamma mp_n} \\
&\quad\times\Bigg[2\int_0^T\int_0^r\IE\Bigg(\left(\int_0^T\frac{\bar{a}_n(t,s)\bar{b}_n(t,r)}{mp_n(s)}w(t)dt\right)^2\alpha_0(s)\bar{\Psi}_n(s;\beta_0)\Bigg)dsdr \\
&\quad\quad+3\int_0^T\IE\Bigg(\left(\int_0^{r-}\int_0^T\frac{\bar{a}_n(t,s)\bar{b}_n(t,r)}{mp_n(s)}w(t)dt\alpha_0(s)\bar{\Psi}_n(s;\beta_0)ds\right)^2\Bigg)dr\Bigg] \\
\leq&\frac{N^2h}{mp_n}\cdot\frac{4C^2\|\alpha_0\|_{\infty}^2\|w\|_{\infty}^2\left(\|\mu_n(\cdot;\beta_0)\|_{\infty}+C\right)}{\gamma} \\
&\quad\times\Bigg[\frac{4\|K\|_{\infty}^2\|\Psi(\cdot;\beta_0)\|_{\infty}}{hmp_n}\left(4Tc_1^2\frac{\log m}{mp_n}+L_{\Psi}^2\eta_0^{-4}\left(c_2N^{-\frac{1}{2}}\right)^2\right) \\
&+\quad\quad3\|\alpha_0\|_{\infty}\sup_{s\in[0,T]}\IE\left(\left(\frac{\bar{\Psi}_n(s;\beta_0)}{mp_n(s)}\right)^2\right)\left(4Tc_1^2\frac{\log m}{mp_n}+L_{\Psi}^2\eta_0^{-4}\left(c_2N^{-\frac{1}{2}}\right)^2\right)\Bigg].
\end{align*}
The above converges to zero because, by (SP) $N/mp_n=O(1)$, by Assumption \eqref{eq:wc8} $\sup_{s\in[0,T]}\IE\left(\bar{\Psi}_n(s;\beta_0)^2/m^2p_n(s)^2\right)=O(1)$, by (KBW) $h\log m\to0$ and by (C), (B) and (KBW) all norms and $\eta_0^{-1}$ are bounded. Thus we have shown that \eqref{eq:s1} converges to zero in probability with the new definitions of $a_n,b_n$. When handling \eqref{eq:s2} we argued that rearranging the integrals yields an expression similar to \eqref{eq:s1}. We can apply the same argument here and show that \eqref{eq:s2} converges to zero in the same way as above. Therefore we skip the details here. Finally, we turn to \eqref{eq:s3}. For the new choices of $a_n,b_n$ we may apply the same arguments as for the old definitions of $a_n,b_n$. We have to show that for any $c_1,c_2,\gamma>0$
\begin{align*}
&\IP\Bigg(N\sqrt{h}\int_0^T\int_0^T\frac{\bar{a}_n(t,r)\bar{b}_n(t,r)}{m^2p_n(r)^2}w(t)dtdN_n(r)>\gamma, \\
&\quad\quad\quad\quad\mathcal{A}(c_1,c_2),\tilde{\mathcal{A}}_n(c_1,c_2),\{\|\tilde{\beta}_n-\beta_0\|\leq c_2N^{-1/2}\}\Bigg)\to0.
\end{align*}
The probability above can be bounded as follows:
\begin{align*}
&\IP\left(N\sqrt{h}\int_0^T\int_0^T \frac{CK_{h,t}(r)^2\left(2c_1\sqrt{\frac{\log m}{mp_n(r)}}+\eta_0^{-2}L_{\Psi}c_2N^{-\frac{1}{2}}\right)
w(t)}{m^2p_n(r)^2}dtdN_n(r)>\gamma\right) \\
=&\frac{N\sqrt{h}}{\gamma}\int_0^T\int_0^T \frac{CK_{h,t}(r)^2\left(2c_1\sqrt{\frac{\log m}{mp_n(r)}}+\eta_0^{-2}L_{\Psi}c_2N^{-\frac{1}{2}}\right)
w(t)\IE\left(\alpha_0(r)\bar{\Psi}_n(r;\beta_0)\right)}{m^2p_n(r)^2}dtdr \\
\leq&\frac{CN\|w\|_{\infty}\|\alpha_0\|_{\infty}\|\Psi(\cdot;\beta_0)\|_{\infty}\|K\|_{\infty}}{\gamma\sqrt{h}mp_n}\left(2Tc_1\sqrt{\frac{\log m}{mp_n}}+\eta_0^{-2}L_{\Psi}c_2N^{-\frac{1}{2}}\right).
\end{align*}
The above converges to zero because, by (C), (B), (KBW) all norms and $\eta_0^{-1}$ are bounded, by (SP) $N/mp_n=O(1)$, by (KBW) $\log m/hmp_n\to0$ and $Nh\to\infty$. This completes the proof.
\end{proof}

The following result uses momentary-$m$-dependence in a similar way as in Proposition \ref{prop:mmd}.
\begin{lemma}
\label{lem:r22}
Let (B), (P), (KBW), (LL), (C), (SP), (mDep) and \eqref{eq:wc1}, \eqref{eq:wc1+1}, \eqref{eq:assump8}, \eqref{eq:assump9} hold. Then,
$$N\sqrt{h}\int_0^T\int_0^r\int_0^TK_{h,t}(s)K_{h,t}(r)\frac{\bar{\Psi}_n(r;\beta_0)-\bar{\Psi}_n(r;\tilde{\beta}_n)}{\bar{\Psi}_n(r;\tilde{\beta}_n)\bar{\Psi}_n(s;\beta_0)}\alpha_0(r)w(t)dtdM_n(s)dr=o_P(1).$$
\end{lemma}
\begin{proof}
Recall that $f_n(r,s):=\int_0^ThK_{h,t}(s)K_{h,t}(r)w(t)dt$. Then, we may rewrite
\begin{align}
&N\sqrt{h}\int_0^T\int_0^r\int_0^TK_{h,t}(s)K_{h,t}(r)\frac{\bar{\Psi}_n(r;\beta_0)-\bar{\Psi}_n(r;\tilde{\beta}_n)}{\bar{\Psi}_n(r;\tilde{\beta}_n)\bar{\Psi}_n(s;\beta_0)}\alpha_0(r)w(t)dtdM_n(s)dr \nonumber \\
=&Nh^{-\frac{1}{2}}\int_0^T\int_0^rf_n(r,s)\frac{\bar{\Psi}_n(r;\beta_0)-\bar{\Psi}_n(r;\tilde{\beta}_n)}{\bar{\Psi}_n(r;\tilde{\beta}_n)\bar{\Psi}_n(s;\beta_0)}\alpha_0(r)dM_n(s)dr \nonumber \\
=&Nh^{-\frac{1}{2}}\int_0^T\int_0^rf_n(r,s)\frac{1}{mp_n(s)}\left(\frac{1}{mp_n(r)}\bar{\Psi}_n(r;\beta_0)-\frac{1}{mp_n(r)}\bar{\Psi}_n(r;\tilde{\beta}_n)\right) \nonumber \\
&\times\left(\left(\frac{1}{mp_n(r)}\bar{\Psi}_n(r;\tilde{\beta}_n)\frac{1}{mp_n(s)}\bar{\Psi}_n(s;\beta_0)\right)^{-1}-\mu_n(r,\tilde{\beta}_n)^{-1}\mu_n(s;\beta_0)^{-1}\right)\alpha_0(r)dM_n(s)dr \label{eq:r221} \\
&+Nh^{-\frac{1}{2}}\int_0^T\int_0^rf_n(r,s)\frac{1}{mp_n(s)}\frac{\frac{1}{mp_n(r)}\bar{\Psi}_n(r;\beta_0)-\frac{1}{mp_n(r)}\bar{\Psi}_n(r;\tilde{\beta}_n)}{\mu_n(r,\tilde{\beta}_n)\mu_n(s;\beta_0)}\alpha_0(r)dM_n(s)dr. \label{eq:r222}
\end{align}
Suppose that we are on the event $\|\tilde{\beta}_n-\beta\|\leq c_2N^{-1/2}$ intersected with $\mathcal{A}_n(c_1,c_2)\cap\tilde{\mathcal{A}}_n(c_1,c_2)$ for $c_1\geq1$ and $c_2>0$ arbitrary. Let also $\eta_0:=\inf_{n\in\IN}\inf_{t\in[0,T]}\inf_{\beta\in K(\beta_0)}\mu_n(t;\beta_0)>0$, by Assumption (B), and choose $n$ so large such that $\sqrt{\log m/mp_n}\leq \min(1,1/\eta_0)$ (possible because $\log m/mp_n\to0$ by (KBW) and (SP)). For the remainder of the proof, $n$ will always be so large such that this relation holds (recall that by assumption $\tilde{\beta}_n\in K(\beta_0)$ with probability equal to one by (P)). Then, by Lemma \ref{lem:contXbar},
\begin{align*}
&\sup_{r\in[0,T]}\left|\frac{1}{mp_n(r)}\bar{\Psi}_n(r;\beta_0)-\frac{1}{mp_n(r)}\bar{\Psi}_n(r;\tilde{\beta}_n)\right|\leq 2c_1\sqrt{\frac{\log m}{mp_n}}+L_{\Psi}c_2N^{-1/2}, \\
&\sup_{r\in[0,T]}\sup_{\beta\in B_n(c_2)}\left(\frac{1}{mp_n(r)}\bar{\Psi}_n(r;\beta)\right)^{-1}\leq c_1\sqrt{\frac{\log m}{mp_n}}+\sup_{r\in[0,T],\beta\in B_n(c_2)}\mu_n(r;\beta)^{-1}\leq\frac{2}{\eta_0}.
\end{align*}
From the last two inequalities we obtain in particular that there is a constant $c^*>0$ (independent of $c_1,c_2$) such that
\begin{align*}
&\left|\left(\frac{1}{mp_n(r)}\bar{\Psi}_n(r;\tilde{\beta}_n)\frac{1}{mp_n(s)}\bar{\Psi}_n(s;\beta_0)\right)^{-1}-\mu_n(r,\tilde{\beta}_n)^{-1}\mu_n(s;\beta_0)^{-1}\right| \\
\leq&\frac{c_1c^*}{\eta_0}\sqrt{\frac{\log m}{mp_n}}.
\end{align*}
Let now $\delta>0$ be arbitrary and choose $c_1\geq1$ and $c_2>0$ such that the event $\|\tilde{\beta}_n-\beta\|\leq c_2N^{-1/2}$ intersected with $\mathcal{A}_n(c_1,c_2)\cap\tilde{\mathcal{A}}_n(c_1,c_2)$ has probability at least $1-\delta$ for all $n\in\IN$ (possible because of (P) and (LL)). Then we obtain for every $\gamma>0$ (note that $f_n$ acts like a kernel in the sense that $f_n(r,s)\neq0$ only if $|r-s|\leq 2h$)
\begin{align*}
&\IP\left(\left|\eqref{eq:r221}\right|>\gamma\right) \\
\leq&\IP\Bigg(Nh^{-\frac{1}{2}}\int_0^T\int_0^rf_n(r,s)\frac{\|\alpha_0\|_{\infty}}{mp_n(s)}\left(2c_1\sqrt{\frac{\log m}{mp_n}}+L_{\Psi}c_2N^{-\frac{1}{2}}\right) \nonumber \\
&\times\frac{c_1c^*}{\eta_0}\sqrt{\frac{\log m}{mp_n}}d|M_n|(s)dr>\gamma\Bigg)+\delta \\
\leq&\sum_{i,j\in V_n}\frac{2}{\gamma}\IE\Bigg(Nh^{-\frac{1}{2}}\int_0^T\int_0^rf_n(r,s)\frac{\|\alpha_0\|_{\infty}}{mp_n(s)}\left(2c_1\sqrt{\frac{\log m}{mp_n}}+L_{\Psi}c_2N^{-\frac{1}{2}}\right) \nonumber \\
&\times\frac{c_1c^*}{\eta_0}\sqrt{\frac{\log m}{mp_n}}\alpha_0(s)\Psi(X_{n,ij}(s);\beta_0)C_{n,ij}(s)dsdr\Bigg)+\delta \\
\leq&\frac{2c_1c^*\|w\|_{\infty}\|\alpha_0\|_{\infty}^2\|\Psi(\cdot;\beta_0)\|_{\infty}}{\gamma\eta_0(c_2)}N\sqrt{\frac{h\cdot\log m}{mp_n}}\left(2Tc_1\sqrt{\frac{\log m}{mp_n}}+L_{\Psi}c_2N^{-\frac{1}{2}}\right)+\delta
\end{align*}
which converges to $\delta$ since, by (C), (KBW) all norms are bounded, by (SP) $N/mp_n=O(1)$ and by (KBW) $\sqrt{h}\log m\to0$. Since $\delta,\gamma>0$ were arbitrarily chosen, we have shown that \eqref{eq:r221} converges to zero. In order to handle \eqref{eq:r222} we have to use momentary-$m$-dependence ideas and define to this end for any set $I\subseteq V_n\times V_n$
\begin{align*}
F_n(r,s):=&\frac{\frac{1}{mp_n(r)}\bar{\Psi}_n(r;\beta_0)-\frac{1}{mp_n(r)}\bar{\Psi}_n(r;\tilde{\beta}_n)}{\mu_n(r;\tilde{\beta}_n)\mu_n(s;\beta_0)}, \\
F_n^I(r,s):=&\frac{\frac{1}{mp_n(r)}\bar{\Psi}_n^I(r,s;\beta_0)-\frac{1}{mp_n(r)}\bar{\Psi}_n^I(r,s;\tilde{\beta}_n)}{\mu_n(r;\tilde{\beta}_n)\mu_n(s;\beta_0)}, \\
\bar{\Psi}_n^I(r,s;\beta):=&\sum_{k,l\in V_n}\Psi(X_{n,kl}(r);\beta)C_{n,kl}(r)\Ind\left(d_{s-4h}^n(kl,I)\geq M\right).
\end{align*}
Note that in the definition of $\bar{\Psi}_n^I$, the distance is always indexed by $s-4h$ regardless which value we plug in for $r$. In particular $\bar{\Psi}_n(r;\beta)=\bar{\Psi}_n^{\emptyset}(r,s;\beta)$ for all $s\in[0,T]$. These definitions ensure that
\begin{align*}
\tilde{\phi}_{n,ij}^I(s):=\frac{N}{mp_n(s)\sqrt{h}}\int_s^{s+2h}f_n(r,s)\alpha_0(r)F_n^I(r,s)dr
\end{align*}
is predictable with respect to $\tilde{F}_{ij,t}^{n,I,m}$. Note that the index $ij$ is not necessary, but we keep it here to reproduce the notation from Proposition 2.9 in \citet{K20}. We write moreover $\tilde{\phi}_{n,ij}:=\tilde{\phi}_{n,ij}^{\emptyset}$. In the following we show two bounds similar to Lemma \ref{lem:phi_bound}. We have for a suitable $C_1>0$ that
\begin{align}
&\left|\tilde{\phi}_{n,ij}^I(s)\right| \nonumber \\
\leq&\frac{N\|\alpha_0\|_{\infty}\|K\|_{\infty}\|w\|_{\infty}}{mp_n(s)\sqrt{h}\eta_0^2}\int_s^{s+2h}\frac{1}{mp_n(r)}\sum_{k,l\in V_n}\left|\Psi(X_{n,kl}(r);\beta_0)-\Psi(X_{n,kl}(r);\tilde{\beta}_n)\right|C_{n,kl}(r)dr \nonumber \\
\leq&\frac{2N\|\alpha_0\|_{\infty}\|K\|_{\infty}\|w\|_{\infty}L_{\Psi}\sqrt{h}}{mp_n(s)\eta_0^2}\|\tilde{\beta}_n-\beta_0\|\sup_{r\in[s,s+2h]}\frac{1}{mp_n(r)}\sum_{k,l\in V_n}C_{n,kl}(r) \nonumber \\
\leq&C_1\frac{N\sqrt{h}}{mp_n(s)}\|\tilde{\beta}_n-\beta_0\|\sup_{r\in[s,s+2h]}\frac{1}{mp_n(r)}\sum_{k,l\in V_n}C_{n,kl}(r).\label{eq:phi_tilde_bound1}
\end{align}
Similarly, after possibly increasing $C_1$,
\begin{align}
&\left|\tilde{\phi}_{n,ij}(s)-\tilde{\phi}_{n,ij}^I(s)\right| \nonumber \\
\leq&\frac{N\|\alpha_0\|_{\infty}\|K\|_{\infty}\|w\|_{\infty}}{mp_n(s)\sqrt{h}\eta_0^2}\int_s^{s+2h}\frac{1}{mp_n(r)}\sum_{k,l\in V_n}\left|\Psi(X_{n,kl}(r);\beta_0)-\Psi(X_{n,kl}(r);\tilde{\beta}_n)\right| \nonumber \\
&\quad\quad\quad\quad\quad\times\left|1-\Ind\left(d_{s-4h}^n(kl,I)\geq m\right)\right|\cdot C_{n,kl}(r)dr \nonumber \\
\leq&\frac{2N\|\alpha_0\|_{\infty}\|K\|_{\infty}\|w\|_{\infty}L_{\Psi}\sqrt{h}}{m^2p_n(s)p_n\eta_0^2}\|\tilde{\beta}_n-\beta_0\||I|B_n^I \nonumber \\
\leq&C_1\frac{N\sqrt{h}}{m^2p_n(s)p_n}\|\tilde{\beta}_n-\beta_0\||I|B_n^I. \label{eq:phi_tilde_bound2}
\end{align}
With these inequalities and since
$$\eqref{eq:r222}=\sum_{i,j\in V_n}\int_0^T\tilde{\phi}_{n,ij}(s)dM_{n,ij}(s)$$
we may apply Lemma 2.9 from \citet{K20} to show that $\eqref{eq:r222}=o_P(1)$. Thus, we need to show that the three terms in the upper bound from this proposition converge to zero. We begin with the first condition: Apply \eqref{eq:phi_tilde_bound1} to get
\begin{align*}
&\sum_{i,j\in V_n}\int_0^T\IE\left(\tilde{\phi}_{n,ij}^{ij}(s)^2C_{n,ij}(s)\alpha_0(s)\Psi(X_{n,ij}(s);\beta_0)\right)ds \\
\leq&C_1^2\|\alpha_0\|_{\infty}\|\Psi(\cdot;\beta_0)\|_{\infty}\sum_{i,j\in V_n}\int_0^T\frac{N^2h}{m^2p_n(s)^2} \\
&\quad\quad\quad\quad\times\IE\left(\|\tilde{\beta}_n-\beta_0\|^2\left(\sup_{r\in[s,s+2h]}\frac{1}{mp_n(r)}\sum_{k,l\in V_n}C_{n,kl}(r)\right)^2C_{n,ij}(s)\right)ds \\
\leq&\frac{C_1^2\|\alpha_0\|_{\infty}\|\Psi(\cdot;\beta_0)\|_{\infty}\cdot N^2h}{mp_n}\IE\left(\|\tilde{\beta}_n-\beta_0\|^2\right) \\
&\quad\quad\quad\quad\times\int_0^T\IE\left(\left(\sup_{r\in[s,s+2h]}\frac{1}{mp_n(r)}\sum_{k,l\in V_n}C_{n,kl}(r)\right)^3\right)ds.
\end{align*}
The above converges to zero because all norms are bounded by (C), $N/mp_n=O(1)$ by (SP), the integral with the third moment is bounded by \eqref{eq:wc1+1}, $h\to0$ by (KBW) and $\IE\left(\|\tilde{\beta}_n-\beta_0\|^2\right)=O(1/N)$ by (P).

We continue with the second term of the upper bound in Proposition 2.9 in \citet{K20} by using \eqref{eq:phi_tilde_bound1} and \eqref{eq:phi_tilde_bound2} in the same way (recall also that $B_n^{\{ij,kl\}}\leq B_n^{ij}B_n^{kl}$)
\begin{align*}
&\sum_{i,j,k,l\in V_n}\IE\left(\int_0^T\tilde{\phi}_{n,ij}^{ij,kl}(t)dM_{n,ij}(t)\int_0^T\left(\phi_{n,kl}(s)-\tilde{\phi}_{n,kl}^{ij,kl}(s)\right)dM_{n,kl}(s)\right) \\
\leq&\frac{2C_1^2N^2h}{mp_n}\sum_{i,j,k,l\in V_n}\IE\Bigg(\|\tilde{\beta}_n-\beta_0\|^2 \\
&\quad\quad\times\int_0^T\frac{B_n^{ij}}{mp_n(t)}\sup_{r\in[t,t+2h]}\frac{1}{mp_n(r)}\sum_{k',l'\in V_n}C_{n,k'l'}(r)d|M_{n,ij}|(t)\cdot\int_0^T\frac{B_n^{kl}}{mp_n(s)}d|M_{n,kl}|(s)\Bigg) \\
\leq&\frac{2C_1^2N^2h}{mp_n}\IE\left(\|\tilde{\beta}_n-\beta_0\|^2\right) \\
&\quad\quad\times\IE\left(\left(\frac{1}{mp_n}\sum_{i,j\in V_n}\int_0^T\sup_{r\in[t,t+2h]}\frac{B_n^{ij}}{mp_n(r)}\sum_{k',l'\in V_n}C_{n,k'l'}(r)d|M_{n,ij}|(t)\right)^2\right)^{\frac{1}{2}} \\
&\quad\quad\times\IE\left(\left(\sum_{k,l\in V_n}\int_0^T\frac{B_n^{kl}}{mp_n(s)}d|M_{n,kl}|(s)\right)^2\right)^{\frac{1}{2}}
\end{align*}
because $\tilde{\beta}_n$ is independent of everything else and the Cauchy-Schwarz Inequality. We have convergence of the above to zero because of the MSE property of $\tilde{\beta}_n$ from (P) and since $N/mp_n$ is bounded by assumption (SP), the integrals remain bounded by \eqref{eq:assump8} and Assumption \eqref{eq:wc1} and $h\to0$ by (KBW). And finally for the third assumption we use \eqref{eq:phi_tilde_bound2} again and $B_n^{\{ij,kl\}}\leq B_n^{ij}B_n^{kl}$ and estimate
\begin{align*}
&\sum_{ij,kl\in V_n}\IE\left(\int_0^T\left(\tilde{\phi}_{n,ij}(t)-\tilde{\phi}_{n,ij}^{ij,kl}(t)\right)dM_{n,ij}(t)\int_0^T\left(\phi_{n,kl}(s)-\tilde{\phi}_{n,kl}^{ij,kl}(s)\right)dM_{n,kl}(s)\right) \\
\leq&\frac{4C_1^2N^2h}{m^2p_n^2}\sum_{ij,kl\in V_n}\IE\left(\|\tilde{\beta}_n-\beta_0\|^2\int_0^T\frac{(B_n^{ij})^2}{mp_n(t)}d|M_{n,ij}|(t)\int_0^T\frac{(B_n^{kl})^2}{mp_n(s)}|dM_{n,kl}|(s)\right) \\
=&\frac{4C_1^2N^2h}{m^2p_n^2}\IE\left(\|\tilde{\beta}_n-\beta_0\|^2\right)\IE\left(\left(\sum_{ij\in V_n}\int_0^T\frac{(B_n^{ij})^2}{mp_n(t)}d|M_{n,ij}|(t)\right)^2\right).
\end{align*}
We conclude convergence of the above to zero from the boundedness assumptions in (SP), the MSE properties of $\tilde{\beta}_n$ in (P) and the property of the expectation in \eqref{eq:assump9}. Thus we have bounded all three expression from Proposition 2.9 in \citet{K20} and the proof is complete.
\end{proof}

\newpage
\bibliographystyle{chicago}
\bibliography{mybib}

\end{document}